\pdfoutput=1
\documentclass[11pt]{article}
\usepackage{authblk}

\usepackage{color}
\usepackage{helvet}         
\usepackage{courier}        
\usepackage{graphicx}        
\usepackage{amsmath}
\usepackage{amssymb}
\usepackage{bbm}
\usepackage{amsthm}
\usepackage{subcaption}
\usepackage{sidecap}
\usepackage{comment}
\usepackage[font=small]{caption}
\usepackage[inline]{enumitem}

\usepackage{yhmath}
\usepackage{thmtools}
\usepackage{thm-restate}
\usepackage{floatrow}
\usepackage{setspace}

\newtheorem{theorem}{Theorem}
\newtheorem{corollary}[theorem]{Corollary}
\newtheorem{lemma}[theorem]{Lemma}

\newtheorem{proposition}[theorem]{Proposition}

\newtheorem{remark}[theorem]{Remark}
\newtheorem{definition}[theorem]{Definition}
\usepackage[top=2cm, bottom=2cm, left=2cm, right=2cm]{geometry}

\numberwithin{theorem}{section}
\numberwithin{conjecture}{section}
\numberwithin{corollary}{section}
\numberwithin{lemma}{section}
\numberwithin{proposition}{section}
\numberwithin{remark}{section}
\numberwithin{definition}{section}

\numberwithin{figure}{section}
\numberwithin{equation}{section}

\begin{document}

\title{Crossing Probabilities of Multiple Ising Interfaces}

\vspace{5cm}

\begin{center}
\LARGE \bf Crossing Probabilities of Multiple Ising Interfaces
\end{center}

\vspace{0.25cm}

\begin{center}
{\bf Eveliina Peltola}\\
{\footnotesize{\texttt{eveliina.peltola@hcm.uni-bonn.de}}}\\
{\small{ \vspace*{1mm}
Department of Mathematics and Systems Analysis, Aalto University, Finland
\\ and \\ 
Institute for Applied Mathematics, University of Bonn, Germany}}

\bigskip

{\bf Hao Wu} \\
{\footnotesize{\texttt{hao.wu.proba@gmail.com}}}\\
{\small{\vspace*{1mm} Yau Mathematical Sciences Center, Tsinghua University, China}}
\end{center}

\vspace{0.75cm}

\begin{center}
\begin{minipage}{0.95\textwidth}
\abstract{
We prove that in the scaling limit,
the crossing probabilities of multiple interfaces in the critical planar Ising model with alternating boundary conditions
are conformally invariant expressions given by the pure partition functions of multiple $\textnormal{SLE}_\kappa$ with $\kappa=3$.
In particular, this identifies the scaling limits with ratios of specific correlation functions of conformal field theory.
}
\end{minipage}
\end{center}

\vspace{0.75cm}

\tableofcontents
\newpage

\global\long\def\ud{\mathrm{d}}
\global\long\def\dist{\mathrm{dist}}

\global\long\def\domainofdef{\mathfrak{U}}
\global\long\def\SLE{\textnormal{SLE}}
\global\long\def\CLE{\textnormal{CLE}}

\global\long\def\LP{\mathrm{LP}}
\global\long\def\DP{\mathrm{DP}}
\global\long\def\DPleq{\preceq} 
\global\long\def\DPgeq{\succeq} 
\newcommand{\wedgeat}[1]{\lozenge_#1} 
\newcommand{\upwedgeat}[1]{\wedge^#1}
\newcommand{\downwedgeat}[1]{\vee_#1}
\newcommand{\slopeat}[1]{\times_#1}
\newcommand{\removewedge}[1]{\setminus \wedgeat{#1}}
\newcommand{\removeupwedge}[1]{\setminus \upwedgeat{#1}}
\newcommand{\removedownwedge}[1]{\setminus \downwedgeat{#1}}
\newcommand{\wedgelift}[1]{\uparrow \wedgeat{#1}} 
\global\long\def\Mmat{\mathcal{M}}
\global\long\def\Minv{\mathcal{M}^{-1}}
\global\long\def\link#1#2{\{#1,#2\}}
\global\long\def\removeLink{/}
\global\long\def\nested{\boldsymbol{\underline{\Cap}}}
\global\long\def\unnested{\boldsymbol{\underline{\cap\cap}}}

\newcommand{\KWleq}{\stackrel{\scriptscriptstyle{()}}{\scriptstyle{\longleftarrow}}}
\newcommand{\CItilingsof}{\mathcal{C}}

\global\long\def\Rpos{\R_{> 0}}
\global\long\def\Znn{\Z_{\geq 0}}
\global\long\def\im#1{\operatorname{Im}(#1)}

\global\long\def\localSLE{\mathsf{P}}

\global\long\def\FKdual{\mathcal{L}}

\global\long\def\Schwartz_space{S}

\global\long\def\CCleft{D_j^L}
\global\long\def\CCright{D_j^R}

\global\long\def\gff{\Gamma}
\global\long\def\Height{\LH}

\global\long\def\graph{\mathcal{G}}

\newcommand{\nestedtilingof}{T_0}
\newcommand{\pathfromto}[2]{#1 \rightsquigarrow #2}

\global\long\def\SLE{\mathrm{SLE}}
\global\long\def\GFF{\mathrm{GFF}}
\global\long\def\LERW{\mathrm{LERW}}
\global\long\def\Ising{\mathrm{Ising}}
\global\long\def\free{\mathrm{free}}

\global\long\def\C{\mathbb{C}}
\global\long\def\E{\mathbb{E}}
\global\long\def\HH{\mathbb{H}}
\global\long\def\R{\mathbb{R}}
\global\long\def\PP{\mathbb{P}}
\global\long\def\U{\mathbb{U}}
\global\long\def\Z{\mathbb{Z}}
\global\long\def\one{\mathbf{1}}

\global\long\def\LA{\mathcal{A}}
\global\long\def\LB{\mathcal{B}}
\global\long\def\LE{\mathcal{E}}
\global\long\def\LF{\mathcal{F}}
\global\long\def\CobloF{\mathcal{U}}
\global\long\def\PartF{\mathcal{Z}}

\global\long\def\eps{\epsilon}
\global\long\def\cond{\;|\;}

\newcommand{\conn}{\vartheta}

\newcommand{\realpt}{\smash{\mathring{x}}}
\global\long\def\edge#1#2{\langle #1,#2 \rangle}

\newcommand{\swal}{\tau}

\allowdisplaybreaks

\section{Introduction}
\label{sec::intro}
\begin{figure}
\includegraphics[width=3cm, angle=45]{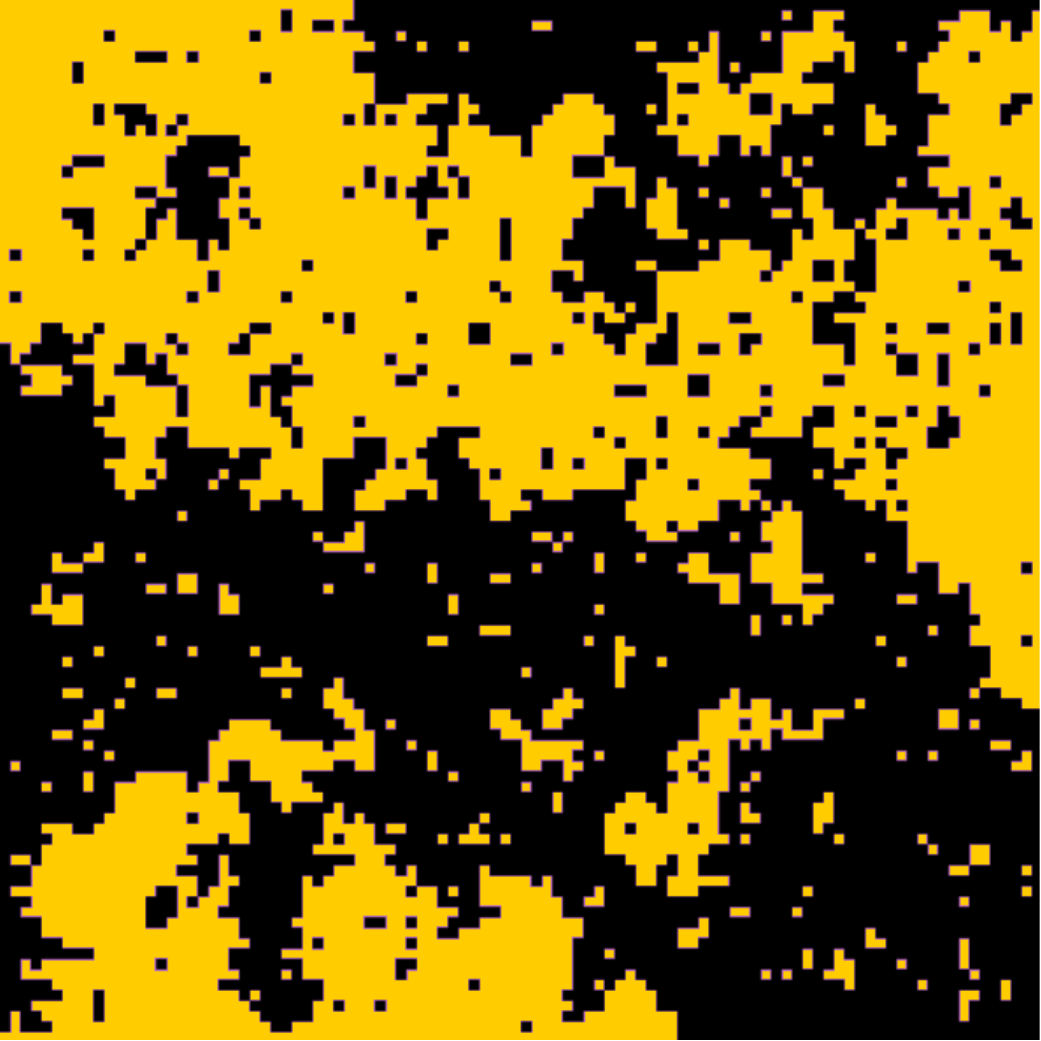} $\quad$
\includegraphics[width=3cm, angle=45]{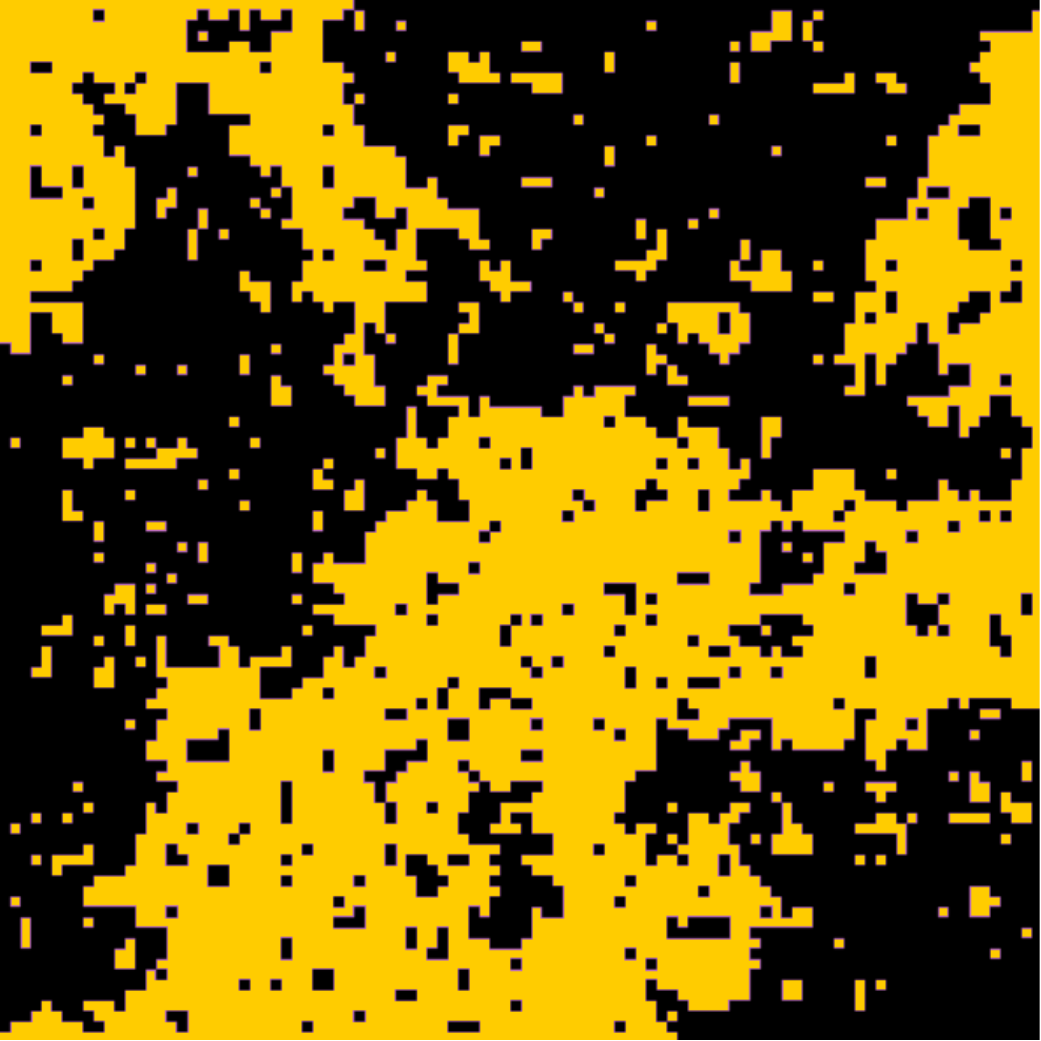} $\quad$
\includegraphics[width=3cm, angle=45]{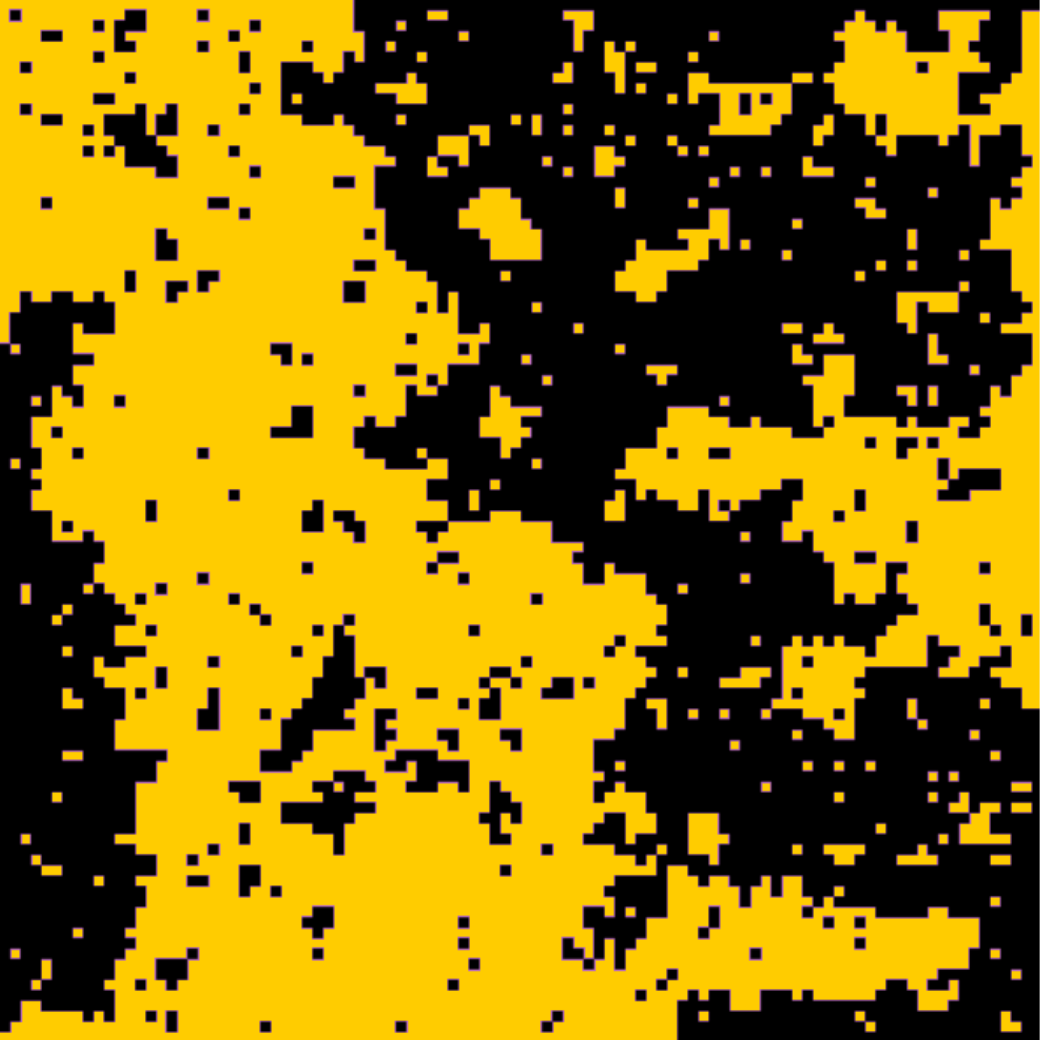} 
\\ \vspace{-1.5cm}
\includegraphics[width=3cm, angle=45]{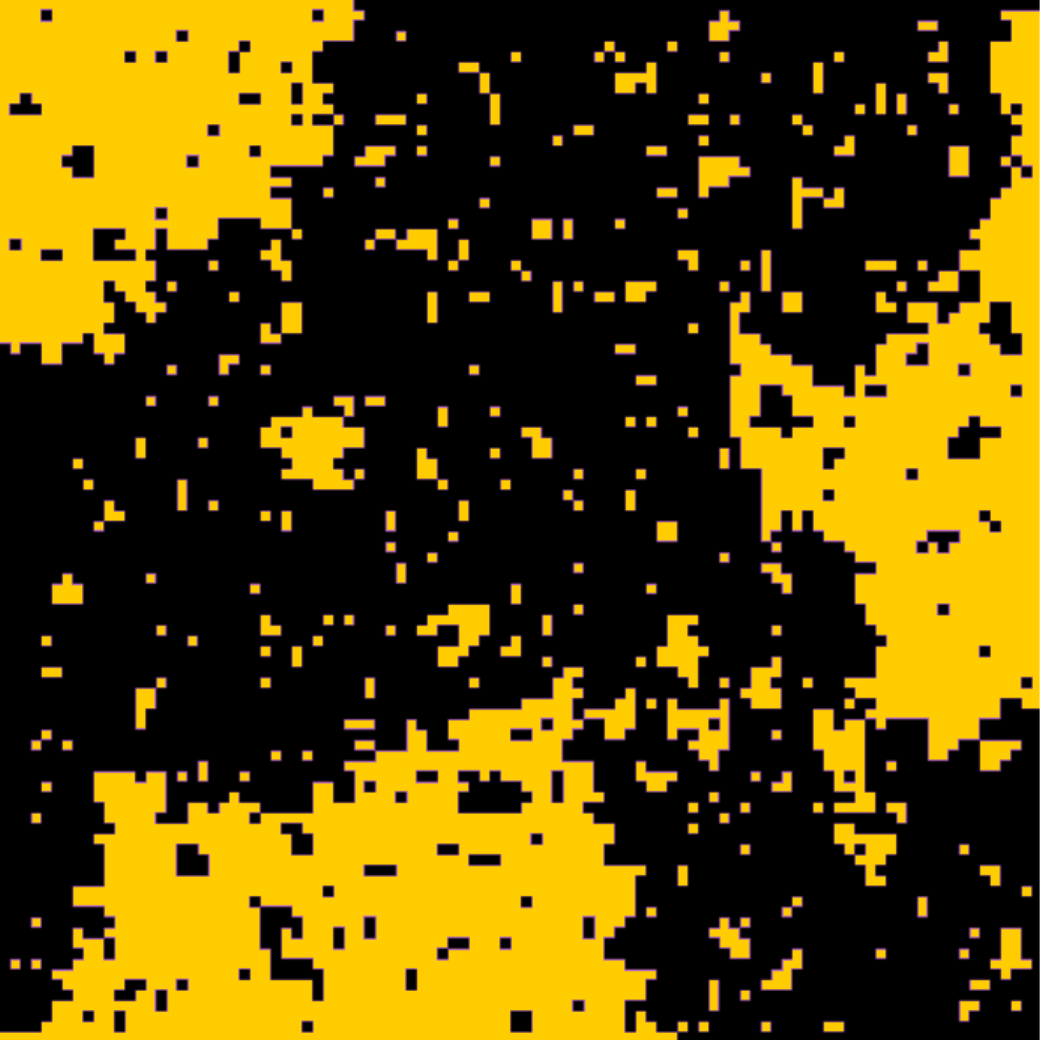} $\quad$
\includegraphics[width=3cm, angle=45]{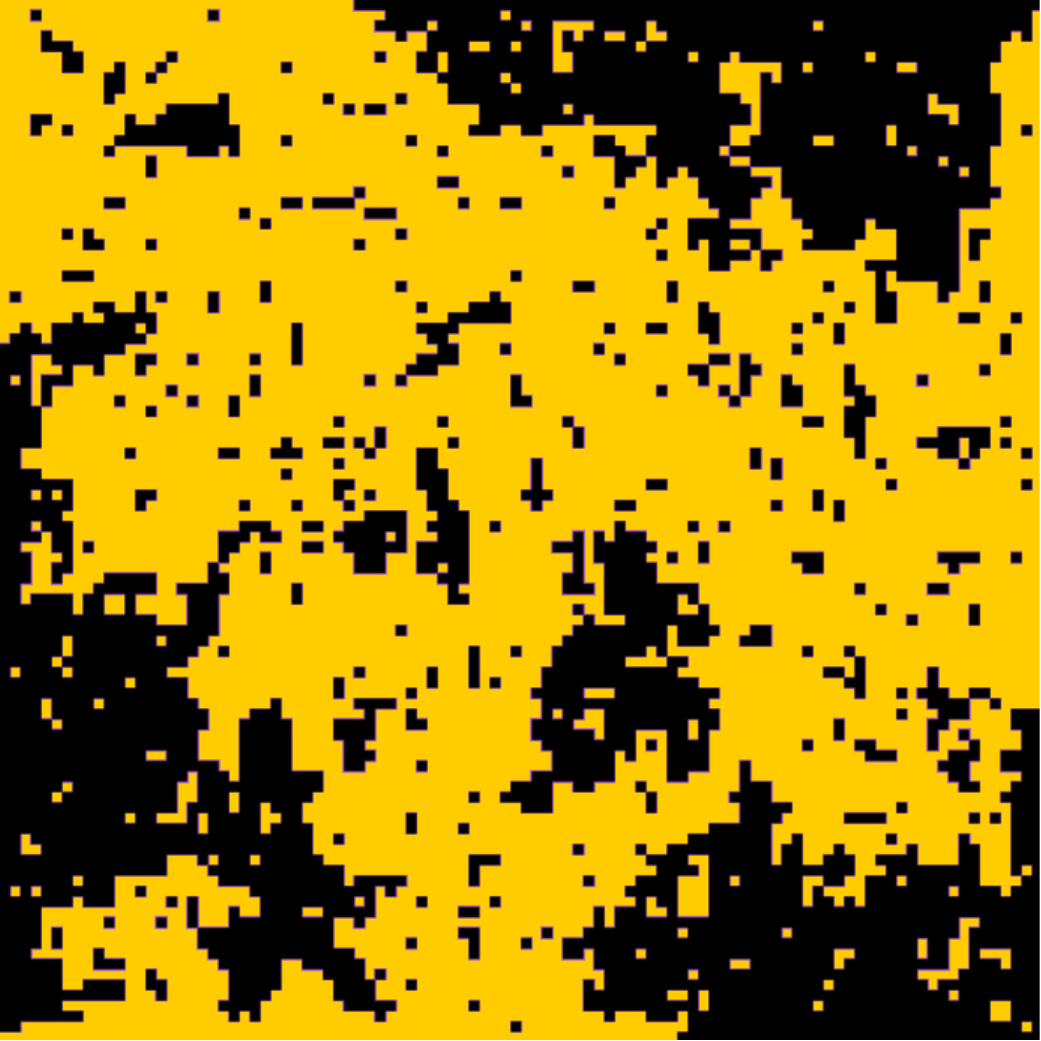} 
\caption{\label{fig::Ising}
Critical Ising model configurations on a $100\times100$ square with alternating boundary conditions: there are six marked boundary points at which the boundary condition changes from $\oplus$ (black) to $\ominus$ (yellow).  
In this case, there are five possible planar topological connectivities of the three chordal interfaces, labeled by $\alpha$ (see Figure~\ref{fig::patterns}). 
The total partition function has the symmetric form 
$\PartF_{\Ising} = \sum_\alpha \PartF_{\alpha}$, see Equation~\eqref{eq::TotalPartPos}, 
also featuring rotational symmetry of the boundary conditions under cyclic permutations of the marked boundary points (combined with a global spin flip $\oplus \leftrightarrow \ominus$ if necessary):
the crossing probabilities remain invariant under such a cyclic permutation.
}
\end{figure} 

The Ising model is arguably one of the most studied lattice models in statistical physics.
It was introduced in the 1920s by W. Lenz as a statistical model of 
two types of spins ($\oplus, \ominus$) defined on a lattice (e.g., $\Z^d$),
describing magnetic material.  
This model was further studied by Lenz's student, E.~Ising, who proved that there is no phase transition in dimension one,
leading him to conjecture that this is the case also in higher dimensions. However, as R.~Peierls showed in 1936, 
in two (and higher) dimensions an order-disorder phase transition in fact occurs at some critical temperature,
identified soon thereafter by a duality argument in 1941 by H.~Kramers and G.~Wannier (and later rigorously proven by L.~Onsager, 1944). 
During the next decades, renormalization group arguments and the introduction of conformal field theory (CFT)
suggested that, due to its (continuous, i.e., second-order) phase transition, at the critical temperature
the planar Ising model should enjoy \textit{conformal invariance} in the scaling 
limit as the lattice mesh tends to zero~\cite{BelavinPolyakovZamolodchikovConformalSymmetry, CardyRenorm}.
Ever since, there has been active research pertaining to understanding the planar Ising model at criticality, 
with recent success towards proving the conformal invariance via methods of discrete complex 
analysis~\cite{SmirnovConformalInvariance,SmirnovConformalInvarianceAnnals, ChelkakSmirnovIsing, 
ChelkakLzyurovSpinorIsing, HonglerKytolaIsingFree, HonglerSmirnovIsingEnergy, CDCHKSConvergenceIsingSLE, ChelkakHonglerLzyurovConformalInvarianceCorrelationIsing, ChelkakHonglerLzyurovConformalInvarianceCorrelationIsingGeneral}.

From the physics point of view, the scaling limit of the critical Ising model should be described by
some CFT.  
Not only such a statement lacks mathematical rigor, but also the limiting object(s) may not even be well-defined.
Interestingly, the scaling limit of the random field formed by the spin variables can be described 
--- not as a random function, but --- as a random distribution~\cite{CGN:Planar_Ising_magnetization_field1, ChelkakHonglerLzyurovConformalInvarianceCorrelationIsing}.
However, e.g.~the energy density only has a rigorous description in the scaling limit 
at the level of its correlation functions~\cite{HonglerSmirnovIsingEnergy},  
and it seems unclear, if even possible, to interpret it probabilistically as a random field.

Important geometric information about the model is encoded in its \textit{crossing probabilities}.
The goal of the present article is to identify the scaling limits of boundary-to-boundary 
crossing probabilities in the critical planar Ising model
as (ratios of) specific correlation functions of CFT
(namely those of a so-called degenerate field on the boundary with conformal weight $h_{1,2} = 1/2$ in a CFT of central charge $c = 1/2$). 
We emphasize that, although such a CFT field has not been mathematically defined, 
its \textit{correlation functions} can be understood as functions of several complex or real variables,
uniquely determined as solutions to certain PDE boundary value problems. 
The appropriate partial differential equations (given in Equation~\eqref{eq:PDEIsing})
are a special case of PDEs known in the physics literature as ``level-two BPZ'' or ``null-field''  equations~\cite{BelavinPolyakovZamolodchikovConformalSymmetry}. They are hypoelliptic equations with singularities on the diagonals (that is, when variables collide).
These singularities can be used to determine, in a sense, boundary behavior of the solutions of interest: by imposing specific asymptotic properties motivated by CFT fusion rules
--- see~\cite[Section~8]{KytolaMultipleSLE} for heuristic derivations ---
one can single out functions that describe the crossing probabilities in the scaling limit (see Theorem~\ref{thm::ising_crossingproba} for the precise statement). 
Interestingly, these correlation functions are also reminiscent of those of the free fermion field, or energy density, 
on the boundary, cf.~\cite{ID_book, DMS:CFT, KytolaMultipleSLE, HonglerThesis, HonglerSmirnovIsingEnergy, 
FSKZ-A_formula_for_crossing_probabilities_of_critical_systems_inside_polygons, PeltolaICMP}.

\bigskip

We shall next give the precise statements of our result. 
We consider the critical Ising model on a $\delta$-scaled square lattice approximation $\Omega^{\delta}$ 
of a simply connected subdomain $\Omega$ of the plane. 
One could, in principle, consider much more general graphs
(e.g.,~isoradial~\cite{CS-discrete_complex_analysis_on_isoradial, ChelkakSmirnovIsing}) 
and thus also address \textit{universality} for the Ising model: its scaling limit should not depend on the microscopic details of the approximation scheme. This would be, however, highly technical, so we stick to the simplest setup.

We impose \textit{alternating boundary conditions}: 
we divide the boundary into $2N$ segments, 
$N$ of which have spin $\oplus$
and $N$ spin $\ominus$, and the different segments alternate
--- see Equation~\eqref{eq::alternating} and Section~\ref{sec::Ising_crossing_proba} for details.
With such boundary conditions, $N$ random macroscopic \textit{interfaces} connect pairwise the $2N$ marked boundary points, as illustrated in Figure~\ref{fig::Ising}.
Supplementing the celebrated results of~\cite{CDCHKSConvergenceIsingSLE} in the case of $N=1$, it was proved 
in~\cite{IzyurovIsingMultiplyConnectedDomains} that ``locally'' (i.e., up to a stopping time for the growth of the curve), 
these interfaces converge (weakly in terms of probability measures on curves) in the scaling limit $\delta \to 0$ to 
a \textit{multiple Schramm-Loewner evolution} process (a version of $N$-$\SLE_\kappa$ with $\kappa = 3$). 

\begin{figure}
\includegraphics[scale=.9]{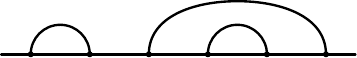} \quad
\includegraphics[scale=.9]{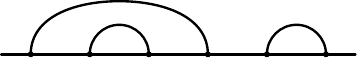} \quad
\includegraphics[scale=.9]{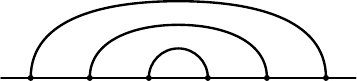} \\[2em]
\includegraphics[scale=.9]{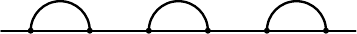} \quad
\includegraphics[scale=.9]{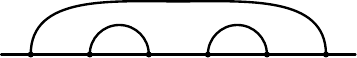} \\[1em]
\caption{\label{fig::patterns}
Illustration of the five possible planar topological connectivities $\alpha \in \LP_N$ of  the three chordal interfaces in Figure~\ref{fig::Ising} transformed to the upper half-plane (serving as a common reference domain by conformal invariance). 
Here, the first marked point on the left corresponds to the lower corner of the rotated square in Figure~\ref{fig::Ising},
and we follow the marked points counterclockwise. We call the labels $\alpha \in \LP_N$ ``link patterns''.
}
\end{figure}

The $N$ interfaces can form a Catalan number $\frac{1}{N+1} \binom{2N}{N}$ of possible planar connectivities, 
as illustrated in Figure~\ref{fig::Ising}. 
We label them by \textit{(planar) link patterns}
$\alpha  =  \{  \link{a_1}{b_1}, \ldots,  \link{a_N}{b_N} \}$,  
or in other words, planar pair partitions of the set
$\{1,2,\ldots,2N\} = \{a_1, b_1,\ldots,  a_N, b_N  \}$.
We denote by $\LP_N$ the set of such link patterns comprising $N$ \textit{links} (i.e., pairs) $\link{a_s}{b_s}$. 
In this article, we are interested in the \textit{crossing probabilities}
of these interfaces as functions of the domain and the marked boundary points. Our main result, Theorem~\ref{thm::ising_crossingproba}, identifies the scaling limits of these crossing probabilities in terms of 
so-called partition functions of multiple $\SLE_3$ curves. In particular, the limit crossing probabilities are conformally invariant and match with their CFT predictions --- see~\cite[Section~8.3]{KytolaMultipleSLE} for a detailed discussion.

\bigskip

We call $(\Omega; x_1, \ldots, x_{2N})$ a \textit{(topological) polygon} if $\Omega \subsetneq \C$ is a simply connected domain such that $\partial\Omega$ is locally connected and $x_1, \ldots, x_{2N}\in \partial \Omega$ are $2N$ distinct boundary points lying on $\partial\Omega$ in counterclockwise order. 
Suppose also that $\Omega$ is bounded. A sequence $( (\Omega^{\delta}; x_1^{\delta}, \ldots,  x_{2N}^{\delta}) )_{\delta > 0}$ 
of discrete polygons, defined precisely in Section~\ref{sec::Ising_crossing_proba},
is said to \textit{converge} to $(\Omega; x_1, \ldots, x_{2N})$ as $\delta \to 0$ \textit{in  the Carath\'{e}odory  sense} 
if  there exist  conformal maps $\varphi_{\delta}$  (resp.~$\varphi$)  from 
the upper-half plane $\HH = \{z\in\C \colon \im{z}>0\}$ to $\Omega^{\delta}$ (resp.~from  $\HH$ to $\Omega$) 
such that $\varphi_{\delta}\to \varphi$ uniformly on any compact subset of $\HH$, 
and $\smash{\underset{\delta   \to   0}{\lim} \, \varphi_\delta^{-1}(x_j^{\delta}) = \varphi^{-1}(x_j)}$, for  all $j$.
The points $x_1^{\delta}, \ldots,  x_{2N}^{\delta}$  indicate $\oplus / \ominus$ boundary condition changes.

\smallbreak

\noindent\textbf{Update.}
In addition, as pointed out recently by 
A.~Karrila~\cite{KarrilaConformalImage}, for the convergence near the marked points we need a slightly stronger notion, \textit{close-Carath\'{e}odory convergence}, discussed in Section~\ref{sub::Ising_notation}. 
Namely, the Carath\'{e}odory convergence allows wild behavior of the boundary approximations, 
while in order to obtain precompactness of the random interfaces as $\delta \to 0$, a slightly stronger convergence which guarantees good approximations around the marked boundary points is required.

\begin{restatable}{theorem}{IsingTHM}
\label{thm::ising_crossingproba}
Let $(\Omega; x_1, \ldots, x_{2N})$ be a bounded polygon with $N \geq 1$, and suppose that discrete polygons
$(\Omega^{\delta}; x_1^{\delta}, \ldots, x_{2N}^{\delta})$ on $\delta\Z^2$ 
converge to $(\Omega; x_1, \ldots, x_{2N})$ as $\delta\to 0$ in the close-Carath\'{e}odory sense.
Consider the critical Ising model on $\Omega^{\delta}$ with alternating boundary conditions.  
Denote by $\conn^{\delta}$ the random connectivity pattern in $\LP_N$
formed by the $N$ discrete interfaces with law $\PP^{\delta}$.
Then, we have
\begin{align} \label{eqn::ising_crossing_proba}
\lim_{\delta\to 0} \PP^{\delta}[\conn^{\delta}=\alpha] 
= \frac{\PartF_{\alpha}(\Omega;x_{1},\ldots,x_{2N})}{\PartF^{(N)}_{\Ising}(\Omega;x_{1},\ldots,x_{2N})} , \quad \text{for all }\alpha\in\LP_N,
\quad \text{where} \quad \PartF^{(N)}_{\Ising} := \sum_{\alpha\in\LP_N}\PartF_{\alpha} ,
\end{align}
and $\{\PartF_{\alpha} \colon \alpha \in \LP_N\}$
is the collection of functions uniquely determined as the solution to the PDE boundary value problem
given in Definition~\ref{defn:PPF}, 
also known as pure partition functions of multiple $\SLE_\kappa$ with $\kappa = 3$. 
\end{restatable}

In the first nontrivial case of $N=2$, the crossing formula~\eqref{eqn::ising_crossing_proba} in Theorem~\ref{thm::ising_crossingproba} was predicted by
L.-P.~Arguin and Y.~Saint-Aubin~\cite{ASA02}: 
the pure partition functions in this case are given by
\begin{align*}
\PartF_{\vcenter{\hbox{\includegraphics[scale=0.2]{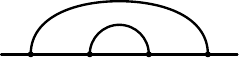}}}}(\HH;x_1,x_2,x_3,x_4)
= \; & \frac{2 \Gamma(4/3)}{\Gamma(8/3) \Gamma(5/3)} \;
\left( \frac{(x_2-x_1)(x_4-x_3)}{(x_4-x_2)(x_3-x_1)} \right)^{2/3} \, 
\frac{{}_2F_1\left(\frac{4}{3}, -\frac{1}{3}, \frac{8}{3}; \frac{(x_2-x_1)(x_4-x_3)}{(x_4-x_2)(x_3-x_1)} \right)}{(x_4-x_1)(x_3-x_2)} \\ 
\PartF_{\vcenter{\hbox{\includegraphics[scale=0.2]{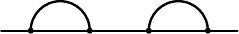}}}}(\HH;x_1,x_2,x_3,x_4)
= \; & \frac{2 \Gamma(4/3)}{\Gamma(8/3) \Gamma(5/3)} \;
\left( \frac{(x_4-x_1)(x_3-x_2)}{(x_4-x_2)(x_3-x_1)} \right)^{2/3} \, 
\frac{{}_2F_1\left(\frac{4}{3}, -\frac{1}{3}, \frac{8}{3}; \frac{(x_4-x_1)(x_3-x_2)}{(x_4-x_2)(x_3-x_1)} \right)}{(x_2-x_1)(x_4-x_3)} ,
\end{align*}
for $x_1<x_2<x_3<x_4$, where
we use the upper-half plane $\HH = \{z\in\C \colon \im{z}>0\}$ as a reference domain (cf.~\eqref{eq:conformal_image}), 
and denote 
$\vcenter{\hbox{\includegraphics[scale=0.2]{figures_arXiv2022/link-1.pdf}}}=\{\link{1}{2}, \link{3}{4}\}$ and $\vcenter{\hbox{\includegraphics[scale=0.2]{figures_arXiv2022/link-2.pdf}}}=\{\link{1}{4}, \link{2}{3}\}$ for the two possible link patterns. 
These formulas and certain other special cases appear in~\cite{KytolaMultipleSLE, IzyurovObservableFree}.
In general, explicit formulas for the probability amplitudes $\PartF_{\alpha}$ are not known, 
while the total partition function 
$\PartF_{\Ising}$ does have an explicit Pfaffian formula, already well-known in the physics literature
(as a correlation function of the free fermion field, or energy density)
and appearing, e.g., 
in~\cite{KytolaPeltolaPurePartitionSLE, IzyurovIsingMultiplyConnectedDomains, PeltolaWuGlobalMultipleSLEs} 
in the $\SLE$ context. 
The functions $\PartF_{\alpha}$ are defined as appropriately chosen solutions to the PDE system~\eqref{eq:PDEIsing} stated below,
which appears in the CFT literature as a ``null-field'', or ``BPZ'' equation~\cite{BelavinPolyakovZamolodchikovConformalSymmetry, ID_book, DMS:CFT, KytolaMultipleSLE}. 
From the $\SLE$ point of view, the functions $\PartF_{\alpha}$ are total masses of ``pure'' multiple SLE 
measures, see~\cite{KytolaPeltolaPurePartitionSLE, PeltolaWuGlobalMultipleSLEs}.

Knowing that the Ising interfaces converge (at least in a local sense, cf.~Section~\ref{subsec::ising_Ztotal_global} and~\cite{IzyurovIsingMultiplyConnectedDomains}) to $\SLE_3$ type curves,
the gist of our proof for Theorem~\ref{thm::ising_crossingproba} is a relatively standard $\SLE$ martingale argument. 
Indeed, due to the PDEs~\eqref{eq:PDEIsing}, the ratio $\PartF_{\alpha} / \PartF_{\Ising}$ gives rise to a martingale with respect to exploring one of the chordal Ising interfaces (see Equation~\eqref{eq: Ising martingale}). 
By investigating the terminal value of this martingale, an induction argument proves~\eqref{eqn::ising_crossing_proba}. 
The subtleties in the proof are related to the analysis of the martingale as well as of the $\SLE_3$ type curve as it approaches one of the marked boundary points.

\begin{corollary} \label{cor::CI}
For any $N \geq 1$, $\alpha \in \LP_N$, and bounded polygon $(\Omega; x_1, \ldots, x_{2N})$, the crossing probability 
\begin{align*}
P_\alpha(\Omega;x_{1},\ldots,x_{2N}) := \lim_{\delta\to 0}\PP^{\delta}[\conn^{\delta}=\alpha] 
\end{align*}
has the following properties:
\begin{itemize}
\item 
{\bf Conformal invariance:} 
For any conformal map $\varphi \colon \Omega \to \varphi(\Omega)$, we have  
\begin{align*}
P_\alpha(\varphi(\Omega); \varphi(x_1), \ldots, \varphi(x_{2N})) = P_\alpha(\Omega; x_1, \ldots, x_{2N}) .
\end{align*}

\item 
{\bf Asymptotics:} We have 
\begin{align*}
\lim_{x_j , x_{j+1} \to \xi} 
P_\alpha(\Omega;x_1 , \ldots , x_{2N})
= \begin{cases}
0 , \quad &
    \text{if } \link{j}{j+1} \notin \alpha , \\
P_{\hat{\alpha}}(\Omega;x_{1},\ldots,x_{j-1},x_{j+2},\ldots,x_{2N}) , &
    \text{if } \link{j}{j+1} \in \alpha ,
\end{cases} 
\end{align*}
for all $N \geq 1$, $\alpha \in \LP_N$, $j \in \{1,2, \ldots, 2N-1 \}$, and $\xi \in (x_{j-1}, x_{j+2})$, 
with initial condition $P_\emptyset = 1$, and where $\hat{\alpha} = \alpha \removeLink \link{j}{j+1} \in \LP_{N-1}$ denotes
the link pattern obtained from $\alpha$ by removing the link $\link{j}{j+1}$ and relabeling 
the remaining indices by the first $2(N-1)$ positive integers. 

\item {\bf PDEs:}
The probability amplitude $\PartF_\alpha$ in~\eqref{eqn::ising_crossing_proba}
satisfies a system of $2N$ partial differential equations, 
which in the upper-half plane $\HH=\{z\in\C \colon \im{z}>0\}$ are given by
\begin{align} \label{eq:PDEIsing}
\hspace*{-5mm}
\left[ \frac{3}{2} \partial_j + \sum_{i \neq j}\left(\frac{2}{x_{i}-x_{j}} \partial_i - 
\frac{1}{(x_{i}-x_{j})^{2}}\right) \right] 
\PartF_\alpha(\HH;x_1,\ldots,x_{2N}) =  0 , \quad \text{for all } j \in \{1,\ldots,2N\} .
\end{align}
\end{itemize}
\end{corollary}

Let us also remark that our results imply that the scaling limit of the Ising interfaces depicted in Figure~\ref{fig::Ising} 
is the ``global'' multiple $\SLE_3$ process whose law is given by
the sum of the extremal multiple $\SLE_3$ probability measures associated 
to the various possible connectivity patterns $\alpha$ of the interfaces 
(cf.~\cite{IzyurovIsingMultiplyConnectedDomains, BeffaraPeltolaWuUniqueness}). 
To avoid introducing more definitions in the present article, 
for this we only refer the reader to the literature for details, discussed,
e.g., in recent works by A.~Karrila~\cite{Karrila_multiple_SLE, KarrilaUSTBranches}.

\bigskip

Analogues of Theorem~\ref{thm::ising_crossingproba} should also hold for other critical planar statistical mechanics models 
(with other $\kappa > 0$).  
In the appendices, 
we discuss the following known examples, 
whose boundary conditions are symmetric under cyclic permutations of the marked boundary points (i.e., rotationally symmetric):
\begin{itemize}[leftmargin=*]
\item Gaussian free field, whose level lines are $\SLE_\kappa$ type curves with $\kappa = 4$
(see~\cite{SchrammSheffieldContinuumGFF});

\item chordal loop-erased random walks, which converge in the scaling limit to $\SLE_\kappa$ type curves with $\kappa=2$ 
(see~\cite{SchrammScalinglimitsLERWUST, LawlerSchrammWernerLERWUST, Zhan-scaling_limits_of_planar_LERW, Karrila_multiple_SLE} for various setups). 
\end{itemize}

Both of these examples are exactly solvable: 
explicit formulas for connection probabilities for loop-erased random walks and level lines of the Gaussian free field
(as well as for crossing probabilities in the double-dimer model), 
were found by R.~Kenyon and D.~Wilson in~\cite{Kenyon-Wilson:Boundary_partitions_in_trees_and_dimers} and further related to $\SLE$s 
in~\cite{PeltolaWuGlobalMultipleSLEs, KarrilaKytolaPeltolaCorrelationsLERWUST}.
Other lattice models seem ``less exactly solvable'' --- 
formulas for discrete crossing probabilities have not been found, and in the scaling limit only certain very special cases are known. 
In critical percolation, the case of $N=2$ is given by Cardy's formula~\cite{CardyPercolation, SmirnovPercolationConformalInvariance}
and the case of $N=3$ was solved by J.~Dub\'edat~\cite[Section~4.4]{Dubedat:Euler_integrals_for_commuting_SLEs}. 
For general $N$, crossing probabilities in critical percolation 
have not been exactly solved 
even in the scaling limit, 
but nevertheless, they do admit a characterization by multiple $\SLE_\kappa$ partition functions with $\kappa = 6$, 
(cf.~\cite{FSKZ-A_formula_for_crossing_probabilities_of_critical_systems_inside_polygons, lpwperco2022}). 
Recently in~\cite{lpwUST2021}, explicit scaling limits of crossing formulas for Peano curves tracking frontiers of uniform spanning trees were found. 
These curves are described by $\SLE_\kappa$ type curves with $\kappa=8$ (dual to loop-erased random walks).

In the random-cluster representation of the Ising model (i.e., FK-Ising model), 
each connection probability describes a natural percolation event, but the alternating boundary conditions in this model are not rotationally symmetric: 
flipping each wired boundary arc to a free boundary arc and vice versa changes the connection probabilities in the model --- this is due to the fact that such a flip is not a global symmetry of the FK-Ising model, but rather a duality. 
This in particular implies that the total partition function $\PartF_{\mathrm{FK}}$ expanded as a linear combination of the pure partition functions $\PartF_\alpha$ with $\kappa=16/3$ does not have the form of Equation~\eqref{eq::TotalPartPos},  but rather 
$\PartF_{\mathrm{FK}} = \sum_\alpha c_\alpha \PartF_\alpha$, where $c_\alpha>0$ are non-trivial coefficients. 
These coefficients can be solved explicitly: they are given by entries of the so-called meander matrix~\cite{FPW22} 
(see also~\cite{FSKZ-A_formula_for_crossing_probabilities_of_critical_systems_inside_polygons}).
For example, with four marked points of cross-ratio 
\begin{align*}
\chi = \frac{(x_4-x_1)(x_3-x_2)}{(x_3-x_1)(x_4-x_2)} ,
\end{align*}
we have\footnote{In~\cite{FPW22}, the total partition function for alternating boundary conditions is denoted $\LF_{\{\link{1}{2}, \link{3}{4}, \ldots, \link{2N-1}{2N}\}}$.}
\begin{align*}
\PartF_{\mathrm{FK}} (\HH; x_1, x_2, x_3, x_4)
= \; & \sqrt{2} \, (x_2-x_1)^{-1/8} (x_4-x_3)^{-1/8} ( \chi^{1/4} + 
\chi^{-1/4} )^{1/2} \\
= \; & 2 \, \PartF_{\vcenter{\hbox{\includegraphics[scale=0.2]{figures_arXiv2022/link-1.pdf}}}} (\HH; x_1, x_2, x_3, x_4) + \sqrt{2} \, \PartF_{\vcenter{\hbox{\includegraphics[scale=0.2]{figures_arXiv2022/link-2.pdf}}}} (\HH; x_1, x_2, x_3, x_4) .
\end{align*}
For the FK-Ising model, 
K.~Izyurov showed in~\cite{IzyurovObservableFree} that at criticality, probabilities of certain unions of connection events have conformally invariant scaling limits, expressed by quadratic irrational functions.
He later clarified and improved these results in~\cite{IzyurovMultipleFKIsing} while still being unable to find explicit expressions for general connection probabilities.
(In fact, predictions did appear in the physics literature~\cite{FSKZ-A_formula_for_crossing_probabilities_of_critical_systems_inside_polygons}.)
The general case is completely solved in the recent work~\cite{FPW22}, 
where also an analogue of Theorem~\ref{thm::ising_crossingproba} is proven for the critical FK-Ising model. 
Also more general boundary conditions are treated in~\cite{FPW22}.

\bigskip

\noindent\textbf{Outline.} 
The article is organized as follows. Section~\ref{sec::pre} is devoted to preliminaries:
we briefly discuss $\SLE$s, their basic properties, and
define the multiple $\SLE$ (pure) partition functions.
In the next Section~\ref{sec::Ising_partition_function}, we focus on the case of $\kappa = 3$ and prove crucial results concerning 
the multiple $\SLE_3$ partition function 
$\PartF_{\Ising}$.
Section~\ref{sec::Loewner_chain} consists of the analysis of the Loewner chain associated to this partition function,
leading to Theorem~\ref{thm::loewner_Ztotal_continuity}: the Loewner chain is indeed 
generated by a transient curve.  
This is one of the main difficulties in the proof, and relies on fine properties of the partition function 
$\PartF_{\Ising}$ from Section~\ref{sec::Ising_partition_function}.

In Section~\ref{sec::Ising_crossing_proba}, we first briefly discuss the Ising model and some existing results on the convergence of Ising interfaces,
and then prove the main result of this article, Theorem~\ref{thm::ising_crossingproba}.
For the proof, we need the following inputs.
First, we use the convergence of multiple Ising interfaces in a local sense from~\cite{IzyurovIsingMultiplyConnectedDomains}.
Second, we need the continuity of the scaling limit curves up to and including the swallowing time of the marked points,
which follows by standard Russo-Seymour-Welsh estimates 
(cf.~\cite{ChelkakDuminilHonglerCrossingprobaFKIsing})
and Aizenman-Burchard \& Kemppainen-Smirnov theory~\cite{AizenmanBurchardHolderRegularity, KemppainenSmirnovRandomCurves}. 
Third, a crucial technical ingredient to the proof 
is the continuity of the Loewner chain associated to the multiple $\SLE_3$ partition function 
$\PartF_{\Ising}$ up to and including the 
swallowing time of the marked points, that we establish in Theorem~\ref{thm::loewner_Ztotal_continuity}.
To finish the proof of Theorem~\ref{thm::ising_crossingproba}, we combine all of these inputs in Section~\ref{ProofSec} with detailed 
analysis of the martingale given by the ratio $\PartF_\alpha / \PartF_{\Ising}$, 
relying on properties of these partition functions from Section~\ref{sec::Ising_partition_function}.

Lastly, Appendices~\ref{app::appendix_gff} and~\ref{app::appendix_lerw} 
briefly summarize results similar to Theorem~\ref{thm::ising_crossingproba} for Gaussian free field and loop-erased random walks, respectively. 
The results from Appendix~\ref{app::appendix_gff} are needed to prove Theorem~\ref{thm::ising_crossingproba},
while Appendix~\ref{app::appendix_lerw} only serves as another example case of a rotationally symmetric model where analogous results hold (for proofs, see~\cite{KarrilaUSTBranches, KarrilaKytolaPeltolaCorrelationsLERWUST}).

\bigskip

\noindent\textbf{Acknowledgments.}
We thank K.~Izyurov for pointing out the useful identity~\eqref{eq::Hafnian} and discussing the limitations of his work.
We have also enjoyed discussions with V.~Beffara, N.~Berestycki, A.~Karrila, and K.~Kyt\"ol\"a on this topic. 
We wish to thank anonymous referees for their comments, which greatly helped to improve the presentation of this article.

H.W. is supported by the Beijing Natural Science Foundation (JQ20001).
E.P. is supported by the Academy of Finland grant number 340461 ``Conformal invariance in planar random geometry'', by the Academy of Finland Centre of Excellence Programme grant number 346315 ``Finnish centre of excellence in Randomness and STructures (FiRST)'', 
and by the Deutsche Forschungsgemeinschaft (DFG, German Research Foundation) under Germany's Excellence Strategy EXC-2047/1-390685813, as well as the DFG collaborative research centre ``The mathematics of emerging effects'' CRC-1060/211504053.
During this work, E.P. was 
affiliated at the University of Geneva 
and supported by the ERC AG COMPASP, the NCCR SwissMAP, and the Swiss~NSF.
A part of this work was completed during E.P.'s visit at the Institut Mittag-Leffler,
and during the authors' visit to the Oberwolfach Research Institute for Mathematics (MFO), 
which we cordially thank for hospitality.
The first version of this paper was finished during and after the workshop 
``Random Conformal Geometry and Related Fields'' at Seoul, and we kindly thank the organizers for the inspiring meeting.

\section{Partition Functions of Multiple SLEs}
\label{sec::pre}
In this section, we briefly discuss Schramm-Loewner evolutions ($\SLE$) and their partition functions.
For more background, the reader may consult~\cite{SchrammScalinglimitsLERWUST, LawlerConformallyInvariantProcesses, RohdeSchrammSLEBasicProperty, LawlerPartitionFunctionsSLE}, for instance.

\smallbreak

Recall that by a \textit{polygon} $(\Omega; x_1, \ldots, x_{2N})$  we refer to a simply connected domain $\Omega \subsetneq \C$ such that $\partial\Omega$ is locally connected and $x_1, \ldots, x_{2N}\in \partial \Omega$ are distinct boundary points in counterclockwise order along $\partial\Omega$. 
When $N=1$, we also call $(\Omega; x_1, x_2)$ a \textit{Dobrushin domain}. 
We say that $U \subset \Omega$ is a \textit{sub-polygon of $\Omega$} if $U$ is simply connected 
and $U$ and $\Omega$ agree in neighborhoods of $x_1, \ldots, x_{2N}$,
and in the case of $N=1$, we also call $(U; x_1, x_2)$ a \textit{Dobrushin subdomain}.
Finally, we say that a polygon $(\Omega; x_1, \ldots, x_{2N})$ is \textit{nice} if 
its boundary $\partial\Omega$ is $C^{1+\eps}$-regular for some $\eps > 0$ in neighborhoods of $x_1, \ldots, x_{2N}$. 
An example of a nice polygon is the upper half-plane
$\HH = \{z\in\C \colon \im{z}>0\}$ with any boundary points $x_1 < \cdots < x_{2N}$. 

\smallbreak

Planar curves are continuous mappings from $[0,1]$ to $\C$ modulo reparameterization. 
For a simply connected domain $\Omega \subsetneq \C$, we will consider curves in $\overline{\Omega}$.
For definiteness, we map $\Omega$ onto the unit disc 
$\U = \{ z \in \C \colon |z| < 1 \}$: for this we shall fix any conformal map $\Phi$ from $\Omega$ onto $\U$. 
Then, we endow the curves with the metric
\begin{align} \label{def::curves_metric}
\dist(\gamma_1, \gamma_2) 
:= \inf_{\psi_1, \psi_2} \sup_{t\in [0,1]} |\Phi(\gamma_1(\psi_1(t)))-\Phi(\gamma_2(\psi_2(t)))| ,
\end{align}
where the infimum is taken over all increasing homeomorphisms $\psi_1, \psi_2 \colon [0,1]\to[0,1]$.
The space of continuous curves on $\C$ modulo reparameterizations then becomes a complete separable metric space.
While the metric~\eqref{def::curves_metric} depends on the choice of the conformal map $\Phi$, the induced topology does not depend on this choice. This topology is important in Section~\ref{subsec::ising_Ztotal_global} for convergence of discrete interfaces.

\subsection{Schramm-Loewner Evolutions}
\label{subsec::pre_SLE}

For $\kappa \geq 0$, the (chordal) \textit{Schramm-Loewner evolution}, $\SLE_{\kappa}$, can be thought of as
a family of probability measures $\mathsf{P}(\Omega; x,y)$ on curves, indexed by Dobrushin domains $(\Omega; x, y)$.
Each measure $\mathsf{P}(\Omega; x,y)$ is supported on continuous unparameterized curves in $\overline{\Omega}$ from $x$ to $y$. 
It will be convenient to choose a parameterized representative 
$\gamma \colon [0,\infty) \to \C$ for each such curve, requiring that 
$\gamma(0) = x$ and $\gamma(t) \to y$ as $t\to\infty$. 
By re-scaling time, we regard $\gamma$ as a representative of an element in the above curve space.

Explicitly, the $\SLE_{\kappa}$ curves can be generated using random Loewner evolutions.
Consider a family of maps $(g_{t}, t\ge 0)$ obtained by solving the Loewner differential equation: for each $z \in \HH$,
\begin{align*}
\partial_{t}{g}_{t}(z)=\frac{2}{g_{t}(z)-W_{t}} 
\qquad \text{and} \qquad 
g_{0}(z)=z ,
\end{align*}
where $(W_t, t\ge 0)$ is a real-valued continuous \textit{driving function}. 
By general ODE theory, for each $z \in \HH$ this initial value problem has a unique solution $(g_{t}(z), 0 \leq t < \swal_z)$ with maximal lifetime
\begin{align*} 
\swal_z := \sup \big\{ t > 0 \colon \inf_{ s \in[0,t]} |g_s(z)-W_s| > 0 \big\} ,
\end{align*} 
the \textit{swallowing time} of $z$. 
Set $K_{t}:=\overline{\{z\in\HH \colon \swal_z \le t\}}$.
Then, $g_{t} \colon \HH\setminus K_t \to \HH$ is the unique conformal map (biholomorphic function) normalized in such a way that $|g_t(z)-z| \to 0$ as $z\to\infty$.
We call the growing sets $(K_{t}, t\ge 0)$ associated with these maps  a \textit{Loewner chain}.
Note that $g_{t}$ is also well-defined on $\R \setminus K_t$, and thus, 
the swallowing time $\swal_z$ can be defined for each $z \in \overline{\HH}$.

Now, the chordal $\SLE_{\kappa}$ in $(\HH;0,\infty)$ is first defined as the random Loewner chain $(K_{t}, t\ge 0)$ driven by 
$W_t=\sqrt{\kappa}B_t$, where $(B_t, t\ge 0)$ is the standard Brownian motion.
This Loewner chain~\cite{RohdeSchrammSLEBasicProperty} is almost surely generated by a continuous transient curve 
$\eta \colon [0,\infty) \to \overline{\HH}$ in the sense that 
for each $t$, the domain $\HH \setminus K_t$ is the unbounded component of $\HH \setminus \eta[0,t]$,
and we have $\eta(0) = 0$ and $|\eta(t)| \to \infty$ as $t\to\infty$. 
We regard $\eta$ as a random curve in $(\HH;0,\infty)$ 
and denote its probability measure by $\mathsf{P}(\HH; 0, \infty)$.
It is known~\cite{RohdeSchrammSLEBasicProperty} that when $\kappa \in (0,4]$ (for instance, $\kappa=3$ for the Ising interfaces), almost surely, 
the $\SLE_{\kappa}$ curve $\eta$ is simple and $K_t = \eta[0,t]$ for all $t$. 
Thus, we may also refer to $(\eta(t), t\ge 0)$ itself as a Loewner chain.
(When $\kappa>4$, $\SLE_{\kappa}$ curves are not simple and $K_t$ also includes regions swallowed by the curve.)

\smallbreak

For any Dobrushin domain $(\Omega; x, y)$, we extend the definition of the $\SLE_{\kappa}$ via \textit{conformal invariance}: 
given any conformal map $\varphi \colon \HH \to \Omega$ such that $\varphi(0)=x$ and $\varphi(\infty)=y$, 
the law $\mathsf{P}(\Omega; x,y)$ of the $\SLE_{\kappa}$ curve $\tilde{\eta}$ in $\Omega$ from $x$ to $y$ is the pushforward by $\varphi$ of the law $\mathsf{P}(\HH; 0,\infty)$ of $\eta = \varphi^{-1}(\tilde{\eta})$.
Because the law $\mathsf{P}(\HH; 0,\infty)$ is scale-invariant (by Brownian scaling), 
the law $\mathsf{P}(\Omega; x, y)$ is independent of the choice of $\varphi$.

The motivation for O.~Schramm to introduce $\SLE$s in his celebrated work~\cite{SchrammScalinglimitsLERWUST} was indeed their conformal invariance --- crucial for the description of critical interfaces in many statistical mechanics models. 
The $\SLE_\kappa$ curves have another important feature,  \textit{domain Markov property}: 
if $\tau$ is a stopping time for the growing $\SLE_\kappa$ curve $\eta \sim \mathsf{P}(\Omega;x,y)$, then, 
given an initial segment $\eta[0,\tau]$,
the conditional law of the remaining piece $\eta[\tau,\infty)$ is $\mathsf{P}(\Omega_\tau; \eta(\tau),y)$, where
$\Omega_\tau$ is the unbounded component of 
the remaining domain $\Omega \setminus \eta[0,\tau]$ containing the target point $y$ on its boundary.
The domain Markov property lies at the heart of many martingale arguments applicable to problems involving $\SLE$ curves, such as Theorem~\ref{thm::ising_crossingproba}.

\subsection{Partition Functions of Multiple SLEs}
\label{subsec::pre_partitionfucntions}

Next, we discuss the crossing probability amplitudes $\PartF_\alpha$ in Theorem~\ref{thm::ising_crossingproba}.
We frequently use the following parameters (mostly focusing on the case of $\kappa = 3$, with $h = 1/2$):
\begin{align*} 
\kappa\in (0,6] \qquad \qquad \text{and} \qquad \qquad h = \frac{6-\kappa}{2\kappa}.
\end{align*}

The functions $\PartF_\alpha$ are examples of \textit{multiple $\SLE_\kappa$ partition functions}. 
For $\HH$, these are defined as positive smooth functions $\PartF$ of $2N$ real variables $x_1 < \cdots < x_{2N}$ 
satisfying the following two properties: 
\begin{itemize}[leftmargin=3.5em]
\item[$\mathrm{(PDE)}$] \textit{Partial differential equations of second order}: 
We have
\begin{align}\label{eq: multiple SLE PDEs}
\left[ \frac{\kappa}{2} \partial_j + \sum_{i \neq j}\left(\frac{2}{x_{i}-x_{j}} \partial_i - 
\frac{2h}{(x_{i}-x_{j})^{2}}\right) \right] 
\PartF(x_1,\ldots,x_{2N}) =  0 , \quad \text{for all } j \in \{1,\ldots,2N\} .
\end{align}
\item[$\mathrm{(COV)}$] \textit{M\"obius covariance}: 
For all M\"obius maps $\varphi \colon \HH \to \HH$ 
such that $\varphi(x_{1}) < \cdots < \varphi(x_{2N})$, we have
\begin{align}\label{eq: multiple SLE Mobius covariance}
\PartF(x_{1},\ldots,x_{2N}) = 
\prod_{i=1}^{2N} \varphi'(x_{i})^{h} 
\times \PartF(\varphi(x_{1}),\ldots,\varphi(x_{2N})) .
\end{align}
\end{itemize}

Multiple $\SLE_\kappa$ partition functions have been  
studied in many works, e.g.,~\cite{KytolaMultipleSLE, Dubedat:Euler_integrals_for_commuting_SLEs, DubedatCommutationSLE,
KozdronLawlerMultipleSLEs, FloresKlebanPDE4, KytolaPeltolaPurePartitionSLE, PeltolaWuGlobalMultipleSLEs, WuHyperSLE}.  
They give rise to $\SLE_\kappa$ variants known as multiple $\SLE$s
(or $N$-$\SLE_\kappa$ processes),
which are also termed ``commuting $\SLE$s'' due to J.~Dub\'edat~\cite{DubedatCommutationSLE}.
For the purposes of the present article, we shall only need the following description of the marginal law of one curve.

If $\PartF$ is such a partition function, then for each $j \in \{1,\ldots, 2N\}$, a \textit{Loewner chain associated to $\PartF$ 
with launching points $(x_1, \ldots, x_{2N})$ and 
starting from $x_j$}
is defined as the Loewner chain growing from $x_j$ with spectator points $(x_1, \ldots, x_{j-1}, x_{j+1}, \ldots, x_{2N})$ 
whose driving function $W_t$ satisfies the SDEs 
\begin{align} \label{eqn::loewnerchain_partition}
\begin{split}
\ud W_t 
= \; & \sqrt{\kappa} \, \ud B_t 
+\kappa \, \partial_{j}\log\PartF \big(V_t^1, \ldots, V_t^{j-1}, W_t, V_t^{j+1},\ldots, V_t^{2N}\big) \ud t, 
\qquad W_0 = x_j , \\
\ud V_t^i = \; & \frac{2 \,  \ud t}{V_t^i-W_t}, \qquad 
\text{and} \qquad V_0^i = x_i, \quad \text{for } i\neq j. 
\end{split}
\end{align}
This process is well-defined up to the first time $\swal_{x_{j-1}} \wedge \swal_{x_{j+1}}$ when either $x_{j-1}$ or $x_{j+1}$ is swallowed. 
Note that $V_t^i$ is the time-evolution of the spectator point $x_i$, which coincides with $g_t(x_i)$ for $t$ smaller than the swallowing time of $x_i$.
From the PDEs~\eqref{eq: multiple SLE PDEs} and It\^o calculus, we know that 
the following process is a local martingale with respect to the growth of the $\SLE_{\kappa}$ curve in 
$(\HH; x_j,\infty)$:  
\begin{align*} 
M_t(\PartF):=\prod_{i\neq j}g_t'(x_i)^h\times\PartF(g_t(x_1), \ldots, g_t(x_{j-1}), W_t, g_t(x_{j+1}), \ldots, g_t(x_{2N})) .
\end{align*}
Moreover, for any stopping time $\tau$ for which $M_{t\wedge\tau}(\PartF)$ is a martingale, the law of the Loewner chain associated to $\PartF$ starting from $x_j$ is the same as $\mathsf{P}(\HH; x_j,\infty)$ 
weighted by $M_{t\wedge\tau}(\PartF)/M_0(\PartF)$ 
by Girsanov theorem, which introduces the drift in~\eqref{eqn::loewnerchain_partition}. 
We say that the measure has been \textit{tilted} by the local martingale $M(\PartF)$.  
See~\cite{KytolaMultipleSLE, DubedatCommutationSLE, KozdronLawlerMultipleSLEs, LawlerPartitionFunctionsSLE} for more details.

\smallbreak

Amongst solutions to~(\ref{eq: multiple SLE PDEs}--\ref{eq: multiple SLE Mobius covariance}), 
the \textit{pure partition functions} 
$\PartF_{\alpha}$ will be singled out by the following asymptotics property
(that serves as a boundary condition for the PDE system~\eqref{eq: multiple SLE PDEs}):
\begin{itemize}[leftmargin=3.5em]
\item[$\mathrm{(ASY)}$] \textit{Asymptotics}: 
Denoting by $\emptyset$ the link pattern in $\LP_0$, we have $\PartF_\emptyset = 1$, 
and for all $N \geq 1$, $\alpha \in \LP_N$, $j \in \{1, 2, \ldots, 2N-1 \}$, and $\xi \in (x_{j-1}, x_{j+2})$ (with the convention that $x_0=-\infty$ and $x_{2N}=\infty$), we have
\begin{align}\label{eq: multiple SLE asymptotics}
\lim_{x_j , x_{j+1} \to \xi} 
\frac{\PartF_\alpha(x_1 , \ldots , x_{2N})}{(x_{j+1} - x_j)^{-2h}} 
=\begin{cases}
0, \quad &
    \text{if } \link{j}{j+1} \notin \alpha , \\
\PartF_{\hat{\alpha}}(x_{1},\ldots,x_{j-1},x_{j+2},\ldots,x_{2N}), &
    \text{if } \link{j}{j+1} \in \alpha ,
\end{cases} 
\end{align}
where $\hat{\alpha} = \alpha \removeLink \link{j}{j+1} \in \LP_{N-1}$ denotes
the link pattern obtained from $\alpha$ by removing the link $\link{j}{j+1}$ and relabeling 
the remaining indices by the first $2(N-1)$ positive integers. 
\end{itemize}

The term ``pure'' partition function is motivated by the multiple $\SLE_\kappa$ \textit{pure geometries} introduced in the physics literature~\cite{KytolaMultipleSLE} by M.~Bauer, D.~Bernard \&~K.~Kyt{\"o}l{\"a}
(see also~\cite{KytolaPeltolaPurePartitionSLE}), predicting (correctly) that Loewner chains associated to the partition functions $\PartF_\alpha$ correspond to critical interfaces joining together according to the given topological connectivity $\alpha \in \LP_N$ (cf.~Figure~\ref{fig::Ising}). 
In fact, the collection 
$\{ \PartF_\alpha \colon \alpha \in \LP_N\}$ forms 
a basis for a 
space of multiple $\SLE_\kappa$ partition functions of dimension $\frac{1}{N+1} \binom{2N}{N}$.
Each partition function $\PartF$ determines, up to normalization
that vanishes in the logarithmic derivative in~\eqref{eqn::loewnerchain_partition}, 
a (local) multiple $\SLE_\kappa$ process. 
Convex combinations of partition functions give rise to convex combinations of probability measures of such multiple $\SLE_\kappa$ processes --- see~\cite[Appendix~A]{KytolaPeltolaPurePartitionSLE} and~\cite[Section~4.2]{PeltolaWuGlobalMultipleSLEs} for details, and~\cite{Karrila_multiple_SLE} for  a general discussion. 
The pure partition functions $\PartF_\alpha$ correspond to extremal points of a convex set of multiple $\SLE_\kappa$ probability measures.

\begin{definition} \label{defn:PPF}
Fix $\kappa\in (0,6]$. 
The \textit{pure partition functions} of multiple $\SLE_{\kappa}$ for $\HH$
comprise the recursive collection $\{\PartF_{\alpha} \colon \alpha \in \smash{\underset{N\geq 0}{\bigsqcup}} \LP_N\}$ 
of functions of boundary points,  
\begin{align*} 
\PartF_{\alpha} \colon \{ (x_{1},\ldots,x_{2N}) \in \R^{2N} \colon x_{1} < \cdots < x_{2N} \}  \to \Rpos ,
\end{align*}
uniquely determined by 
the properties $\mathrm{(PDE)}$~\eqref{eq: multiple SLE PDEs}, 
$\mathrm{(COV)}$~\eqref{eq: multiple SLE Mobius covariance},
$\mathrm{(ASY)}$~\eqref{eq: multiple SLE asymptotics},
and the following power-law bound: there exist constants $C>0$ and $p>0$ such that for all $N \geq 1$ and $\alpha \in \LP_N$, we have
\begin{align}\label{eqn::powerlawbound}
|\PartF_\alpha(x_1, \ldots, x_{2N})| \le C \prod_{1 \leq i<j \leq 2N}(x_j-x_i)^{\mu_{ij}(p)}, 
\quad \text{where } 
\mu_{ij}(p) :=
\begin{cases}
p, &\text{if } |x_j-x_i| > 1, \\
-p, &\text{if } |x_j-x_i| < 1.
\end{cases}
\end{align}
\end{definition}

The existence and uniqueness of the pure partition functions in this form was shown in~\cite[Theorem~1.1]{PeltolaWuGlobalMultipleSLEs} for $\kappa \leq 4$ and 
in~\cite[Theorem~1.1]{WuHyperSLE} for $\kappa \leq 6$, based on and supplementing various earlier results cited above. 
For the case of $\kappa=3$, which is the central concern of the present article, the Coulomb gas approach of~\cite{KytolaMultipleSLE, Dubedat:Euler_integrals_for_commuting_SLEs, DubedatCommutationSLE, FloresKlebanPDE4, KytolaPeltolaPurePartitionSLE} has problems (because $\kappa$ is rational), while the configurational probabilistic approach of~\cite{KozdronLawlerMultipleSLEs, LawlerPartitionFunctionsSLE, PeltolaWuGlobalMultipleSLEs, WuHyperSLE} 
gives an explicit construction in terms of total masses of multiple $\SLE_\kappa$ measures.
We will not need the explicit construction in this work and thus refer the reader to the literature for more details.

\smallbreak

More generally, the multiple $\SLE_{\kappa}$ partition functions are defined for any 
nice polygon $(\Omega; x_1, \ldots, x_{2N})$ via their conformal images: 
if $\varphi \colon \Omega \to \HH$ is any conformal map such that $\varphi(x_1)<\cdots<\varphi(x_{2N})$, we set
\begin{align} \label{eq:conformal_image}
\PartF(\Omega; x_1, \ldots, x_{2N}) := 
\prod_{i=1}^{2N} |\varphi'(x_i)|^h \times \PartF(\varphi(x_1), \ldots, \varphi(x_{2N})) .
\end{align}
By the M\"obius covariance~\eqref{eq: multiple SLE Mobius covariance} property, this definition is independent of the choice of $\varphi$. 
When $N=1$, there exists only one multiple $\SLE_\kappa$ pure partition function, namely
\begin{align*}
\PartF_{\link{1}{2}}(\Omega; x_1, x_2)=H_{\Omega}(x_1, x_2)^{h}, 
\end{align*}
where $H_{\Omega}(x,y)$ is the \textit{boundary Poisson kernel}, that is, the unique function determined by the properties
\begin{align*}
H_{\HH}(x,y) = |y-x|^{-2}  \qquad \text{and} \qquad H_{\Omega}(x,y) = |\varphi'(x)| |\varphi'(y)| \, H_{\varphi(\Omega)}(\varphi(x),\varphi(y)),
\end{align*}
for any conformal map $\varphi \colon \Omega\to\varphi(\Omega)$.
Let us also note that the boundary Poisson kernel 
has the following useful monotonicity property: for any Dobrushin subdomain $(U; x, y)$ of $(\Omega; x, y)$, we have
\begin{align}\label{eqn::poissonkernel_mono}
H_U(x,y)\le H_{\Omega}(x,y) .
\end{align}

Lastly, we remark that \textit{ratios} of partition functions can be defined for general polygons (i.e., when considering ratios, the niceness assumption can be dropped). 
Namely, suppose that $\PartF_1$ and $\PartF_2$ are two partition functions, and denote their ratio as 
$P(x_1, \ldots, x_{2N}) := \PartF_1(x_1, \ldots, x_{2N})/\PartF_2(x_1, \ldots, x_{2N})$. 
Since both $\PartF_1$ and $\PartF_2$ satisfy~\eqref{eq: multiple SLE Mobius covariance}, the ratio $P$ is M\"{o}bius-invariant (in $\HH$).
For any polygon $(\Omega; x_1, \ldots, x_{2N})$, we set 
\begin{align*}
P(\Omega; x_1, \ldots, x_{2N}) := P(\varphi(x_1), \ldots, \varphi(x_{2N})) ,
\end{align*}
where $\varphi \colon \Omega\to \HH$ is any conformal map such that $\varphi(x_1)<\cdots<\varphi(x_{2N})$. 
The above definition of $P$ is independent of the choice of $\varphi$ thanks to~\eqref{eq: multiple SLE Mobius covariance}. 
We view $P(\Omega; x_1, \ldots, x_{2N})$ as the ratio of $\PartF_1(\Omega; x_1, \ldots, x_{2N})$ and $\PartF_2(\Omega; x_1, \ldots, x_{2N})$, although these two latter functions may be not well-defined, writing 
\begin{align*}
P(\Omega; x_1, \ldots, x_{2N}) = \frac{\PartF_1(\Omega; x_1, \ldots, x_{2N})}{\PartF_2(\Omega; x_1, \ldots, x_{2N})} .
\end{align*}

\subsection{Useful Properties and Bounds}
\label{sub::BoundFunctions}

To end this section, we collect properties of the multiple $\SLE_\kappa$ partition functions that will be needed for a priori estimates in Section~\ref{sec::Ising_partition_function}.
First, we set $\LB_\emptyset := 1$ and define, for all $N \ge 1$ and  $x_1 < \cdots < x_{2N}$, the following \textit{bound functions}: 
\begin{align} \label{eqn::b_def}
\begin{split}
\LB_{\alpha}(x_1, \ldots, x_{2N})  := \; & \prod_{\link{a}{b} \in \alpha} |x_{b}-x_{a}|^{-1} , \quad \alpha \in \LP_N , \\
\LB^{(N)}(x_1, \ldots, x_{2N}) := \; & \prod_{1\leq i<j \leq 2N}(x_j-x_i)^{(-1)^{j-i}} .
\end{split}
\end{align}
More generally, for each nice polygon $(\Omega; x_1, \ldots, x_{2N})$, we set
\begin{align*}
\LB_{\alpha}(\Omega; x_1, \ldots, x_{2N}) := &\; \prod_{\link{a}{b} \in \alpha} H_{\Omega}(x_{a}, x_{b})^{1/2} \; 
=  \; \prod_{i=1}^{2N} |\varphi'(x_i)|^{1/2} \times \LB_{\alpha}(\varphi(x_1),\ldots, \varphi(x_{2N})), \\
\LB^{(N)}(\Omega; x_1, \ldots, x_{2N}) := &\; \prod_{i=1}^{2N} |\varphi'(x_i)|^{1/2} \times \LB^{(N)}(\varphi(x_1),\ldots, \varphi(x_{2N})),
\end{align*}
where $\varphi$ is again any conformal map from $\Omega$ onto $\HH$ such that $\varphi(x_1)<\cdots<\varphi(x_{2N})$.
Note that, in the above definition, $\LB_{\alpha}(\Omega; x_1, \ldots, x_{2N})$ and $\LB^{(N)}(\Omega; x_1, \ldots, x_{2N})$ do not depend on the choice of the conformal map $\varphi$ because $\LB_{\alpha}$ and $\LB^{(N)}$ in~\eqref{eqn::b_def} satisfy the M\"{o}bius covariance~\eqref{eq: multiple SLE Mobius covariance} with $h=1/2$. 

In applications, the following strong bound for the pure partition functions
is very important: for $\kappa\in (0,6]$ and for any nice polygon $(\Omega; x_1, \ldots, x_{2N})$, we have
\begin{align} \label{eqn::optimal_bound_polygon}
0 < \PartF_{\alpha}(\Omega; x_1, \ldots, x_{2N}) \leq \LB_{\alpha}(\Omega; x_1, \ldots, x_{2N})^{2h}. 
\end{align}
This bound was proved in~\cite[Theorem~1.1]{PeltolaWuGlobalMultipleSLEs} for $\kappa \leq 4$
and~\cite[Theorem~1.6]{WuHyperSLE} for $\kappa \leq 6$,
and it was used in~\cite{PeltolaWuGlobalMultipleSLEs} to prove that the pure partition functions with $\kappa=4$
give formulas for connection probabilities of the level lines of the Gaussian free field with alternating boundary data
(see also Appendix~\ref{app::appendix_gff} of the present article).
Properties of the bound functions $\LB_{\alpha}$ were crucial in that proof, and they will also play an essential role in the present article, focusing on the case where $\kappa=3$. 
Note in particular that the bound~\eqref{eqn::optimal_bound_polygon} implies the power law bound~\eqref{eqn::powerlawbound}. 

\smallbreak

Another useful property of the collection 
$\{\PartF_{\alpha} \colon \alpha \in \smash{\underset{N\geq 0}{\bigsqcup}} \LP_N\}$ 
is the following refinement of $\mathrm{(ASY)}$~\eqref{eq: multiple SLE asymptotics}:
\begin{align} \label{eqn::partf_alpha_asy_refined}
\lim_{\substack{\tilde{x}_j , \tilde{x}_{j+1} \to \xi, \\ \tilde{x}_i\to x_i \text{ for } i \neq j, j+1}} 
\frac{\PartF_\alpha(\tilde{x}_1 , \ldots , \tilde{x}_{2N})}{(\tilde{x}_{j+1} - \tilde{x}_j)^{-2h}} 
=\begin{cases}
0 , \quad &
    \text{if } \link{j}{j+1} \notin \alpha , \\
\PartF_{\hat{\alpha}}(x_{1},\ldots,x_{j-1},x_{j+2},\ldots,x_{2N}) , &
    \text{if } \link{j}{j+1} \in \alpha ,
\end{cases}
\end{align}
for all $N \geq 1$, $\alpha \in \LP_N$, $j \in \{1, 2, \ldots, 2N-1 \}$, and $\xi \in (x_{j-1}, x_{j+2})$. 
The property~\eqref{eqn::partf_alpha_asy_refined} 
is proved in~\cite[Lemma~4.3]{PeltolaWuGlobalMultipleSLEs} for $\kappa \leq 4$ and \cite[Corollary~6.9]{WuHyperSLE} for $\kappa \leq 6$.

\smallbreak

Lastly, we define the \textit{symmetric (total) partition function} as
\begin{align}\label{eqn::zsymmetric_def}
\PartF^{(N)}:=\sum_{\alpha\in\LP_N}\PartF_{\alpha}.
\end{align}
By linearity, it satisfies the PDE system~\eqref{eq: multiple SLE PDEs} and M\"obius covariance~\eqref{eq: multiple SLE Mobius covariance}.
Also,~\eqref{eqn::partf_alpha_asy_refined} implies the following 
asymptotics property for the symmetric partition function $\smash{\PartF^{(N)}}$ (see also~\cite[Section~4]{KytolaPeltolaPurePartitionSLE}):
\begin{align}\label{eqn::partf_total_asy_refined}
\lim_{\substack{\tilde{x}_j , \tilde{x}_{j+1} \to \xi, \\ \tilde{x}_i\to x_i \text{ for } i \neq j, j+1}} 
\frac{\PartF^{(N)}(\tilde{x}_1 , \ldots , \tilde{x}_{2N})}{(\tilde{x}_{j+1} - \tilde{x}_j)^{-2h}} 
=
\PartF^{(N-1)}(x_{1},\ldots,x_{j-1},x_{j+2},\ldots,x_{2N}).
\end{align}

\section{Analyzing the Ising Partition Function}
\label{sec::Ising_partition_function}
In this section, we consider the symmetric partition function~\eqref{eqn::zsymmetric_def} for $\kappa=3$,
denoted by $\PartF_{\Ising} = \PartF^{(N)}_{\Ising}$, 
which appears in the denominator of the formula~\eqref{eqn::ising_crossing_proba} for the Ising crossing probabilities.
We prove key results needed in the proof of the main Theorem~\ref{thm::ising_crossingproba},
concerning the following properties of $\PartF_{\Ising}$:

\begin{proposition}\label{prop::Z_total_asy_refined}
Fix $N \geq 1$. For all $\ell \in \{ 1,2, \ldots, N \}$
and $\xi < x_{2\ell+1} < x_{2\ell+2} < \cdots < x_{2N}$, we have 
\begin{align} \label{eqn::cascade_asy}
\lim_{\substack{\tilde{x}_1, \ldots, \tilde{x}_{2\ell} \to \xi, \\ \tilde{x}_{i} \to x_i \text{ for } 2\ell < i \le  2N}} 
\frac{\PartF^{(N)}_{\Ising}(\tilde{x}_1, \ldots, \tilde{x}_{2N})}{\PartF^{(\ell)}_{\Ising}(\tilde{x}_1, \ldots, \tilde{x}_{2\ell})} 
= \PartF^{(N-\ell)}_{\Ising}(x_{2\ell+1},\ldots, x_{2N}).
\end{align}
\end{proposition}

\begin{proposition} \label{prop::Z_total_ineq}
Fix $N \geq 1$. We have
\begin{align} \label{eqn::Z_total_ineq} 
\frac{1}{\sqrt{N!}}  \; 
\LB^{(N)}(x_1,\ldots, x_{2N}) \le \PartF_{\Ising}^{(N)}(x_1,\ldots, x_{2N}) \le (2N-1)!! \; \LB^{(N)}(x_1,\ldots, x_{2N}) .
\end{align}
\end{proposition} 
These bounds are not sharp in general, but they are nevertheless sufficient for our purposes. 

\smallbreak

Another important result in this section is Proposition~\ref{prop::mart_wrong_zero}, 
concerning the boundary behavior of the ratios 
$\PartF_\alpha/\PartF_{\Ising}$ of partition functions
when the variables move under a Loewner evolution.
Next, in Section~\ref{subsec::properties_boundfunctions} 
we collect general identities concerning the bound functions $\LB_{\alpha}$ and $\smash{\LB^{(N)}}$.
Then, we focus on the case of $\kappa=3$. 
We prove Propositions~\ref{prop::Z_total_asy_refined} and~\ref{prop::Z_total_ineq}
respectively in Sections~\ref{subsec::proof_Z_total_asy_refined} and~\ref{subsec::proof_Z_total_ineq}. 
Finally, we state and prove 
Proposition~\ref{prop::mart_wrong_zero} in Section~\ref{subsec::technical_Z_total}. 
Its proof relies on Propositions~\ref{prop::Z_total_asy_refined} and~\ref{prop::Z_total_ineq}.

\subsection{Properties of the Bound Functions}
\label{subsec::properties_boundfunctions}

To begin, we consider the bound functions $\LB_{\alpha}$ and $\smash{\LB^{(N)}}$ defined in Section~\ref{sub::BoundFunctions}.
In particular, they satisfy properties similar to those appearing in Propositions~\ref{prop::Z_total_asy_refined} and~\ref{prop::Z_total_ineq} --- 
see Lemmas~\ref{lem::b_total_asy_refined} and~\ref{lem::b_total_ineq}.
These results can in fact be applied to analyze multiple $\SLE_\kappa$ partition functions for any $\kappa \in (0,6]$. 
For instance, in~\cite{PeltolaWuGlobalMultipleSLEs} related results were used to 
prove that connection probabilities of level lines of the Gaussian free field are given by multiple $\SLE_\kappa$ 
pure partition functions with $\kappa=4$ (see also Appendix~\ref{app::appendix_gff}).

\begin{lemma}\label{lem::b_total_asy_refined}
\textnormal{\cite[Lemma~A.2]{PeltolaWuGlobalMultipleSLEs}}
Fix $N \geq 1$. For all $\ell \in \{ 1, \ldots, N \}$
and $\xi < x_{2\ell+1} < \cdots < x_{2N}$,~we~have 
\begin{align*}
\lim_{\substack{\tilde{x}_1, \ldots, \tilde{x}_{2\ell} \to \xi, \\ \tilde{x}_{i} \to x_i \text{ for } 2\ell < i \le  2N}} 
\frac{\LB^{(N)}(\tilde{x}_1, \ldots, \tilde{x}_{2N})}{\LB^{(\ell)}(\tilde{x}_1, \ldots, \tilde{x}_{2\ell})} 
= \LB^{(N-\ell)}(x_{2\ell+1},\ldots, x_{2N}).
\end{align*}
\end{lemma}

\begin{lemma}\label{lem::b_total_ineq}
Fix $p \ge 0$ and $N\ge 1$. We have
\begin{align}\label{eqn::b_total_ineq}
\frac{1}{((2N-1)!!)^{p} }
\le \sum_{\alpha\in\LP_N} \bigg(\frac{\LB_{\alpha}(x_1,\ldots, x_{2N})}{\LB^{(N)}(x_1,\ldots, x_{2N})}\bigg)^{p} 
\le (2N-1)!! .
\end{align}
\end{lemma}
\begin{proof}
We prove~\eqref{eqn::b_total_ineq} by induction on $N \geq 1$. The initial case $N=1$ is trivial. 
Let then $N \ge 2$ and assume that Equation~\eqref{eqn::b_total_ineq} holds up to $N-1$. 
A straightforward calculation shows that
\begin{align*}
\frac{\LB_{\alpha}(x_1, \ldots, x_{2N})}{\LB^{(N)}(x_1, \ldots, x_{2N})} 
= \frac{\LB_{\alpha \removeLink \link{j}{j+1}}(x_1, \ldots, x_{j-1}, x_{j+2}, \ldots, x_{2N})}{\LB^{(N-1)}(x_1, \ldots, x_{j-1}, x_{j+2}, \ldots, x_{2N})}
\prod_{\substack{1 \leq i \leq 2N , \\ i \neq j, j+1}} \bigg| \frac{x_i - x_{j+1}}{x_i - x_j} \bigg|^{(-1)^{i+j}} ,
\end{align*}
for any $1\le j\le 2N-1$ and $\alpha \in \LP_N$ such that $\link{j}{j+1}\in\alpha$.
The product expression in the above formula is the same as the probability $P^{(j,j+1)} = P^{(j+1,j)}$ in~\eqref{eqn::levellines_proba_lk}
in Appendix~\ref{app::appendix_gff} in Appendix~\ref{app::appendix_gff}. 
Thus, we have 
\begin{align*}
\frac{\LB_{\alpha}(x_1, \ldots, x_{2N})}{\LB^{(N)}(x_1, \ldots, x_{2N})} 
= \frac{\LB_{\alpha \removeLink \link{j}{j+1}}(x_1, \ldots, x_{j-1}, x_{j+2}, \ldots, x_{2N})}{\LB^{(N-1)}(x_1, \ldots, x_{j-1}, x_{j+2}, \ldots, x_{2N})} \;
P^{(j,j+1)}(x_1, \ldots, x_{2N}). 
\end{align*}
Using the above observation, we first prove the upper bound in~\eqref{eqn::b_total_ineq}:
\begin{align*}
& \; \sum_{\alpha\in\LP_N}\bigg(\frac{\LB_{\alpha}(x_1, \ldots, x_{2N})}{\LB^{(N)}(x_1, \ldots, x_{2N})}\bigg)^{p} \\
\le & \; \sum_{j=1}^{2N-1}
\sum_{\substack{\alpha \in \LP_N , \\ \link{j}{j+1} \in \alpha}} 
\bigg(\frac{\LB_{\alpha}(x_1, \ldots, x_{2N})}{\LB^{(N)}(x_1, \ldots, x_{2N})}\bigg)^{p} \\
\le & \; \sum_{j=1}^{2N-1} \sum_{\substack{\alpha \in \LP_N , \\ \link{j}{j+1} \in \alpha}} 
\bigg(\frac{\LB_{\alpha \removeLink \link{j}{j+1}}(x_1, \ldots, x_{j-1}, x_{j+2}, \ldots, x_{2N})}{\LB^{(N-1)}(x_1, \ldots, x_{j-1}, x_{j+2}, \ldots, x_{2N})}\bigg)^{p}
&& \text{[since $P^{(j,j+1)}\le 1$ in~\eqref{eqn::levellines_proba_lk}]} \\ 
\le & \; \sum_{j=1}^{2N-1}(2N-3)!!=(2N-1)!!.
&& \text{[by the ind.~hypothesis]} 
\end{align*}
To prove the lower bound in~\eqref{eqn::b_total_ineq}, we first estimate
\begin{align*}
& \; \sum_{\alpha\in\LP_N}
\bigg(\frac{\LB_{\alpha}(x_1, \ldots, x_{2N})}{\LB^{(N)}(x_1, \ldots, x_{2N})}\bigg)^{p} \\
\ge & \; \max_{1\le j\le 2N-1}\sum_{\alpha \colon \link{j}{j+1} \in\alpha}\bigg(\frac{\LB_{\alpha}(x_1, \ldots, x_{2N})}{\LB^{(N)}(x_1, \ldots, x_{2N})}\bigg)^{p}\\
=& \; \max_{1\le j\le 2N-1}
\sum_{\substack{\alpha \in \LP_N , \\ \link{j}{j+1} \in \alpha}} 
\bigg(\frac{\LB_{\alpha \removeLink \link{j}{j+1}}(x_1, \ldots, x_{j-1}, x_{j+2}, \ldots, x_{2N})}{\LB^{(N-1)}(x_1, \ldots, x_{j-1}, x_{j+2}, \ldots, x_{2N})}\bigg)^{p} \big(P^{(j,j+1)}(x_1, \ldots, x_{2N})\big)^{p} \\
\ge & \;  \big((2N-3)!!\big)^{-p} \max_{1\le j\le 2N-1} \big(P^{(j,j+1)}(x_1, \ldots, x_{2N})\big)^{p}.
\qquad \qquad \qquad \quad \text{[by the ind.~hypothesis]}
\end{align*}
We now note that $\sum_j P^{(j,j+1)}\ge 1$, 
which implies 
$\underset{j}{\max} \, P^{(j,j+1)}\ge \frac{1}{2N-1}$ 
and shows the lower bound. 
\end{proof}

\begin{remark} \label{rem::btotal_mono}
Suppose $(U; x_1, \ldots, x_{2N})$ is a nice sub-polygon of $(\Omega; x_1, \ldots, x_{2N})$. Then by~\eqref{eqn::poissonkernel_mono}, we have
\begin{align*} 
\LB_{\alpha}(U; x_1, \ldots, x_{2N})\le \LB_{\alpha}(\Omega; x_1, \ldots, x_{2N}). 
\end{align*}
Combining this with~\eqref{eqn::b_total_ineq}, we see that 
\begin{align} \label{eqn::btotal_mono}
\LB^{(N)}(U; x_1, \ldots, x_{2N})\le ((2N-1)!!)^2 \, \LB^{(N)}(\Omega; x_1, \ldots, x_{2N}). 
\end{align}
\end{remark}

\begin{corollary}\label{cor::Ztotal_upper}
Fix $\kappa\in (0,6]$ and let $\PartF^{(N)}$ be the symmetric partition function~\eqref{eqn::zsymmetric_def}. Then, we have 
\begin{align*}
\PartF^{(N)}(x_1, \ldots, x_{2N})\le (2N-1)!! \big(\LB^{(N)}(x_1, \ldots, x_{2N})\big)^{2h}.
\end{align*} 
\end{corollary}
\begin{proof}
This follows by combining~\eqref{eqn::optimal_bound_polygon} with the upper bound in~\eqref{eqn::b_total_ineq} for $p = 2h$.
\end{proof}

Corollary~\ref{cor::Ztotal_upper} with $\kappa=3$ immediately gives the upper bound in Proposition~\ref{prop::Z_total_ineq}.
However, the ratio $\PartF_\alpha/\LB_\alpha$ can be arbitrarily small, so the lower bound in Proposition~\ref{prop::Z_total_ineq} 
cannot be derived easily from the lower bound in Lemma~\ref{lem::b_total_ineq}. 
To establish the lower bound, 
we first prove a useful identity for $\smash{\LB^{(N)}}$ in Lemma~\ref{lem: compare Bs}.

\begin{remark} \label{rem::crossratio}
Let us record a trivial but helpful inequality here: for $x_1<x_2<x_3<x_4$, we have
\begin{align} \label{eqn::crossratio}
\frac{(x_4-x_1)(x_3-x_2)}{(x_3-x_1)(x_4-x_2)}\le 1. 
\end{align}
\end{remark}

\begin{lemma} \label{lem: compare Bs}
Fix $N \geq 1$ and $j \in \{1, 2, \ldots,2N-1\}$, 
and denote by 
$y_i^{j} := x_i$ for $1 \leq i \leq j-1$ and $y_i^{j} := x_{i+1}$ for $j \leq i \leq 2N-2$. 
Then, we have 
\begin{align} \label{eq: compare Bs}
\begin{split}
\; & \frac{\LB^{(N)}(x_1,\ldots, x_{2N})}{\LB^{(N-1)}(y_1^{j},\ldots, y_{2N-2}^{j})} \\
= \; & (x_{2N}-x_j)^{(-1)^{j}} 
\prod_{\substack{1 \leq l < 2N , \\ l \neq j}} |x_l-x_j|^{(-1)^{l-j}} \prod_{\substack{1 \leq k < 2N , \\ k \neq j}} (x_{2N}-x_k)^{(-1)^{k}}
\Bigg( \prod_{\substack{1\leq k < j , \\ j \leq l < 2N-1}} (y_l^{j} - y_k^{j})^{(-1)^{l-k+1}} \Bigg)^2 .
\end{split}
\end{align}
In particular, by taking $j=2N-1$, we have
\begin{align}
\label{eqn::btotal_compare}
\begin{split}
\frac{\LB^{(N)}(x_1, \ldots, x_{2N})}{\LB^{(N-1)}(x_1, \ldots, x_{2N-2})}
= \; & \frac{1}{x_{2N}-x_{2N-1}}
\; \prod_{i=1}^{2N-2}\bigg(\frac{x_{2N}-x_i}{x_{2N-1}-x_i}\bigg)^{(-1)^i}\\
\le \; & \frac{x_{2N}-x_{2N-2}}{(x_{2N}-x_{2N-1})(x_{2N-1}-x_{2N-2})}.
\end{split}
\end{align}

\end{lemma}
\begin{proof}
First, we use the definition~\eqref{eqn::b_def} to write
\begin{align*}
\LB^{(N)}(x_1,\ldots, x_{2N}) 
= \; & (x_{2N}-x_j)^{(-1)^{j}} 
\prod_{\substack{1 \leq l < 2N , \\ l \neq j}} |x_l-x_j|^{(-1)^{l-j}} \prod_{\substack{1 \leq k < 2N , \\ k \neq j}} (x_{2N}-x_k)^{(-1)^{k}}  \\
\; & \times 
\prod_{1\leq k < l < j}(x_l-x_k)^{(-1)^{l-k}} 
\prod_{j \leq k < l < 2N}(x_l-x_k)^{(-1)^{l-k}} 
\prod_{\substack{1\leq k < j , \\ j < l < 2N}}(x_l-x_k)^{(-1)^{l-k}} .
\end{align*}
Then, we convert the last three products into expressions in the variables $y_1^{j},\ldots, y_{2N-2}^{j}$:
\begin{align*}
\; & \prod_{1\leq k < l < j}(x_l-x_k)^{(-1)^{l-k}} 
\prod_{j \leq k < l < 2N}(x_l-x_k)^{(-1)^{l-k}} 
\prod_{\substack{1\leq k < j , \\ j < l < 2N}}(x_l-x_k)^{(-1)^{l-k}} \\
= \; & \prod_{1\leq k < l < j}(y_l^{j}-y_k^{j})^{(-1)^{l-k}} 
\prod_{j -1 \leq k < l < 2N-1}(y_l^{j}-y_k^{j})^{(-1)^{l-k}} 
\prod_{\substack{1\leq k < j , \\ j \leq l < 2N-1}}(y_l^{j}-y_k^{j})^{(-1)^{l-k+1}} \\
= \; &  \Bigg( \prod_{\substack{1\leq k < j , \\ j \leq l < 2N-1}} (y_l^{j} - y_k^{j})^{(-1)^{l-k+1}}  \Bigg)^2
\LB^{(N-1)}(y_1^{j},\ldots, y_{2N-2}^{j}) .
\end{align*}
Combining these formulas, we get~\eqref{eq: compare Bs}. 
The first line of~\eqref{eqn::btotal_compare} follows from~\eqref{eqn::b_def} and the second from~\eqref{eqn::crossratio}.   
\end{proof}

\subsection{Cascade Asymptotics --- Proof of Proposition~\ref{prop::Z_total_asy_refined}}
\label{subsec::proof_Z_total_asy_refined}

The symmetric partition function $\PartF_{\Ising}$ 
has an explicit Pfaffian formula, already well-known in the physics literature, and appearing, e.g., 
in~\cite{KytolaPeltolaPurePartitionSLE, IzyurovIsingMultiplyConnectedDomains, PeltolaWuGlobalMultipleSLEs} in the context of $\SLE$s. 
To state it, we use the following notation: we let $\Pi_N$ denote the set of all \textit{pair partitions} 
$\varpi = \{\{a_1,b_1\} , \ldots, \{a_N, b_N\}\}$ of the set $\{1,2,\ldots,2N\}$, 
that is, partitions of this set into $N$ disjoint two-element subsets $\{a_j,b_j\} \subset \{1,\ldots,2N\}$, 
with the convention that $a_1 < a_2 < \cdots < a_N$ and $a_j < b_j$ for all $j \in \{ 1, \ldots, N\}$. 
We also denote by $\mathrm{sgn}(\varpi)$ the sign of the partition $\varpi$ 
defined as the sign of the product $\prod (a-c) (a-d) (b-c) (b-d)$ over pairs of distinct elements $\{a,b\},\{c,d\} \in \varpi$.
Note that the set of link patterns $\LP_N$ is the subset of $\Pi_N$ consisting of planar pair partitions. 
With this notation, the function~\eqref{eqn::zsymmetric_def} with $\kappa=3$ reads
(see, e.g.,~\cite[Lemma~4.13]{PeltolaWuGlobalMultipleSLEs} and~\cite[Proposition~4.6]{KytolaPeltolaPurePartitionSLE} for a proof) 
\begin{align}\label{eq: ising symmetric pff}
\PartF^{(N)}_{\Ising}(x_1, \ldots, x_{2N}) 
= \mathrm{pf} \bigg(  \frac{1}{x_{j}-x_{i}} \bigg)_{i,j=1}^{2N}
:= \sum_{\varpi \in \Pi_N} 
\mathrm{sgn}(\varpi) \prod_{\{a,b\} \in \varpi} \frac{1}{x_{b}-x_{a}}  ,
\end{align}
for $\Omega = \HH$, and it is again defined for general nice polygons via conformal covariance:
\begin{align*}
\PartF^{(N)}_{\Ising}(\Omega; x_1, \ldots, x_{2N}) := 
\prod_{i=1}^{2N} |\varphi'(x_i)|^{1/2}\times\PartF^{(N)}_{\Ising}(\varphi(x_1), \ldots, \varphi(x_{2N})) ,
\end{align*}
with any conformal map $\varphi \colon \Omega \to \HH$ such that $\varphi(x_1)<\cdots<\varphi(x_{2N})$. 
The formula~\eqref{eq: ising symmetric pff} shows 
that, up to a sign, the function $\PartF_{\Ising}$ satisfies M\"{o}bius covariance also for maps that move the point $\infty$. Indeed, for any conformal map $\varphi \colon \HH \to \HH$, we have
(see, e.g.,~\cite[Lemma~4.7 and proof of Proposition~4.6]{KytolaPeltolaPurePartitionSLE})
\begin{align} \label{eq:general mobius}
\begin{split}
\PartF^{(N)}_{\Ising}(x_1, \ldots, x_{2N}) 
= \; & \prod_{i=1}^{2N} \varphi'(x_i)^{1/2} \times \PartF^{(N)}_{\Ising}(\varphi(x_{1}), \ldots, \varphi(x_{2N})) \\
= \; & (-1)^{|\iota|} 
\prod_{i=1}^{2N} \varphi'(x_i)^{1/2} \times \PartF^{(N)}_{\Ising}(\varphi(x_{\iota(1)}), \ldots, \varphi(x_{\iota(2N)})) ,
\end{split}
\end{align}
where $\{x_{\iota(1)}, \ldots, x_{\iota(2N)}\} = \{x_{1}, \ldots, x_{2N}\}$ satisfy $\varphi(x_{\iota(1)}) < \cdots < \varphi(x_{\iota(2N)})$ and $|\iota|$ is the number of indices $j$ in $\{2,3,\ldots,2N\}$ such that $\varphi(x_{\iota(j)}) < \varphi(x_{\iota(1)})$. Note that $\iota$ is a cyclic permutation of $\{1,\ldots,2N\}$ and 
the Pfaffian function $\mathrm{Pf}(z_1, \ldots, z_{2N})$ is odd in the sense that for any $i<j$, we have
\begin{align} \label{eqn::odd}
\mathrm{Pf}(z_1, \ldots, z_i, \ldots, z_j, \ldots,  z_{2N}) = - \mathrm{Pf}(z_1, \ldots, z_j, \ldots, z_i, \ldots,  z_{2N}) .
\end{align}

From the Pfaffian formula, it is not obvious that $\PartF_{\Ising} > 0$, but this is indeed the case:
the positivity follows, e.g., from its definition~\eqref{eqn::zsymmetric_def} and the fact that each $\PartF_{\alpha}$ is positive by~\cite[Theorem~1.1]{PeltolaWuGlobalMultipleSLEs}: 
\begin{align} \label{eq::TotalPartPos}
\PartF^{(N)}_{\Ising} := \sum_{\alpha\in\LP_N}\PartF_{\alpha} > 0 .
\end{align}
Alternatively, as pointed out by the referee, 
by using Lemma~\ref{lem::Hafnian} and studying the asymptotics of~\eqref{eq: ising symmetric pff}, one may check that 
$\PartF_{\Ising} > 0$ from the formula~\eqref{eq::Hafnian} below.

\smallbreak

Our next aim is to prove the cascade asymptotics property, Proposition~\ref{prop::Z_total_asy_refined}, for the function 
$\PartF_{\Ising}$.
In general, limiting behavior of functions of several variables is rather delicate, and indeed,
even with the explicit formula~\eqref{eq: ising symmetric pff} for 
$\PartF_{\Ising}$, the analysis of its behavior as the variables tend together is non-trivial. 
The problem with using formula~\eqref{eq: ising symmetric pff} is that it includes
a sum of positive and negative terms (which could in principle lead to cancellations and a signed expression).
However, thanks to the following Hafnian identity, a sum over \textit{non-negative} terms, 
we are able to carry out the analysis required for Proposition~\ref{prop::Z_total_asy_refined}. 
This identity is well-known in the literature~\cite{ID_book, DMS:CFT}: it is a manifestation of ``bosonization identities'' for the Ising model.  
For completeness, we give a proof for it here.

\begin{lemma} \label{lem::Hafnian}
The following identity holds for all $z_1, \ldots, z_{2N} \in \C$ with $z_i \neq z_j$ for all $i \neq j$:
\begin{align}\label{eq::Hafnian}
\left( \mathrm{pf} \bigg(  \frac{1}{z_{j}-z_{i}} \bigg)_{i,j=1}^{2N} \right)^2
= \mathrm{hf} \bigg(  \frac{1}{(z_{j}-z_{i})^2} \bigg)_{i,j=1}^{2N}
:= \sum_{\varpi \in \Pi_N} \prod_{\{a,b\} \in \varpi} \frac{1}{(z_{b}-z_{a})^2}  .
\end{align}
\end{lemma}
\begin{proof}
Expanding the square of the Pfaffian~\eqref{eq: ising symmetric pff}, we have
\begin{align} \label{eq::Expand}
\left( \mathrm{pf} \bigg(  \frac{1}{z_{j}-z_{i}} \bigg)_{i,j=1}^{2N} \right)^2
= \; & \sum_{\varpi, \varpi' \in \Pi_N} 
\mathrm{sgn}(\varpi) \, \mathrm{sgn}(\varpi') \prod_{\{a,b\} \in \varpi} \frac{1}{z_{b}-z_{a}} \prod_{\{c,d\} \in \varpi'} \frac{1}{z_{d}-z_{c}} .
\end{align}
We see from the asserted formula~\eqref{eq::Hafnian} that the diagonal terms $\varpi = \varpi'$ yield the desired Hafnian expression,
and therefore, we only need to prove that all of the off-diagonal terms in~\eqref{eq::Expand} cancel out. 
To establish this, we use induction on $N \geq 1$. The initial case $N = 1$ is clear. Let us then assume that
\begin{align*} 
\sum_{\substack{\varpi, \varpi' \in \Pi_{N-1} , \\ \varpi \neq \varpi' }} \mathrm{sgn}(\varpi) \, \mathrm{sgn}(\varpi')
\prod_{\{a,b\} \in \varpi} \frac{1}{u_{b}-u_{a}} \prod_{\{c,d\} \in \varpi'} \frac{1}{u_{d}-u_{c}} = 0 , \qquad
\left\{\parbox{4.5cm}{\text{ for all $u_1, \ldots, u_{2N-2} \in \C$ } \\ \text{ with $u_i \neq u_j$ for all  $i \neq j$ } }\right\} .
\end{align*}
Fix the variables $z_1, \ldots, z_{2N-1} \in \C$ at arbitrary distinct positions, denote by
\begin{align*}
\mathrm{Pf}(z_1, \ldots, z_{2N})& := \sum_{\varpi \in \Pi_N} 
\mathrm{sgn}(\varpi) \prod_{\{a,b\} \in \varpi} \frac{1}{z_{b}-z_{a}},\\
\mathrm{Hf}(z_1, \ldots, z_{2N}) & := \sum_{\varpi \in \Pi_N}  \prod_{\{a,b\} \in \varpi} \frac{1}{(z_{b}-z_{a})^2} ,
\end{align*}
and consider the following meromorphic function of $z \in \C$:
\begin{align} \label{eq::merom}
\begin{split}
F(z) := \; & (\mathrm{Pf}(z_1, \ldots, z_{2N-1}, z))^2 - \mathrm{Hf}(z_1, \ldots,  z_{2N-1}, z) \\
= \; & \Bigg( \sum_{\substack{\varpi, \varpi' \in \Pi_N , \\ \varpi \neq \varpi' }} 
\frac{\mathrm{sgn}(\varpi) \, \mathrm{sgn}(\varpi')}{(z-z_{\varpi(2N)})(z-z_{\varpi'(2N)})} 
\prod_{\substack{\{a,b\} \in \varpi , \\ b \neq 2N}} \frac{1}{z_{b}-z_{a}} 
\prod_{\substack{\{c,d\} \in \varpi' , \\ d \neq 2N}} \frac{1}{z_{d}-z_{c}} \Bigg) ,
\end{split}
\end{align}
where $\varpi(2N)$ (resp.~$\varpi'(2N)$) is the pair of $2N$ in $\varpi$ (resp.~$\varpi'$), that is, $\{\varpi(2N),2N\} \in \varpi$ (resp. $\{\varpi'(2N),2N\} \in \varpi'$). 
We aim to prove that the function $F$ is identically zero.
$F(z)$ vanishes at $z \to \infty$ and it can only have poles of degree at most two at $z = z_j$ for some $j \in \{1, 2, \ldots,2N-1\}$. 
We will show that these points are in fact not poles. 
By symmetry, it suffices to consider the Laurent series expansion of $F$ at $\varepsilon = z - z_{2N-1}$, 
and by translation invariance, $F$ depends on $z_{2N-1}$ and $z$ only via their difference.
In particular, $F$ is a meromorphic function of $\varepsilon \in \C$ and the Laurent series reads
\begin{align} \label{eq::Laurent}
F(z) = \varepsilon^{-2} A_{-2}(z_1, \ldots, z_{2N-2}) + \varepsilon^{-1} A_{-1}(z_1, \ldots, z_{2N-2}) 
+ \sum_{n = 0}^\infty \varepsilon^n A_{n}(z_1, \ldots, z_{2N-2}) .
\end{align}
To exclude the first order poles, we show that the right-hand side of~\eqref{eq::Laurent} is an even function of $\varepsilon$.
Notice that the Pfaffian function $\mathrm{Pf}(z_1, \ldots, z_{2N})$ is odd in the sense of~\eqref{eqn::odd}, while 
the function $\mathrm{Hf}(z_1, \ldots, z_{2N})$ is even in the sense that for any $i<j$, we have
\begin{align*}
\mathrm{Hf}(z_1, \ldots, z_i, \ldots, z_j, \ldots,  z_{2N}) = \mathrm{Hf}(z_1, \ldots, z_j, \ldots, z_i, \ldots,  z_{2N}) .
\end{align*}
In particular, choosing $i = 2N-1$ and $j = 2N$, we have 
\begin{align*}
F(z) := \; & (\mathrm{Pf}(z_1, \ldots, z_{2N-2}, z_{2N-1}, z))^2 - \mathrm{Hf}(z_1, \ldots, z_{2N-2}, z_{2N-1}, z) \\
= \; & (\mathrm{Pf}(z_1, \ldots, z_{2N-2}, z, z_{2N-1}))^2 - \mathrm{Hf}(z_1, \ldots, z_{2N-2}, z, z_{2N-1}) .
\end{align*}
Now, the Laurent series expansion of the second line as a function of $z_{2N-1} \in \C$ at $\delta = z_{2N-1} - z$ reads
\begin{align*}
\delta^{-2} A_{-2}(z_1, \ldots, z_{2N-2}) + \delta^{-1} A_{-1}(z_1, \ldots, z_{2N-2}) 
+ \sum_{n = 0}^\infty \delta^n A_{n}(z_1, \ldots, z_{2N-2}) = F(z) ,
\end{align*}
and because $\delta = - \varepsilon$, this shows that the Laurent series expansion~\eqref{eq::Laurent} in $\varepsilon = z - z_{2N-1}$ 
is invariant under the flip $\varepsilon \mapsto -\varepsilon$. Therefore, $A_{-1} \equiv 0$ in~\eqref{eq::Laurent}.

The term of order $\varepsilon^{-2}$ in~\eqref{eq::Laurent} 
is given by those terms in~\eqref{eq::merom} which 
satisfy $\varpi(2N) = \varpi'(2N) = 2N-1$, that is, $\{2N-1,2N\} \in \varpi \cap \varpi'$.
Removing this pair from both $\varpi$ and $\varpi'$ results in two different pair partitions of $2N-2$ points, which have the same signs as $\varpi$ and $\varpi'$, respectively.
Therefore, we have 
\begin{align*} 
A_{-2}(z_1, \ldots, z_{2N-2})
= \sum_{\substack{\varpi, \varpi' \in \Pi_{N-1} , \\ \varpi \neq \varpi' }} 
\mathrm{sgn}(\varpi) \, \mathrm{sgn}(\varpi')
\prod_{\{a,b\} \in \varpi} \frac{1}{z_{b}-z_{a}} \prod_{\{c,d\} \in \varpi'} \frac{1}{z_{d}-z_{c}} .
\end{align*}
This expression is zero by the induction hypothesis. Thus, the function $F \colon \C \to \C$ has no poles.

In conclusion, we have shown that $F$ is an entire function with $\lim_{z \to \infty} F(z) = 0$.
Therefore, it is bounded and by Liouville's theorem, $F \equiv 0$. This shows that $(\mathrm{Pf})^2 \equiv \mathrm{Hf}$ and implies~\eqref{eq::Hafnian}. 
\end{proof}

\begin{proof}[Proof of Proposition~\ref{prop::Z_total_asy_refined}]
By the identity~\eqref{eq::TotalPartPos}, we have $\PartF_{\Ising}^{(\ell)} > 0$ for all $\ell \geq 1$, 
so it is sufficient to prove the corresponding statement for the squares, that is, to verify the identity 
\begin{align*}
\lim_{\substack{\tilde{x}_1, \ldots, \tilde{x}_{2\ell} \to \xi, \\ \tilde{x}_{i} \to x_i \text{ for } 2\ell < i \le  2N}} 
\bigg( \frac{\PartF^{(N)}_{\Ising}(\tilde{x}_1, \ldots, \tilde{x}_{2N})}{\PartF^{(\ell)}_{\Ising}(\tilde{x}_1, \ldots, \tilde{x}_{2\ell})} \bigg)^2 
= \big( \PartF^{(N-\ell)}_{\Ising}(x_{2\ell+1},\ldots, x_{2N})  \big)^2 .
\end{align*}
Denoting $J_\ell = \{1,2,\ldots,2\ell\}$ and using the explicit formula from Lemma~\ref{lem::Hafnian}, we write
\begin{align*}
\big( \PartF_{\Ising}^{(N)}(\tilde{x}_1, \ldots, \tilde{x}_{2N}) \big)^2
= \; & \sum_{\varpi \in \Pi_N} \prod_{\{a,b\} \in \varpi} \frac{1}{(\tilde{x}_{b}-\tilde{x}_{a})^2} \\
= \; & \sum_{\varpi \in \Pi_N} 
\bigg( \prod_{\substack{\{a,b\} \in \varpi , \\ a,b \in J_\ell }} \frac{1}{(\tilde{x}_{b}-\tilde{x}_{a})^2} \bigg) 
\bigg( \prod_{\substack{\{a,b\} \in \varpi , \\ a \in J_\ell, \, b \notin J_\ell }} \frac{1}{(\tilde{x}_{b}-\tilde{x}_{a})^2} \bigg) 
\bigg( \prod_{\substack{\{a,b\} \in \varpi , \\ a,b\notin J_\ell }} \frac{1}{(\tilde{x}_{b}-\tilde{x}_{a})^2} \bigg)  .
\end{align*}
In the limit~\eqref{eqn::cascade_asy} of the ratio $( \PartF_{\Ising}^{(N)} / \PartF_{\Ising}^{(\ell)} )^2$, 
only those terms in the sum over $\varpi \in \Pi_N$ in $(\PartF_{\Ising}^{(N)} )^2$ 
for which $\# \{ \{a,b\} \in \varpi \colon a,b \in J_\ell \} \geq \ell$ can give a non-zero contribution. 
These are exactly the pair partitions that factorize into two parts: $\varpi = \varpi_1 \cup \varpi_2$, 
with $\varpi_1 \in \Pi_\ell$ and $\varpi_2 \in \Pi_{N-\ell}$ being pair partitions of $J_\ell$ and $\{2\ell+1,2\ell+2,\ldots,2N\}$, respectively. 
Thus, we have verified that
\begin{align*}
& \lim_{\substack{\tilde{x}_1, \ldots, \tilde{x}_{2\ell} \to \xi, \\ \tilde{x}_{i} \to x_i \text{ for } 2\ell < i \le  2N}} 
\bigg( \frac{\PartF_{\Ising}^{(N)}(\tilde{x}_1, \ldots, \tilde{x}_{2N})}{\PartF_{\Ising}^{(\ell)}(\tilde{x}_1, \ldots, \tilde{x}_{2\ell})} \bigg)^2 \\
= \; & \lim_{\substack{\tilde{x}_1, \ldots, \tilde{x}_{2\ell} \to \xi, \\ \tilde{x}_{i} \to x_i \text{ for } 2\ell < i \le  2N}} 
\frac{\sum_{\varpi_1 \in \Pi_\ell} \Big(\prod_{\{a,b\} \in \varpi_1} (\tilde{x}_{b}-\tilde{x}_{a})^{-2} \Big)
\sum_{\varpi_2 \in \Pi_{N-\ell}} \Big(\prod_{\{a,b\} \in \varpi_2} (\tilde{x}_{b}-\tilde{x}_{a})^{-2} \Big) }{
\sum_{\varpi_1 \in \Pi_\ell} \Big(\prod_{\{a,b\} \in \varpi_1} (\tilde{x}_{b}-\tilde{x}_{a})^{-2} \Big) } \\
= \; & \sum_{\varpi_2 \in \Pi_{N-\ell}} \prod_{\{a,b\} \in \varpi_2} \frac{1}{(x_{b}-x_{a})^2} 
= \big( \PartF_{\Ising}^{(N-\ell)}(x_{2\ell+1},\ldots, x_{2N}) \big)^2 ,
\end{align*}
where we omitted terms that tend to zero in the limit. This proves the asserted identity~\eqref{eqn::cascade_asy}.
\end{proof}

\subsection{Upper and Lower Bounds --- Proof of Proposition~\ref{prop::Z_total_ineq}}
\label{subsec::proof_Z_total_ineq}

The purpose of this section is to finish the proof of Proposition~\ref{prop::Z_total_ineq}, that is, to show~\eqref{eqn::cascade_asy}:
\begin{align} \label{eq::Z_total_ineq}
\frac{1}{\sqrt{N!}} \; \LB^{(N)}(x_1,\ldots, x_{2N}) \le \PartF_{\Ising}^{(N)}(x_1,\ldots, x_{2N}) \le (2N-1)!! \; \LB^{(N)}(x_1,\ldots, x_{2N}) .
\end{align}

\begin{proof}
The upper bound in~\eqref{eq::Z_total_ineq} follows directly from Corollary~\ref{cor::Ztotal_upper} with $\kappa=3$.
To prove the lower bound in~\eqref{eq::Z_total_ineq}, a slightly more clever calculation is needed, where
the Hafnian identity of Lemma~\ref{lem::Hafnian} plays a crucial role.
The lower bound follows from Lemma~\ref{lem::Z_total_ineq_lower} below.
\end{proof}

\begin{lemma} \label{lem::Z_total_ineq_lower}
For any $N \geq 1$, we have
\begin{align}\label{eqn::Z_total_ineq_lower}
\frac{1}{\sqrt{N!}} \; \LB^{(N)}(x_1,\ldots, x_{2N}) \le \PartF_{\Ising}^{(N)}(x_1,\ldots, x_{2N}) .
\end{align}
\end{lemma}
\begin{proof}
We prove~\eqref{eqn::Z_total_ineq_lower} by induction on $N \geq 1$.
The initial case $N=1$ follows from the equality $\LB^{(1)} = \PartF_{\Ising}^{(1)}$. 
Let then $N \ge 2$ and assume that Equation~\eqref{eqn::Z_total_ineq_lower} holds up to $N-1$. 
By expanding in Lemma~\ref{lem::Hafnian} 
the sum over pair partitions $\varpi \in \Pi_N$ into terms according to the pair $j$ of the last point $2N$, we obtain
\begin{align*}
\big( \PartF_{\Ising}^{(N)}(x_1,\ldots, x_{2N})\big)^2 
= \sum_{j=1}^{2N-1} \frac{1}{(x_{2N}-x_{j})^2} 
\sum_{\hat{\varpi}\in \Pi_{N-1}^{j}} \prod_{\{a,b\} \in \hat{\varpi}} \frac{1}{(x_{b}-x_{a})^2} ,
\end{align*}
where $\hat{\varpi}\in \Pi_{N-1}^{j}$ are pair partitions of $\{1,2,\ldots,2N\} \setminus \{j,2N\}$. 
Again, using Lemma~\ref{lem::Hafnian}, we get
\begin{align*}
\; & \big( \PartF_{\Ising}^{(N)}(x_1,\ldots, x_{2N})\big)^2 
= \sum_{j=1}^{2N-1} \frac{1}{(x_{2N}-x_{j})^{2} }
\big( \PartF_{\Ising}^{(N-1)}(y_1^{j},\ldots, y_{2N-2}^{j}) \big)^2 ,
\end{align*}
where $y_i^{j} := x_i$ for $1 \leq i \leq j-1$ and $y_i^{j} := x_{i+1}$ for $j \leq i \leq 2N-2$.
Now,  we have
\begin{align*}
\; & \bigg( \frac{\PartF_{\Ising}^{(N)}(x_1,\ldots, x_{2N})}{\LB^{(N)}(x_1,\ldots, x_{2N})} \bigg)^2\\
= \; & \sum_{j=1}^{2N-1}
\frac{1}{(x_{2N}-x_{j})^{2}} \Bigg( \frac{\PartF_{\Ising}^{(N-1)}(y_1^{j},\ldots, y_{2N-2}^{j})}{\LB^{(N)}(x_1,\ldots, x_{2N})} \Bigg)^2 \\
= \; & \sum_{j=1}^{2N-1}
\frac{1}{(x_{2N}-x_{j})^{2}} \Bigg( 
\frac{\PartF_{\Ising}^{(N-1)}(y_1^{j},\ldots, y_{2N-2}^{j})}{\LB^{(N-1)}(y_1^{j},\ldots, y_{2N-2}^{j})} \Bigg)^2
\Bigg( \frac{\LB^{(N-1)}(y_1^{j},\ldots, y_{2N-2}^{j})}{\LB^{(N)}(x_1,\ldots, x_{2N})}  \Bigg)^2 \\
\geq \; & \frac{1}{(N-1)!} \sum_{j=1}^{2N-1}
\frac{1}{(x_{2N}-x_{j})^{2}} \Bigg( \frac{\LB^{(N-1)}(y_1^{j},\ldots, y_{2N-2}^{j})}{\LB^{(N)}(x_1,\ldots, x_{2N})}  \Bigg)^2 .
\end{align*}
Recalling Lemma~\ref{lem: compare Bs} and separating the sum according to the parity of $j$, we obtain
\begin{align*}
\; & \sum_{j=1}^{2N-1}
\frac{1}{(x_{2N}-x_{j})^{2}} \Bigg( \frac{\LB^{(N-1)}(y_1^{j},\ldots, y_{2N-2}^{j})}{\LB^{(N)}(x_1,\ldots, x_{2N})}  \Bigg)^2 \\
= \; & \sum_{\substack{1 \leq j < 2N , \\ j \text{ is even}}} \frac{1}{(x_{2N}-x_{j})^{4} }
\prod_{\substack{1 \leq i < 2N , \\ i \neq j}} (x_i - x_j)^{2(-1)^{i+1}} 
\prod_{\substack{1 \leq i < 2N , \\ i \neq j}} (x_{2N}-x_i)^{2(-1)^{i+1}}
\prod_{\substack{1\leq k < j , \\ j \leq l < 2N-1}} (y_l^{j} - y_k^{j})^{4(-1)^{l-k}} \\
\; & + \sum_{\substack{1 \leq j < 2N , \\ j \text{ is odd}}} \;
\prod_{\substack{1 \leq i < 2N , \\ i \neq j}} (x_i - x_j)^{2(-1)^{i}} 
\prod_{\substack{1 \leq i < 2N , \\ i \neq j}} (x_{2N}-x_i)^{2(-1)^{i+1}} 
\prod_{\substack{1\leq k < j , \\ j \leq l < 2N-1}} (y_l^{j} - y_k^{j})^{4(-1)^{l-k}} .
\end{align*}
Because all terms are non-negative, dropping the sum over even $j$ yields the rough lower bound
\begin{align} \label{eq:auxiliary}
\begin{split}
\; & \bigg( \frac{\PartF_{\Ising}^{(N)}(x_1,\ldots, x_{2N})}{\LB^{(N)}(x_1,\ldots, x_{2N})} \bigg)^2 \\
\geq \; & \frac{1}{(N-1)!} \sum_{j=1}^{2N-1}
\frac{1}{(x_{2N}-x_{j})^{2}} \Bigg( \frac{\LB^{(N-1)}(y_1^{j},\ldots, y_{2N-2}^{j})}{\LB^{(N)}(x_1,\ldots, x_{2N})}  \Bigg)^2 \\
\geq \; & \frac{1}{(N-1)!} 
\sum_{\substack{1 \leq j < 2N , \\ j \text{ is odd}}} \;
\prod_{\substack{1 \leq i < 2N , \\ i \neq j}} \bigg( \frac{x_i - x_j}{x_{2N}-x_i} \bigg)^{2(-1)^{i}}
\prod_{\substack{1\leq k < j , \\ j \leq l < 2N-1}} (y_l^{j} - y_k^{j})^{4(-1)^{l-k}} .
\end{split}
\end{align}
Now, using Remark~\ref{rem::crossratio}, we see that the last product is in fact larger than one:
\begin{align*}
\prod_{\substack{1\leq k < j , \\ j \leq l < 2N-1}} |y_l^{j} - y_k^{j}|^{(-1)^{l-k}} 
= \prod_{\substack{1 \leq m \leq \frac{1}{2}(j-1) , \\ \frac{1}{2}(j+1) \leq n \leq N-1}} 
\Bigg| \frac{(y_{2\ell-1}^{j} - y_{2m-1}^{j})(y_{2\ell}^{j} - y_{2m}^{j})}{(y_{2\ell-1}^{j} - y_{2m}^{j})(y_{2\ell}^{j} - y_{2m-1}^{j})} \Bigg| \geq 1 .
\end{align*}
Therefore,~\eqref{eq:auxiliary} gives
\begin{align*} 
\bigg( \frac{\PartF_{\Ising}^{(N)}(x_1,\ldots, x_{2N})}{\LB^{(N)}(x_1,\ldots, x_{2N})} \bigg)^2 
\geq \frac{1}{(N-1)!} 
\sum_{\substack{1 \leq j < 2N , \\ j \text{ is odd}}} \;
\prod_{\substack{1 \leq i < 2N , \\ i \neq j}} \bigg( \frac{x_i - x_j}{x_{2N}-x_i} \bigg)^{2(-1)^{i}} .
\end{align*}
Finally, we note that
\begin{align*}
\prod_{\substack{1 \leq i < 2N , \\ i \neq j}} \bigg| \frac{x_i - x_j}{x_i - x_{2N}} \bigg|^{(-1)^{i}} 
= P^{(j,2N)}(x_1, \ldots, x_{2N}) \in (0,1),
\end{align*}
where $P^{(j,2N)}$ is the probability in~\eqref{eqn::levellines_proba_lk}. In particular, we have
\begin{align*}
\sum_{\substack{1 \leq j < 2N , \\ j \text{ is odd}}} \;
\prod_{\substack{1 \leq i < 2N , \\ i \neq j}} \bigg| \frac{x_i - x_j}{x_i - x_{2N}} \bigg|^{(-1)^{i}} 
= \sum_{\substack{1 \leq j < 2N , \\ j \text{ is odd}}} P^{(j,2N)}(x_1, \ldots, x_{2N}) = 1 .
&& \text{[by~\eqref{eqn::levellines_proba_lk}]}
\end{align*}
Using the Cauchy-Schwarz inequality, we eventually get the following lower bound for the square:
\begin{align*} 
\bigg( \frac{\PartF_{\Ising}^{(N)}(x_1,\ldots, x_{2N})}{\LB^{(N)}(x_1,\ldots, x_{2N})} \bigg)^2 
\geq \; & \frac{1}{(N-1)!} \sum_{\substack{1 \leq j < 2N , \\ j \text{ is odd}}} \big( P^{(j,2N)}(x_1, \ldots, x_{2N}) \big)^2 \\
\geq \; & \frac{1}{(N-1)!} \frac{1}{N} \Bigg( \sum_{\substack{1 \leq j < 2N , \\ j \text{ is odd}}} P^{(j,2N)}(x_1, \ldots, x_{2N}) \Bigg)^2 = \frac{1}{N!} .
\end{align*}
Because by~\eqref{eq::TotalPartPos} and~\eqref{eqn::b_def}, we have 
$\PartF_{\Ising} > 0$ and $\LB > 0$, this gives the lower bound~\eqref{eqn::Z_total_ineq_lower}.
\end{proof}

\subsection{Boundary Behavior}
\label{subsec::technical_Z_total}

The final result of this section  
concerns the boundary behavior of the ratios 
$\PartF_\alpha/\PartF_{\Ising}$ of partition functions
when the variables move under a Loewner evolution. 
It is crucial for proving Theorem~\ref{thm::ising_crossingproba} in Section~\ref{sec::Ising_crossing_proba}.

\begin{proposition}\label{prop::mart_wrong_zero}
Fix $\kappa = 3$,
$\alpha =\{ \link{a_1}{b_1}, \ldots, \link{a_N}{b_N} \} \in \LP_N$ and suppose that $\link{1}{2} \in \alpha$.
Fix an index $\ell \in \{2,3, \ldots, N\}$ and real points $x_1 < \cdots < x_{2N}$.
Suppose $\eta$ is a continuous simple curve in $\HH$ starting from $x_1$ and terminating at $x_{2\ell}$ at time $T$, 
which hits $\R$ only at $\{x_1, x_{2\ell}\}$. Let $(W_t, 0\le t\le T)$ be its Loewner driving function and $(g_t, 0\le t\le T)$ 
the corresponding conformal maps. Then, we have
\begin{align} \label{eqn::mart_wrong_zero}
\lim_{t \to T} 
\frac{\PartF_\alpha (W_t, g_t(x_2), \ldots, g_t(x_{2N})) }{\PartF^{(N)}_{\Ising} (W_t, g_t(x_2), \ldots, g_t(x_{2N})) }
= 0 .
\end{align}
\end{proposition}

\begin{proof}
By~\eqref{eqn::optimal_bound_polygon} and~\eqref{eqn::Z_total_ineq}, we have
\begin{align*}
\frac{\PartF_\alpha}{\PartF^{(N)}_{\Ising}}\le \frac{\LB_{\alpha}}{\PartF^{(N)}_{\Ising}}
\le \sqrt{N!}\frac{\LB_{\alpha}}{\LB^{(N)}} ,
\end{align*}
and $\LB_{\alpha}/\LB^{(N)} \to 0$ in the limit of~\eqref{eqn::mart_wrong_zero} by~\cite[Proposition~B.1]{PeltolaWuGlobalMultipleSLEs}.
\end{proof}

\section{Loewner Chains Associated to Partition Functions When $\kappa = 3$}
\label{sec::Loewner_chain}
Fix $j \in \{1,2,\ldots,2N\}$. Recall from Sections~\ref{subsec::pre_SLE} 
and~\ref{subsec::pre_partitionfucntions} that, for launching points $x_1<\cdots<x_{2N}$ on $\R = \partial \HH$, 
the Loewner chain associated to an $\SLE_{\kappa}$ partition function $\PartF$ 
starting from $x_j$ 
is the process with driving function $W$ given by the SDEs~\eqref{eqn::loewnerchain_partition}. 
In this section, we specialize to the case where $\kappa=3$ 
and consider the Loewner chain associated to 
$\PartF_{\Ising}$ starting from $x_j$. 
The Loewner chain is well-defined up to the first time 
when it swallows a spectator point:
\begin{align} \label{eq:total_swallowing_time}
T = T(j) := \min_{i\neq j} \swal_{x_i} \in (0, +\infty] .
\end{align} 
We will see that the Loewner chain is generated by a continuous curve $\eta$ up to the swallowing time $T$, and it terminates on the real line at time $T$, that is, $\eta(T) \in \R$. 
However, the continuity of the curve as it approaches the swallowing time is difficult to prove in general. 
The main purpose of this section is to analyze the behavior of the Loewner chain when approaching the swallowing time and to prove the continuity 
\textit{up to and including $T$}. 
The main result of this section is Theorem~\ref{thm::loewner_Ztotal_continuity} stated below, which we prove in Section~\ref{subsec::TheoremProof}.
Proposition~\ref{prop::Z_total_ineq} plays an important role in the proof.

\smallbreak

To begin, we fix notation to be used throughout.
If $(\Omega; x_1, \ldots, x_{2N})$ is a polygon and $\gamma$ a simple curve from $x_a$ to $x_b$ 
that only touches $\partial \Omega$ at its endpoints, we denote by $D_{\gamma}^R$ (resp.~$D_{\gamma}^L$)
the connected component of $\Omega\setminus\gamma$ having the counterclockwise boundary arc $(x_{a+1} \, x_{b-1})$ 
(resp.~$(x_{b+1} \, x_{a-1})$) on its boundary. Also, for each $j \in \{1,\ldots,2N\}$, 
we denote by $I_j := \{\ldots, j-3, j-1, j+1, j+3, \ldots\}$ the set of indices in $\{1,\ldots,2N\}$ that have different parity than $j$.

\begin{theorem} \label{thm::loewner_Ztotal_continuity}
Fix $\kappa=3$, $N \geq 1$, and $j \in \{1,2,\ldots,2N\}$. 
The Loewner chain associated to $\PartF_{\Ising}$ with launching points $(x_1, \ldots, x_{2N})$ starting from $x_j$
is almost surely generated by a continuous simple curve
$(\eta(t), 0\le t\le T)$ up to and including $T$. 
This curve almost surely terminates at one of the points 
$\{x_k \colon k \in I_j\}$ and touches $\R$ only at its two endpoints. 
Furthermore, for any $k \in I_j$, the probability of 
terminating at $x_k$ is given by 
\begin{align}\label{eqn::Loewner_Ztotal_proba}
\PP[\eta(T)=x_k] 
= \sum_{\alpha\in\LP_N} \frac{\PartF_{\alpha}(x_1, \ldots, x_{2N})}{\PartF_{\Ising}^{(N)}(x_1,\ldots, x_{2N})} \, \one \{ \link{j}{k} \in \alpha\} . 
\end{align}
Moreover, conditionally on the event $\{\eta(T)=x_k\}$, the law of $\eta$ is that of the $\SLE_3$ curve $\gamma$ in 
$(\HH;x_j,x_k)$ 
weighted by the Radon-Nikodym derivative
\begin{align}\label{eqn::Loewner_Ztotal_RN}
\frac{\PartF_{\Ising}^{(\ell)}(D_{\gamma}^R; x_{j+1}, \ldots, x_{k-1})\times\PartF_{\Ising}^{(N-\ell-1)}(D_{\gamma}^L; x_{k+1}, \ldots, x_{j-1})}{|x_j-x_k| \; \underset{\alpha\in\LP_N}{\sum} \PartF_{\alpha}(x_1, \ldots, x_{2N}) 
\, \one \{ \link{j}{k} \in \alpha \} } ,
\end{align}
where $\ell = \ell(j,k) = \frac{1}{2}(|j - k| - 1)$. 
\end{theorem}

The outline of this section is as follows. 
In Section~\ref{subsec:cascade}, we decompose the symmetric partition function $\PartF_{\Ising}$ into two parts: one for the relevant left component $D_{\gamma}^L$ and one for the relevant right component $D_{\gamma}^R$ of $\Omega\setminus\gamma$ (see Lemma~\ref{lem::Ztotal_CAS}).
This will be used to find the Radon-Nikodym derivative~\eqref{eqn::Loewner_Ztotal_RN}.
In Section~\ref{subsec::TheoremProof}, we use several lemmas to prove 
Theorem~\ref{thm::loewner_Ztotal_continuity}:
Lemmas~\ref{lem::Ztotal_continuity_ind}--\ref{lem::Ztotal_RN_ind} treat the event that the curve $\eta$ accumulates in the open interval $(x_1, x_{2N-1})$, and  
Lemma~\ref{lem::Ztotal_continuity_ind2} the event that $\eta$ accumulates in $ [x_{2N-1}, \infty) \cup (-\infty, x_1)$. 
Corollary~\ref{cor::Ztotal_continuity_ind2} summarizes the induction step, which concludes the proof of Theorem~\ref{thm::loewner_Ztotal_continuity}.

\subsection{Cascade Relation for Partition Functions}
\label{subsec:cascade}

Fix a nice polygon $(\Omega; x_1, \ldots, x_{2N})$ and $\alpha \in \LP_N$.
Given any link $\link{a}{b} \in \alpha$, let $\gamma$ be the $\SLE_{3}$ curve in $\Omega$ from $x_{a}$ to $x_{b}$,
and assume that $a < b$ for notational simplicity.
Then, the link $\link{a}{b}$ divides the link pattern $\alpha$ into two sub-link patterns, connecting respectively the points
$\{a+1, \ldots, b-1\}$ and $\{b+1, \ldots, a-1\}$. After relabeling of the indices, we denote these two link patterns by $\alpha^R$ 
and $\alpha^L$. Then, we have the following cascade relation for the pure partition functions \cite[Proposition~3.5]{PeltolaWuGlobalMultipleSLEs}: 
\begin{align} 
\label{eqn::purepartition_CAS}
\PartF_{\alpha}(\Omega; x_1, \ldots, x_{2N}) 
= H_{\Omega}( x_{a}, x_{b})^{1/2} \, 
\E\big[\PartF_{\alpha^R}(D_{\gamma}^R; x_{a+1}, \ldots, x_{b-1})\times\PartF_{\alpha^L}(D_{\gamma}^L; x_{b+1},\ldots, x_{a-1})\big] .
\end{align}
As a consequence, we obtain a similar cascade relation for the symmetric partition function $\PartF_{\Ising}$. 

\begin{lemma} \label{lem::Ztotal_CAS}
Let $(\Omega; x_1, \ldots, x_{2N})$ be a nice polygon.
Fix $a,b \in \{1,2, \ldots, 2N\}$
with different parities, and let $\gamma$ be the $\SLE_3$ curve in $\Omega$ from $x_a$ to $x_b$. 
Then, we have 
\begin{align}
\label{eqn::Ztotal_CAS}
\begin{split}
\; & \sum_{\alpha\in\LP_N}\PartF_{\alpha}(\Omega; x_1, \ldots, x_{2N}) 
\, \one \{ \link{a}{b} \in \alpha \} \\
= \; & H_{\Omega}(x_a, x_b)^{1/2} \; \E\big[\PartF_{\Ising}^{(\ell)}(D_{\gamma}^R; x_{a+1}, \ldots, x_{b-1})\times\PartF_{\Ising}^{(N-\ell-1)}(D_{\gamma}^L; x_{b+1}, \ldots, x_{a-1})\big],
\end{split}
\end{align}
where $\ell = \ell(a,b) = \frac{1}{2}(|a - b| - 1)$.  
\end{lemma}

Notice that for fixed $(\Omega; x_1, \ldots, x_{2N})$, the partition functions $\PartF_{\Ising}^{(\ell)}$ and $\PartF_{\Ising}^{(N-\ell-1)}$ in~\eqref{eqn::Ztotal_CAS} 
are bounded uniformly in the random domains $D_{\gamma}^R$ and $D_{\gamma}^L$, respectively:  
Proposition~\ref{prop::Z_total_ineq} and Remark~\ref{rem::btotal_mono} give 
\begin{align*}
\PartF_{\Ising}^{(\ell)}(D_{\gamma}^R; x_{a+1}, \ldots, x_{b-1}) 
\le \; & (2\ell-1)!! \; \LB^{(\ell)}(D_{\gamma}^R; x_{a+1}, \ldots, x_{b-1}) 
&&\text{[by~\eqref{eqn::Z_total_ineq}]} \\
\le \; & ((2\ell-1)!!)^3 \; \LB^{(\ell)}(\Omega;x_{a+1}, \ldots, x_{b-1}) , 
&&\text{[by~\eqref{eqn::btotal_mono}]} \\[1em]
\PartF_{\Ising}^{(N-\ell-1)}(D_{\gamma}^L; x_{b+1}, \ldots, x_{a-1})
\le \; & (2N-2\ell-3)!! \; \LB^{(N-\ell-1)}(D_{\gamma}^L; x_{b+1}, \ldots, x_{a-1}) 
&&\text{[by~\eqref{eqn::Z_total_ineq}]} \\
\le \; & ((2N-2\ell-3)!!)^3 \; \LB^{(N-\ell-1)}(\Omega;x_{b+1}, \ldots, x_{a-1}) .
&&\text{[by~\eqref{eqn::btotal_mono}]}
\end{align*}

\begin{proof}[Proof of Lemma~\ref{lem::Ztotal_CAS}]
Without loss of generality, we may assume that $a < b$.
The identity~\eqref{eqn::Ztotal_CAS} then follows by summing over all possible $\alpha^L$ and $\alpha^R$
in the cascade relation~\eqref{eqn::purepartition_CAS}:
\begin{align*}
& \;  \sum_{\alpha\in\LP_N}\PartF_{\alpha}(\Omega; x_1, \ldots, x_{2N})\, \one \{ \link{a}{b} \in \alpha \} \\
= & \; H_{\Omega}(x_a, x_b)^{1/2} \;
 \E \bigg[\bigg(\sum_{\alpha^R\in\LP_{\ell}}\PartF_{\alpha^R}(D_{\gamma}^R; x_{a+1}, \ldots, x_{b-1})\bigg)  \bigg(\sum_{\alpha^L\in\LP_{N-\ell-1}}\PartF_{\alpha^L}(D_{\gamma}^L; x_{b+1},\ldots, x_{a-1})\bigg)\bigg] \\
= & \; H_{\Omega}(x_a, x_b)^{1/2} \; \E\big[\PartF_{\Ising}^{(\ell)}(D_{\gamma}^R; x_{a+1}, \ldots, x_{b-1})\times\PartF_{\Ising}^{(N-\ell-1)}(D_{\gamma}^L; x_{b+1}, \ldots, x_{a-1})\big] .
\qedhere
\end{align*} 
\end{proof}

\subsection{Continuity of the Loewner Chain --- Proof of Theorem~\ref{thm::loewner_Ztotal_continuity}}
\label{subsec::TheoremProof}

We prove Theorem~\ref{thm::loewner_Ztotal_continuity} by induction on $N \geq 1$. 
There is nothing to prove in the initial case $N=1$, so we assume that $N \geq 2$ and that 
Theorem~\ref{thm::loewner_Ztotal_continuity} holds for all $j \in \{1,2,\ldots,2(N-1)\}$. 
For definiteness, we also assume that $j=1$ and prove Theorem~\ref{thm::loewner_Ztotal_continuity} for the
Loewner chain associated to $\PartF_{\Ising}$ with launching points $(x_1, \ldots, x_{2N})$ starting from $x_1$.
We denote by $(g_t, t \geq 0)$ the corresponding conformal maps, and we denote the Loewner chain by $\eta$ (with suggestive notation anticipating that it will be generated by a curve).
We break the proof into several lemmas, each addressing a part of the claim. 

\begin{lemma}\label{lem::Ztotal_continuity_ind}
Assume that Theorem~\ref{thm::loewner_Ztotal_continuity} holds for $N-1$. 
Define $\LE_{2N-1}$ to be the event that the
Loewner chain $\eta$ associated to $\PartF_{\Ising}$
with launching points $(x_1, \ldots, x_{2N})$ starting from $x_1$ accumulates in the open interval $(x_1, x_{2N-1})$. 
Then, on the event $\LE_{2N-1}$, the Loewner chain $\eta$ is almost surely generated by 
a continuous curve up to and including $T$.
This curve almost surely terminates at one of the points $\{x_2, x_4, \ldots, x_{2N-2}\}$ and touches $\R$ only at its two endpoints. 
\end{lemma}

\begin{proof}
Let $\hat{\eta}$ be the Loewner chain associated to $\PartF_{\Ising}^{(N-1)}$ with launching points $(x_1, \ldots, x_{2N-2})$ starting from $x_1$. 
From the discussion after~\eqref{eqn::loewnerchain_partition}, we see that 
the law of $\eta$ is the same as the law of $\hat{\eta}$ tilted by the following local martingale, for small enough $t$: 
\begin{align} \label{eq::MN}
R_t
:= g_t'(x_{2N-1})^{1/2}g_t'(x_{2N})^{1/2}
\; \frac{\PartF_{\Ising}^{(N)}(W_t,g_t(x_2), \ldots, g_t(x_{2N}))}{\PartF_{\Ising}^{(N-1)}(W_t,g_t(x_2), \ldots, g_t(x_{2N-2}))}.
\end{align}
\begin{itemize}
\item By the bounds~\eqref{eqn::Z_total_ineq} and~\eqref{eqn::btotal_compare}, we have 
\begin{align*}
R_t 
\le & \; ((2N-1)!!)^2 \; g_t'(x_{2N-1})^{1/2} g_t'(x_{2N})^{1/2} \; 
\frac{\LB^{(N)}(W_t,g_t(x_2), \ldots, g_t(x_{2N}))}{\LB^{(N-1)}(W_t,g_t(x_2), \ldots, g_t(x_{2N-2}))}
&& \text{[by~\eqref{eqn::Z_total_ineq}]} \\
\le & \; ((2N-1)!!)^2 \; \bigg(\frac{g_t'(x_{2N-1})^{1/2} g_t'(x_{2N})^{1/2}}{g_t(x_{2N})-g_t(x_{2N-1})} \bigg) 
\bigg(\frac{g_t(x_{2N})-g_t(x_{2N-2})}{g_t(x_{2N-1})-g_t(x_{2N-2})}\bigg) .
&& \text{[by~\eqref{eqn::btotal_compare}]} 
\end{align*}
\item By the monotonicity property~\eqref{eqn::poissonkernel_mono}, we have
\begin{align*}
\frac{g_t'(x_{2N-1})^{1/2}g_t'(x_{2N})^{1/2}}{g_t(x_{2N})-g_t(x_{2N-1})}\le \frac{1}{x_{2N}-x_{2N-1}} .
\end{align*}
\item The SDEs~\eqref{eqn::loewnerchain_partition} show that $g_t(x_{2N})-g_t(x_{2N-2})$ is decreasing in $t$, so 
\begin{align*}
g_t(x_{2N})-g_t(x_{2N-2})\le x_{2N}-x_{2N-2} .
\end{align*}
\end{itemize}
In conclusion, we obtain the following bound for $R$:
\begin{align*}
0 < R_t
\le ((2N-1)!!)^2 \bigg(\frac{x_{2N}-x_{2N-2}}{x_{2N}-x_{2N-1}}\bigg) \bigg(\frac{1}{g_t(x_{2N-1})-g_t(x_{2N-2})}\bigg) .
\end{align*}
From this, we see that for any $\eps>0$, the local martingale $R_t$
is bounded (and hence a true martingale) up to the stopping time
\begin{align} \label{def::induction_stoppingtime}
S_{\eps} := \inf\{t > 0 \colon g_t(x_{2N-1})-g_t(x_{2N-2})\le \eps\}.
\end{align}
Hence, for any $\eps>0$, the law of $\eta$ is absolutely continuous with respect to the law of $\hat{\eta}$ up to time~$S_{\eps}$. 
Therefore, on the event $\LE_{2N-1}$ that $\eta$ accumulates in the interval $(x_1, x_{2N-1})$, the law of $\eta$ is absolutely continuous 
with respect to the law of $\hat{\eta}$. 
Thus, the claim follows by the induction hypothesis on $\hat{\eta}$.
\end{proof}

\begin{lemma}\label{lem::Ztotal_proba_ind}
Assume the same setup as in Lemma~\ref{lem::Ztotal_continuity_ind}. 
Then, for any $\ell \in \{1,2,\ldots, N-1\}$, we have
\begin{align}\label{eqn::Ztotal_proba_ind}
\PP[\eta(T)=x_{2\ell}]=\sum_{\alpha\in\LP_N}\frac{\PartF_{\alpha}(x_1, \ldots, x_{2N})}{\PartF_{\Ising}^{(N)}(x_1,\ldots, x_{2N})} 
\, \one \{ \link{1}{2\ell} \in \alpha \}  .
\end{align}
\end{lemma}

\begin{proof}
Using the same notation as in the proof of Lemma~\ref{lem::Ztotal_continuity_ind},
we recall that $\eta$ has the law of $\hat{\eta}$ tilted by the local martingale~\eqref{eq::MN},
which we write in the form
\begin{align*}
R_t
= g_t'(x_{2N-1})^{1/2}g_t'(x_{2N})^{1/2} \;
\bigg( \frac{\PartF_{\Ising}^{(N)}(W_t,g_t(x_2), \ldots, g_t(x_{2N}))}{\PartF_{\Ising}^{(\ell)}(W_t, g_t(x_2), \ldots, g_t(x_{2\ell}))} \bigg)
\bigg( \frac{\PartF_{\Ising}^{(\ell)}(W_t, g_t(x_2), \ldots, g_t(x_{2\ell}))}{\PartF_{\Ising}^{(N-1)}(W_t,g_t(x_2), \ldots, g_t(x_{2N-2}))} \bigg) .
\end{align*}
By Lemma~\ref{lem::Ztotal_continuity_ind}, on the event $\{\eta(T)=x_{2\ell}\}$, the curve $\eta$ is continuous up to and including $T$.
Hence, Proposition~\ref{prop::Z_total_asy_refined} and the conformal covariance~\eqref{eq:conformal_image}
show that as $t\to T$, we have 
\begin{align*}
R_T := \lim_{t \to T} R_t 
= \; & g_T'(x_{2N-1})^{1/2}g_T'(x_{2N})^{1/2} \; \frac{\PartF_{\Ising}^{(N-\ell)}(g_T(x_{2\ell+1}), \ldots, g_T(x_{2N}))}{\PartF_{\Ising}^{(N-1-\ell)}(g_T(x_{2\ell+1}), \ldots, g_T(x_{2N-2}))} \\
= \; & \frac{\PartF_{\Ising}^{(N-\ell)}(D_{\hat{\eta}}^L; x_{2\ell+1}, \ldots, x_{2N})}{\PartF_{\Ising}^{(N-1-\ell)}(D_{\hat{\eta}}^L; x_{2\ell+1}, \ldots, x_{2N-2})} ,
\qquad \text{almost surely}.
\end{align*} 
Furthermore, the convergence also holds in $L^1(\hat{\PP}[\cdot\cond \hat{\eta}(T)=x_{2\ell}])$,
where $\hat{\PP}[\cdot\cond \hat{\eta}(T)=x_{2\ell}]$ denotes the law of $\hat{\eta}$ conditioned on the event $\{\hat{\eta}(T)=x_{2\ell}\}$,
because $R_{t\wedge T}$ is a positive and uniformly integrable martingale under the measure $\hat{\PP}[\cdot\cond \hat{\eta}(T)=x_{2\ell}]$
for the following reason. 
For $m \ge 1$, denote by $\hat{\PP}_m^*$ the law of $\hat{\eta}$ tilted 
by the martingale $\smash{R_{t \wedge T \wedge S_{1/m}}}$, 
where $\smash{S_{1/m}}$ is the stopping time defined in~\eqref{def::induction_stoppingtime} 
(so that $\smash{R_{t \wedge T \wedge S_{1/m}}}$ is bounded for all $t$). 
Then, $\hat{\PP}_m^*$ is the same as the law of $\eta$ up to time $\smash{S_{1/m}}$. 
On the one hand, the sequences $\hat{\PP}_m^*$ are manifestly  consistent in $m$, 
so by Kolmogorov's extension theorem, 
there exists a probability measure $\hat{\PP}^*$ such that  under $\hat{\PP}^*$ 
the curve has the same law as $\eta$ up to time $\smash{S_{1/m}}$, for any $m\ge 1$. 
On the other hand, by Lemma~\ref{lem::Ztotal_continuity_ind} 
the Loewner chain $\eta$ is almost surely continuous up to and including $T$ 
and only touches $\R$ at $x_1$ and $x_{2\ell}$. 
Thus, on the event $\{\eta(T)=x_{2\ell}\}$, there is almost surely a positive distance between $\eta[0,T]$ 
and the points $\{x_{2\ell+1}, x_{2\ell+2}, \ldots, x_{2N}\}$.  
In particular, almost surely on the event $\{\eta(T)=x_{2\ell}\}$, we have
\begin{align*}
g_t(x_{2N-1})-g_t(x_{2N-2}) \, \ge \, 
g_T(x_{2N-1})-g_T(x_{2N-2}) \, > \, 0 , 
\qquad \text{for all } t\in [0,T].
\end{align*}
It follows that (see, e.g.,~\cite[Theorem~5.3.3]{Durrett}) 
the law of $\eta$ conditioned on $\{\eta(T)=x_{2\ell}\}$ is the same as $\hat{\PP}^*$ conditioned on the same event, which is the same as $\hat{\PP}[\cdot\cond \hat{\eta}(T)=x_{2\ell}]$  tilted by $R_{t \wedge T}$. 
This gives the uniform integrability of the martingale $R_{t \wedge T}$. 
With the $L^1$-convergence, we conclude that 
\begin{align} \label{eqn::Ztotal_continuity_aux1}
\begin{split} 
\PP[\eta(T)=x_{2\ell}] 
= \; & \frac{1}{R_0} \, \hat{\E} \big[ \one \{\hat{\eta}(T)=x_{2\ell}\} R_T \big] \\
= \; &
\frac{1}{R_0}
\, \hat{\E} \bigg[ \one \{\hat{\eta}(T)=x_{2\ell}\} \frac{\PartF_{\Ising}^{(N-\ell)}(D_{\hat{\eta}}^L; x_{2\ell+1}, \ldots, x_{2N})}{\PartF_{\Ising}^{(N-1-\ell)}(D_{\hat{\eta}}^L; x_{2\ell+1}, \ldots, x_{2N-2})} \bigg].
\end{split} 
\end{align}
Now, by the induction hypothesis on $\hat{\eta}$, on the event $\{\hat{\eta}(T)=x_{2\ell}\}$, 
the law $\hat{\PP}$ of $\hat{\eta}$ is that of the $\SLE_3$ curve $\gamma$ in $\HH$ from $x_1$ to $x_{2\ell}$ weighted by the Radon-Nikodym derivative
\begin{align} \label{eqn::Ztotal_contiuity_aux2}
\frac{\PartF_{\Ising}^{(\ell-1)}(D_{\gamma}^R; x_2, \ldots, x_{2\ell-1})\times\PartF_{\Ising}^{(N-1-\ell)}(D_{\gamma}^L; x_{2\ell+1}, \ldots, x_{2N-2})}{(x_{2\ell}-x_1) 
\underset{\beta\in\LP_{N-1}}{\sum} \PartF_{\beta}(x_1, \ldots, x_{2N-2})\, \one \{ \link{1}{2\ell} \in \beta \} }.
\end{align}
Combining~\eqref{eqn::Ztotal_contiuity_aux2} with~\eqref{eqn::Ztotal_continuity_aux1}, we obtain\footnote{Here, we denote by $\mathsf{E} = \mathsf{E}(\HH;x_1, x_{2\ell})$ the expectation corresponding to the law $\mathsf{P} = \mathsf{P}(\HH;x_1, x_{2\ell})$ of $\gamma$. } 
\begin{align*}
& \; \PP[\eta(T)=x_{2\ell}] \\
= & \; \frac{1}{R_0}
\; \hat{\E} \bigg[\frac{\PartF_{\Ising}^{(N-\ell)}(D_{\hat{\eta}}^L; x_{2\ell+1}, \ldots, x_{2N})}{\PartF_{\Ising}^{(N-1-\ell)}(D_{\hat{\eta}}^L; x_{2\ell+1}, \ldots, x_{2N-2})} \; \Big| \;  {\hat{\eta}}(T)=x_{2\ell}\bigg]
\hat{\PP}[\hat{\eta}(T)=x_{2\ell}] \\
= & \;  
\frac{1}{R_0}
\; \frac{\mathsf{E} \big[\PartF_{\Ising}^{(N-\ell)}(D_{\gamma}^L; x_{2\ell+1}, \ldots, x_{2N})\times\PartF_{\Ising}^{(\ell-1)}(D_{\gamma}^R; x_2, \ldots, x_{2\ell-1})\big]}{(x_{2\ell}-x_1)\underset{\beta\in\LP_{N-1}}{\sum} \PartF_{\beta}(x_1, \ldots, x_{2N-2}) \, \one \{ \link{1}{2\ell} \in \beta \} } \; \hat{\PP}[\hat{\eta}(T)=x_{2\ell}] .
\end{align*}
Now, we use Lemma~\ref{lem::Ztotal_CAS} with $a=2\ell$ and $b=1$ to evaluate the numerator, obtaining
\begin{align*}
\PP[\eta(T)=x_{2\ell}] & =
\frac{1}{R_0}
\; \frac{\underset{\alpha\in\LP_N}{\sum} \PartF_{\alpha}(x_1, \ldots, x_{2N}) 
\, \one \{ \link{1}{2\ell} \in \alpha \} }{\underset{\beta\in\LP_{N-1}}{\sum} \PartF_{\beta}(x_1, \ldots, x_{2N-2} )\, \one \{ \link{1}{2\ell} \in \beta \} } \; \hat{\PP}[\hat{\eta}(T)=x_{2\ell}].
\end{align*}
On the other hand, by the induction hypothesis on $\hat{\eta}$, we know that 
\begin{align}\label{eqn::Ztotal_continuity_aux3}
\hat{\PP}[\hat{\eta}(T)=x_{2\ell}] = \sum_{\beta\in\LP_{N-1}} \frac{\PartF_{\beta}(x_1, \ldots, x_{2N-2}) }{\PartF_{\Ising}^{(N-1)}(x_1, \ldots, x_{2N-2})} \, \one \{ \link{1}{2\ell} \in \beta \} ,
\end{align}
and therefore, with $R_0 = \smash{\PartF_{\Ising}^{(N)} / \PartF_{\Ising}^{(N-1)}}$
we find the probability~\eqref{eqn::Ztotal_proba_ind} of interest. 
\end{proof}

\begin{lemma}\label{lem::Ztotal_RN_ind}
Assume the same setup as in Lemma~\ref{lem::Ztotal_continuity_ind}. 
Then, for all $\ell \in \{1,2, \ldots, N-1\}$, conditionally on the event $\{\eta(T)=x_{2\ell}\}$, the law of $\eta$ 
is that of the $\SLE_3$ curve $\gamma$ in $\HH$ from $x_1$ to $x_{2\ell}$ weighted by the Radon-Nikodym derivative
\begin{align}\label{eqn::Ztotal_RN_ind}
\frac{\PartF_{\Ising}^{(\ell-1)}(D_{\gamma}^R; x_2, \ldots, x_{2\ell-1})\times\PartF_{\Ising}^{(N-\ell)}(D_{\gamma}^L; x_{2\ell+1}, \ldots, x_{2N})}{(x_{2\ell}-x_1) \underset{\alpha\in\LP_{N}}{\sum} \PartF_{\alpha}(x_1, \ldots, x_{2N}) \, \one \{ \link{1}{2\ell} \in \alpha \} } .
\end{align}
\end{lemma}

\begin{proof}
We still use the same notation as in the proofs of Lemmas~\ref{lem::Ztotal_continuity_ind} \&~\ref{lem::Ztotal_proba_ind}.
On the event $\{\eta(T)=x_{2\ell}\}$, the law of $\eta$ is the same as the law of ${\hat{\eta}}$ weighted by the Radon-Nikodym derivative
\begin{align*}
\frac{1}{R_0}
\; \frac{\PartF_{\Ising}^{(N-\ell)}(D_{\hat{\eta}}^L; x_{2\ell+1}, \ldots, x_{2N})}{\PartF_{\Ising}^{(N-1-\ell)}(D_{\hat{\eta}}^L; x_{2\ell+1}, \ldots, x_{2N-2})}.
\end{align*}
On the other hand, on the event $\{{\hat{\eta}}(T)=x_{2\ell}\}$, the law of ${\hat{\eta}}$ is the same as the law of $\gamma$ weighted by 
the Radon-Nikodym derivative~\eqref{eqn::Ztotal_contiuity_aux2}. Therefore, conditionally on the event $\{\eta(T)=x_{2\ell}\}$, 
the law of $\eta$ is the same as the law of $\gamma$ weighted by the Radon-Nikodym derivative 
\begin{align*}
\; & \frac{\hat{\PP}[{\hat{\eta}}(T)=x_{2\ell}]}{\PP[\eta(T)=x_{2\ell}]} \; 
\frac{1}{R_0}
\; \frac{\PartF_{\Ising}^{(\ell-1)}(D_{\gamma}^R; x_2, \ldots, x_{2\ell-1})\times\PartF_{\Ising}^{(N-\ell)}(D_{\gamma}^L; x_{2\ell+1}, \ldots, x_{2N})}{(x_{2\ell}-x_1) \underset{\beta\in\LP_{N-1}}{\sum} \PartF_{\beta}(x_1, \ldots, x_{2N-2}) \, \one \{ \link{1}{2\ell} \in \beta \} } .
\end{align*}
This gives the asserted formula~\eqref{eqn::Ztotal_RN_ind} due to~\eqref{eqn::Ztotal_continuity_aux3} and~\eqref{eqn::Ztotal_proba_ind}.
\end{proof}

With Lemmas~\ref{lem::Ztotal_continuity_ind}--\ref{lem::Ztotal_RN_ind} at hand, 
we have proved the conclusions in Theorem~\ref{thm::loewner_Ztotal_continuity} for the Loewner chain $\eta$ associated to 
$\PartF_{\Ising}$ with launching points $(x_1, \ldots, x_{2N})$ starting from $x_1$ on the event $\LE_{2N-1}$ that $\eta$ accumulates in the interval $(x_1, x_{2N-1})$. 
It remains to analyze the behavior of $\eta$ when it accumulates in 
$[x_{2N-1},\infty) \cup (-\infty,x_{1})$.
By the conformal covariance~\eqref{eq:general mobius} of the partition function $\PartF_{\Ising}$, 
we may apply 
a M\"{o}bius transformation $\varphi$ of $\HH$ such that $\varphi(x_{2N-1})<\varphi(x_{2N})<\varphi(x_1)<\varphi(x_2)<\cdots<\varphi(x_{2N-2})$. 
Then, $\varphi(\eta)$ is the Loewner chain starting from $\varphi(x_1)$, 
and the boundary segment  
$[x_{2N-1},\infty) \cup (-\infty,x_{1})$ becomes $[\varphi(x_{2N-1}), \varphi(x_{1}))$.
Thus, we may equivalently analyze the Loewner chain starting from $x_3$ when it accumulates in the boundary segment $[x_1, x_3)$. 
This remaining case is addressed in the following lemma. 

\begin{lemma} \label{lem::Ztotal_continuity_ind2}
Assume that Theorem~\ref{thm::loewner_Ztotal_continuity} holds for $N-1$.
Define $\tilde{\LE}_{2N-1}$ to be the event that the
Loewner chain $\eta$ associated to $\PartF_{\Ising}$ with launching points $(x_1, \ldots, x_{2N})$ starting from $x_3$ accumulates in the semi-open interval $[x_1, x_3)$.
Then, on the event $\tilde{\LE}_{2N-1}$, the Loewner chain $\eta$ is almost surely generated by 
a continuous simple curve up to and including $T$, which almost surely terminates at $x_2$
and touches $\R$ only at its two endpoints.
Moreover, the law of $\eta$ is that of the $\SLE_3$ curve $\gamma$ in $\HH$ from $x_3$ to $x_2$ weighted by the Radon-Nikodym derivative
\begin{align}\label{eqn::Ztotal_RN_ind2}
\frac{\PartF_{\Ising}^{(N-1)}(D_{\gamma}^R; x_1, x_4, x_5, \ldots, x_{2N})}{(x_3-x_2) \underset{\alpha\in\LP_N}{\sum} \PartF_{\alpha}(x_1, \ldots, x_{2N}) \, \one \{ \link{2}{3} \in \alpha \} } .
\end{align}
\end{lemma}

\begin{proof}
Let ${\hat{\eta}}$ be the Loewner chain associated to $\PartF_{\Ising}^{(N-1)}$ with launching points $(x_1, \ldots, x_{2N-2})$  starting from $x_3$, and denote its law by $\hat{\PP}$. 
Then, similarly as in the proof of Lemma~\ref{lem::Ztotal_continuity_ind}, we see that 
the law of $\eta$ is the same as the law of ${\hat{\eta}}$ tilted by the following local martingale, for small enough $t$: 
\begin{align*}
\tilde{R}_t
:= g_t'(x_{2N-1})^{1/2}g_t'(x_{2N})^{1/2} \; 
\frac{\PartF_{\Ising}^{(N)}(g_t(x_1), g_t(x_2), W_t, g_t(x_4), \ldots, g_t(x_{2N}))}{\PartF_{\Ising}^{(N-1)}(g_t(x_1), g_t(x_2), W_t, g_t(x_4), \ldots, g_t(x_{2N-2}))}.
\end{align*}
Thus, similar analysis as in the proof of Lemma~\ref{lem::Ztotal_continuity_ind} shows that,  
on the event $\tilde{\LE}_{2N-1}$, the law of $\eta$ is absolutely continuous with respect to the law of ${\hat{\eta}}$. 
In particular, $\eta$ is almost surely (on the event $\tilde{\LE}_{2N-1}$) generated by a continuous simple curve up to and including $T$, terminating almost surely at $x_2$.

Then, by similar analysis as in the proof of Lemma~\ref{lem::Ztotal_proba_ind}, we obtain 
\begin{align*}
\tilde{R}_T := \; & \lim_{t \to T} \tilde{R}_t \\
= \; & \lim_{t \to T}
g_t'(x_{2N-1})^{1/2}g_t'(x_{2N})^{1/2}
\bigg( \frac{\PartF_{\Ising}^{(1)}(g_t(x_2), W_t)}{\PartF_{\Ising}^{(1)}(g_t(x_2), W_t)} \bigg)
\bigg( \frac{\PartF_{\Ising}^{(N)}(g_t(x_1), g_t(x_2), W_t,  \ldots, g_t(x_{2N}))}{\PartF_{\Ising}^{(N-1)}(g_t(x_1), g_t(x_2), W_t,  \ldots, g_t(x_{2N-2}))} \bigg) \\
= \; & g_T'(x_{2N-1})^{1/2}g_T'(x_{2N})^{1/2}
\;  \frac{\PartF_{\Ising}^{(N-1)}(g_T(x_1), g_T(x_4), \ldots, g_T(x_{2N}))}{\PartF_{\Ising}^{(N-2)}(g_T(x_1), g_T(x_4), \ldots, g_T(x_{2N-2}))}. 
\qquad\qquad\qquad \text{[by~\eqref{eqn::partf_total_asy_refined}]} \\
= \; & \frac{\PartF_{\Ising}^{(N-1)}(D_{{\hat{\eta}}}^{R}; x_1, x_4, \ldots, x_{2N})}{\PartF_{\Ising}^{(N-2)}(D_{\hat{\eta}}^{R}; x_1, x_4, \ldots, x_{2N-2})}. 
\qquad\qquad\qquad\qquad\qquad\qquad \qquad\qquad\qquad \qquad\;\,  \text{[by~\eqref{eq:conformal_image}]}
\end{align*}
Thus, the law of $\eta$ is the same as the law of ${\hat{\eta}}$ tilted by $\tilde{R}_{t \wedge T}$. 
Also, by the induction hypothesis on ${\hat{\eta}}$, 
the law of ${\hat{\eta}}$ is the same as $\gamma$ weighted by 
\begin{align*} 
\frac{\PartF_{\Ising}^{(N-2)}(D_{\gamma}^R; x_1, x_4, x_5, \ldots, x_{2N-2})}{(x_3-x_2) \underset{\beta\in\LP_{N-1}}{\sum} \PartF_{\beta}(x_1, \ldots, x_{2N-2})\, \one \{ \link{2}{3} \in \beta \} }.
\end{align*}
Hence, we see that 
the law of $\eta$ is that of $\gamma$ weighted by the Radon-Nikodym derivative~\eqref{eqn::Ztotal_RN_ind2}. 
\end{proof}

\begin{corollary}\label{cor::Ztotal_continuity_ind2}
Assume that Theorem~\ref{thm::loewner_Ztotal_continuity} holds for $N-1$. 
Then, Theorem~\ref{thm::loewner_Ztotal_continuity} also holds for $N$. 
\end{corollary}

\begin{proof}
For definiteness, we assume that $j=1$. 
Lemmas~\ref{lem::Ztotal_continuity_ind} 
and~\ref{lem::Ztotal_continuity_ind2} together with the 
conformal covariance~\eqref{eq:general mobius} of the partition function $\PartF_{\Ising}$
show that the Loewner chain associated to $\PartF_{\Ising}$ with launching points $(x_1, \ldots, x_{2N})$ starting from $x_1$
is almost surely generated by a continuous simple curve
$(\eta(t), 0\le t\le T)$ up to and including $T$, and that
this curve almost surely terminates at one of the points 
$\{x_2, x_4, \ldots, x_{2N}\}$. 
Lemmas~\ref{lem::Ztotal_proba_ind} 
and~\ref{lem::Ztotal_continuity_ind2} then imply that  
for any $\ell \in \{1,2,\ldots, N-1\}$, we have
\begin{align*} 
\PP[\eta(T)=x_{2\ell}] = \sum_{\alpha\in\LP_N}\frac{\PartF_{\alpha}(x_1, \ldots, x_{2N})}{\PartF_{\Ising}^{(N)}(x_1,\ldots, x_{2N})} 
\, \one \{ \link{1}{2\ell} \in \alpha \}  .
\end{align*}
Since $\PartF_{\Ising} = \sum_\alpha \PartF_{\alpha}$ 
and since the total probabilities sum up to one, 
this also gives  
\begin{align*} 
\PP[\eta(T)=x_{2N}] 
= \; & 1 - \sum_{\ell=1}^{N-1} \PP[\eta(T)=x_{2\ell}] \\
= \; & \sum_{\alpha\in\LP_N} \frac{\PartF_{\alpha}(x_1, \ldots, x_{2N}) }{\PartF_{\Ising}^{(N)}(x_1,\ldots, x_{2N})}  
- \sum_{\ell=1}^{N-1}
\sum_{\alpha\in\LP_N} \frac{\PartF_{\alpha}(x_1, \ldots, x_{2N})}{\PartF_{\Ising}^{(N)}(x_1,\ldots, x_{2N})} 
\, \one \{ \link{1}{2\ell} \in \alpha \} \\
= \; & \sum_{\alpha\in\LP_N} \frac{\PartF_{\alpha}(x_1, \ldots, x_{2N}) }{\PartF_{\Ising}^{(N)}(x_1,\ldots, x_{2N})}  
\, \one \{ \link{1}{2N} \in \alpha \} .
\end{align*}
Thus, we conclude that the asserted formula~\eqref{eqn::Loewner_Ztotal_proba} with $j=1$ holds.

Lastly, Lemma~\ref{lem::Ztotal_RN_ind} shows that 
for all $\ell \in \{1,2, \ldots, N-1\}$, conditionally on the event $\{\eta(T)=x_{2\ell}\}$, the law of $\eta$ 
is that of the $\SLE_3$ curve $\gamma$ in $\HH$ from $x_1$ to $x_{2\ell}$ weighted by the Radon-Nikodym derivative~\eqref{eqn::Ztotal_RN_ind}, which gives the asserted formula~\eqref{eqn::Loewner_Ztotal_RN} on these events. 
On the other hand, 
Equation~\eqref{eqn::Ztotal_RN_ind2} from Lemma~\ref{lem::Ztotal_continuity_ind2} together with the covariance~\eqref{eq:general mobius} shows that, conditionally on the event $\{\eta(T)=x_{2N}\}$, the law of $\eta$ 
is that of the $\SLE_3$ curve $\gamma$ in $\HH$ from $x_1$ to $x_{2N}$ weighted by
\begin{align*} 
\frac{\PartF_{\Ising}^{(N-1)}(D_{\gamma}^R; x_{2}, \ldots, x_{2N-1})}{(x_{2N} - x_1) \; \underset{\alpha\in\LP_N}{\sum} \PartF_{\alpha}(x_1, \ldots, x_{2N}) 
\, \one \{ \link{1}{2N} \in \alpha \} } .
\end{align*}
This gives the asserted formula~\eqref{eqn::Loewner_Ztotal_RN} on this event and concludes the proof. 
\end{proof}

With Corollary~\ref{cor::Ztotal_continuity_ind2}, the proof of Theorem~\ref{thm::loewner_Ztotal_continuity} is complete by induction.

\section{Crossing Probabilities in the Critical Ising Model}
\label{sec::Ising_crossing_proba}
Let $\Omega^{\delta}$ be a family of finite subgraphs of the rescaled square lattice $\delta\Z^2$, for $\delta > 0$, together with
$2N$ fixed boundary points $x_1^{\delta}, \ldots, x_{2N}^{\delta}$ for each $\Omega^{\delta}$ in counterclockwise order.
As illustrated in Figure~\ref{fig::Ising} in Section~\ref{sec::intro}, we 
consider the critical Ising model on $(\Omega^{\delta}; x_1^{\delta}, \ldots, x_{2N}^{\delta})$ with alternating boundary conditions:
\begin{align} \label{eq::alternating}
  \begin{cases}
    \oplus \text{ on }(x_{2s-1}^{\delta} \, x_{2s}^{\delta}) ,   & \quad\text{for } s \in \{1,2, \ldots, N\} ,   \\
    \ominus \text{ on }(x_{2s}^{\delta} \,  x_{2s+1}^{\delta}) , & \quad \text{for } s \in \{0,1,\ldots, N\} , 
  \end{cases}
\end{align}
where $(x_i^{\delta} \, x_{i+1}^{\delta})$ stands for the counterclockwise boundary arc from $x_i^{\delta}$ to $x_{i+1}^{\delta}$, 
with   the   convention  that   $x_{2N}^{\delta}=x_{0}^{\delta}$   and
$x_{2N+1}^{\delta}=x_1^{\delta}$.
In this setup, each Ising model configuration on $\Omega^{\delta}$ contains 
$N$ random macroscopic interfaces which connect pairwise the $2N$ boundary points $x_1^{\delta}, \ldots, x_{2N}^{\delta}$. 
When $N \geq 2$, these interfaces can form more than one possible connectivity pattern, as illustrated in Figure~\ref{fig::Ising}.

Suppose that $(\Omega^{\delta}; x_1^{\delta}, \ldots, x_{2N}^{\delta})$ approximate some 
polygon $(\Omega; x_1, \ldots, x_{2N})$ as $\delta\to 0$, as detailed below.
K.~Izyurov proved in his article~\cite{IzyurovIsingMultiplyConnectedDomains} that ``locally", the scaling limits of these interfaces
are given by the Loewner chain~\eqref{eqn::loewnerchain_partition} with $\kappa=3$ and $\PartF = \PartF_{\Ising}$.
In Section~\ref{subsec::ising_Ztotal_global}, 
we briefly explain how
to extend this result to a ``global" one, 
that is, to establish the convergence for the \textit{whole} curves instead only up to a stopping time. 
The proof crucially relies on the continuity of the Loewner chain, 
Theorem~\ref{thm::loewner_Ztotal_continuity}.

In this article, we are primarily interested in the probability that the Ising interfaces form a given connectivity encoded in a link pattern $\alpha$.
In general, formulas for such crossing probabilities are not known --- a few special cases appear in~\cite{IzyurovObservableFree}.
Nevertheless, in this section we will prove Theorem~\ref{thm::ising_crossingproba}, which shows that
the critical Ising crossing probabilities do indeed have a conformally invariant scaling limit (cf.~Corollary~\ref{cor::CI}), 
specified as the ratio $\PartF_\alpha / \PartF_{\Ising}$ of partition functions discussed in Sections~\ref{sec::pre}--\ref{sec::Ising_partition_function}. Interestingly, this ratio also gives a 
characterization of the Ising crossing probabilities in terms of a $c = 1/2$ conformal field theory: 
indeed, the probability amplitudes $\PartF_\alpha$ can be seen as 
correlation functions of a degenerate field with conformal weight $h_{1,2} = 1/2$,
associated to the free fermion 
(or the energy density) on the boundary --- for more discussion on these concepts,
see, e.g., the textbooks~\cite{ID_book, DMS:CFT}, the article~\cite{PeltolaICMP}, 
and the results in~\cite{HonglerThesis, HonglerSmirnovIsingEnergy}.

\subsection{Ising Model} \label{sub::Ising_notation}

To begin, we fix notation to be used throughout. 
We consider finite subgraphs $\graph = (V(\graph), E(\graph))$ of 
the  (possibly translated, rotated, and rescaled) square lattice $\Z^2$. 
We call two vertices $v$ and $w$ \textit{neighbors} if their Euclidean distance equals one, and we then write $v \sim w$.
We denote the inner boundary of $\graph$ by 
\begin{align*}
\partial \graph = \{ v \in V(\graph) \, \colon \, \exists \; w \not\in V(\graph) \text{ such that } \edge{v}{w} \in E(\Z^2)\} .
\end{align*}

The \textit{dual lattice} $(\Z^2)^*$ is a translated version of $\Z^2$: its vertex set is $(1/2, 1/2)+\Z^2$ and its edges are given by all pairs of vertices that are neighbors. 
The vertices and edges of $(\Z^2)^*$ are called dual-vertices and dual-edges, 
while we sometimes call the vertices and edges of $\Z^2$ primal-vertices and primal-edges.
In particular, for each primal-edge $e$ of $\Z^2$, we associate a unique dual-edge, denoted by $e^*$, that crosses $e$ in the middle. 
For a subgraph $\graph$ of $\Z^2$, we define $\graph^*$ to be the subgraph of $(\Z^2)^*$ with edge set $E(\graph^*)=\{e^*: e\in E(\graph)\}$ and 
vertex set given by the endpoints of these dual-edges.

We define \textit{a discrete  Dobrushin domain} to be  a  triple $(\graph; v, w)$ with $v, w \in  \partial \graph$,  $v \neq w$,
where $\graph$ is  a  finite connected  subgraph  of  $\Z^2$ such that 
the complement of $\graph$ is also  connected (that is, $\graph$ is simply connected). 
The boundary  $\partial \graph$ is divided into  two arcs $(v \, w)$ and $(w \, v)$, where Dobrushin boundary conditions for the Ising model will be specified.
We also define  \textit{a  discrete  polygon} to be a $(2N+1)$-tuple $(\graph;  v_1,  \ldots,  v_{2N})$, 
where $v_1, \ldots, v_{2N} \in \partial  \graph$ are distinct boundary vertices in counterclockwise order. 
In this case, the boundary $\partial \graph$ is divided into
$2N$ arcs, where alternating boundary conditions will be specified.
We also  let  $\graph$ denote the simply connected domain formed  by all of the faces, edges, and vertices of~$\graph$.

\smallbreak

The \textit{Ising model} on $\graph$ is a random assignment
$\sigma = (\sigma_v)_{v \in V(\graph)} \in \{\ominus, \oplus\}^{V(\graph)} =: \Sigma_\graph$ of spins. 
With free boundary conditions, the probability measure of the Ising model is given by the Boltzmann measure 
\begin{align*}
  \mu^{\free}_{\beta,\graph}[\sigma]
  = \frac{\exp(-\beta H^{\free}_{\graph}(\sigma))}{Z^{\free}_{\beta, \graph}}, \qquad
  \text{where} \quad Z^{\free}_{\beta, \graph} 
  = \sum_{\sigma \in \Sigma_\graph}\exp(-\beta H^{\free}_{\graph}(\sigma)) ,
\end{align*}
with inverse-temperature $\beta>0$ and Hamiltonian 
\begin{align*}
H^{\free}_{\graph}(\sigma) = - \sum_{v \sim w}\sigma_v \sigma_w .
\end{align*}
Only the neigboring spins interact with each other. This model exhibits an order-disorder phase transition~\cite{McCoyWuIsing}:
there exists a critical temperature such that above it, the Ising configurations are disordered, 
and below it, large clusters of equal spins appear. 
For the square lattice, the critical inverse-temperature can be found exactly: 
\begin{align*}
\beta_{\textnormal{crit}}= \frac{\log(1+\sqrt{2})}{2} .
\end{align*}
At criticality, the system does not have a typical length scale, and using renormalization arguments, 
physicists argued, e.g., in~\cite{BelavinPolyakovZamolodchikovConformalSymmetry, CardyRenorm},
and mathematicians later proved in a series of works starting from~\cite{SmirnovConformalInvariance,SmirnovConformalInvarianceAnnals},
that the model becomes conformally invariant in the scaling limit.
In this article, we consider the scaling limit of the Ising model at criticality and verify a feature of its conformal invariance.

For $\varrho \in \{\ominus, \oplus\}^{\Z^2}$, we define the Ising model 
$\smash{\mu^{\varrho}_{\beta,\graph}}$
with boundary condition $\varrho$ via the Hamiltonian
\begin{align*}
  H^{\varrho}_{\graph}(\sigma) = - \sum_{\substack{v \sim w, \\ \edge{v}{w} \cap \graph \neq \emptyset}} \sigma_v \sigma_w, \qquad
  \text{where} \quad \sigma_v = \varrho_v ,  \text{ for all } v \not\in \graph .
\end{align*}

The Ising model satisfies the following domain \textit{Markov property}, which enables the martingale argument that will be used to prove Theorem~\ref{thm::ising_crossingproba} in Section~\ref{ProofSec}.
Suppose $\graph \subset \graph'$ and fix boundary condition $\varrho \in \{\ominus, \oplus\}^{\Z^2}$  
for the Ising model on $\graph'$.
If $X$ is a random variable which is measurable with respect to the status of the vertices of the smaller graph $\graph$, then we have
\begin{align*}
\mu^{\varrho}_{\beta,\graph'} \big[X \; | \; \sigma_v = \varrho_v \text{ for all } v \in V(\graph')\setminus V(\graph)\big] 
= \mu^{\varrho}_{\beta,\graph} [X] .
\end{align*}

If $(\graph; v, w)$ is a discrete Dobrushin domain, we may consider the Ising model with \textit{Dobrushin boundary conditions}
(domain-wall boundary conditions): we set
$\oplus$ along the arc $(v \, w)$, and $\ominus$ along the complementary arc $(w \, v)$. 
More generally, in a discrete polygon $(\graph;  v_1,  \ldots,  v_{2N})$, 
we consider alternating boundary conditions, where $\oplus$ and $\ominus$ alternate along the boundary as in~\eqref{eq::alternating} (see also Figure~\ref{fig::Ising}).

Note that the spins lie on the primal-vertices $v \in \graph$,  while interfaces lie on the dual lattice $\graph^*$. 
Let $v_1^*, \ldots, v_{2N}^*$ be dual-vertices nearest to $v_1, \ldots, v_{2N}$, respectively. 
Then, given $s \in \{1,2,\ldots,N\}$,  we  define   the  Ising \textit{interface} starting from $v_{2s}^*$ as follows. 
It starts from $v_{2s}^*$, traverses on the dual-edges, and turns at every dual-vertex in  such a  way that  it always  
has primal-vertices with  spin $\oplus$  on its  left and  spin $\ominus$ on its right. If there is an indetermination  when arriving at a vertex 
(this may  happen  on   the  square  lattice),  it  turns   left. 
The Ising interface starting from $v_{2s-1}^*$ is defined similarly with the left/right switched.

\smallbreak

We focus on \textit{scaling limits} of the Ising model on planar domains: we let $\graph = \Omega^{\delta}$ be a subgraph of 
the rescaled square lattice\footnote{In this article, we only consider square lattice  approximations of simply connected continuum domains.
However, more general results could also be derived, e.g.,  for so-called isoradial graphs~\cite{CS-discrete_complex_analysis_on_isoradial,ChelkakSmirnovIsing}.} $\delta \Z^2$ with small $\delta > 0$, which will tend to zero.
Our precise approximation scheme is the following~\cite[Section~4.3]{KarrilaConformalImage}. 
We say that a sequence of discrete polygons $(\Omega^{\delta}; x_1^{\delta}, \ldots, x_{2N}^{\delta})$  
converges as $\delta \to 0$ to a polygon $(\Omega; x_1, \ldots, x_{2N})$ in the \textit{close-Carath\'{e}odory sense} if it converges in the Carath\'{e}odory sense, and in addition, for each $j\in\{1,2, \ldots, 2N\}$, we have $x_j^{\delta}\to x_j$ as $\delta\to 0$ and the following is fulfilled: 
Given a reference point $z \in \Omega$ and 
$r>0$ small enough, let $S_r$ be the arc of $\partial B(x_j,r)\cap\Omega$ disconnecting (in $\Omega$) $x_j$ from $z$ and from all other arcs of this set. We require that, for each $r$ small enough and for all sufficiently small $\delta$ (depending on $r$), the boundary point $x_j^{\delta}$ is connected to the midpoint of $S_r$ inside $\Omega^{\delta}\cap B(x_j,r)$.

We emphasize that the Ising spins lie on the primal-vertices the interfaces traverse on the dual graph. 
However, we shall abuse notation by writing $\Omega^{\delta}$ for both $\Omega^{\delta}$ and $(\Omega^*)^{\delta}$, 
and $x^{\delta}$ for both $x^{\delta}$ and $(x^{*})^{\delta}$.

\subsection{Convergence of Interfaces} \label{subsec::ising_Ztotal_global}

In this section, we first summarize some existing results on the convergence of Ising interfaces, 
and then explain how to extend Izyurov's result~\cite{IzyurovIsingMultiplyConnectedDomains} on the local convergence of multiple interfaces to be global.
The convergence will take place weakly in the space of unparameterized curves with metric~\eqref{def::curves_metric}.

Starting from the celebrated work of S.~Smirnov~\cite{SmirnovConformalInvariance,SmirnovConformalInvarianceAnnals}, 
conformal invariance for 
correlations~\cite{ChelkakLzyurovSpinorIsing, HonglerSmirnovIsingEnergy, ChelkakHonglerLzyurovConformalInvarianceCorrelationIsing, 
ChelkakHonglerLzyurovConformalInvarianceCorrelationIsingGeneral}
and interfaces~\cite{HonglerKytolaIsingFree, 
CDCHKSConvergenceIsingSLE,  IzyurovIsingMultiplyConnectedDomains,
BenoistHonglerIsingCLE, BeffaraPeltolaWuUniqueness}
for the critical planar Ising model has now been verified. 
The key tool in this work is the so-called discrete holomorphic fermion, developed by Smirnov with D. Chelkak~\cite{ChelkakSmirnovIsing}. 
This led in particular to the convergence of the Ising interface in Dobrushin domains~\cite{CDCHKSConvergenceIsingSLE}: 
if $(\Omega^{\delta}; x^{\delta}, y^{\delta})$ is a sequence of discrete  Dobrushin domains converging  to  a Dobrushin  domain $(\Omega; x,  y)$  
in the close-Carath\'eodory sense, then, as $\delta\to  0$,  the  interface  of  the  critical  Ising  model  on $(\Omega^{\delta}; x^{\delta},  y^{\delta})$  with  
Dobrushin boundary conditions  converges weakly to the chordal  $\SLE_{3}$ in $(\Omega; x,y)$. 
Later, C.~Hongler, K.~Kyt\"{o}l\"{a}, and K.~Izyurov extended the discrete holomorphic fermion to 
more general settings~\cite{HonglerKytolaIsingFree, IzyurovObservableFree, IzyurovIsingMultiplyConnectedDomains}. 
In particular, it follows from Izyurov's work~\cite[Theorem~1.1]{IzyurovIsingMultiplyConnectedDomains}
that the Ising interfaces in discrete polygons 
with alternating boundary conditions~\eqref{eq::alternating} converge to multiple $\SLE_3$ curves in the following local sense.

\bigskip

\noindent\textbf{\textnormal{($\clubsuit$)} Setup.}
Let discrete polygons $(\Omega^{\delta}; x_1^{\delta}, \ldots, x_{2N}^{\delta})$ on $\delta\Z^2$ 
converge to a polygon $(\Omega; x_1, \ldots, x_{2N})$ as $\delta\to 0$ in the close-Carath\'{e}odory sense.
Consider the critical Ising model on $\Omega^{\delta}$ with alternating boundary conditions~\eqref{eq::alternating}. 
Let $\varphi_{\delta} \colon \HH \to \Omega^{\delta}$
and $\varphi \colon \HH \to \Omega$ be conformal maps such that
as $\delta \to 0$, we have 
$\varphi_{\delta} \to \varphi$ uniformly on any compact subset of $\HH$, and $\varphi_\delta^{-1}(x_j^{\delta}) \to \varphi^{-1}(x_j) =: \realpt_j$ for all $j$. 
Fix $j$ throughout, and let $\eta_j^{\delta}$ be the Ising interface in $\Omega^{\delta}$ starting from $x_j^{\delta}$. 
For $r>0$, define $T_r^\delta = T_r^\delta(j)$ to be the first time when $\eta_j^{\delta}$ gets within distance $r$ from the other marked points 
$\{x_1^{\delta}, \ldots, x_{j-1}^{\delta}, x_{j+1}^{\delta}, \ldots, x_{2N}^{\delta}\}$. 
Let $\eta$ be the Loewner chain 
associated to the partition function $\PartF_{\Ising}$ with launching points $(\realpt_1, \ldots, \realpt_{2N})$  
starting from $\realpt_j$, up to the stopping time 
$T_r$, i.e., the first time when $\eta$ gets within distance $r$ from the other marked points. 
The fact that the stopping times for the interface $\eta_j^{\delta}$ indeed converge in probability to a stopping time for $\eta$
can be proved using the Russo-Seymour-Welsh estimate --- see~\cite[Section~4]{GarbanWuFKIsing}.
Then, as $\delta\to 0$, the curves $(\varphi_\delta^{-1}(\eta_j^{\delta}(t)), 0\le t\le T_r^\delta)$ 
converge weakly to $(\eta(t), 0\le t\le T_r)$ in the metric~\eqref{def::curves_metric}~\cite[Theorem~1.1]{IzyurovIsingMultiplyConnectedDomains}. 
We call this \textit{local convergence} since the convergence only holds up to the cutoff time $T_r$.

\smallbreak

To establish the convergence globally, that is, without a cutoff time, we need three pieces of input: 
\begin{enumerate}
\item \label{1} the local convergence, as explained above; 

\item \label{2} the fact that the limiting curve $\tilde{\eta}_j$ of $\varphi_\delta^{-1}(\eta_j^{\delta}(t))$ is continuous up to and including 
the first swallowing time~\eqref{eq:total_swallowing_time} of one of the other marked points, denoted $T$; and 

\item \label{3} the fact that the Loewner chain $\eta$ is continuous up to and including the same swallowing time~\eqref{eq:total_swallowing_time}. 
\end{enumerate}
With these three facts at hand, we know that $\tilde{\eta}_j$ has the same law as $\eta$ up to the cutoff time $T_r$ by~Input~\ref{1}, 
and letting $r\to 0$, we find that $\tilde{\eta}_j$ and $\eta$ have the same law up to and including $T$ thanks to of~Inputs~\ref{2} and~\ref{3}. 
Among the three pieces of input, we have the local convergence (Input~\ref{1}) from~\cite{IzyurovIsingMultiplyConnectedDomains} and~\cite{KemppainenSmirnovRandomCurves}. 
The continuity of the scaling limit (Input~\ref{2}) is a consequence of the Russo-Seymour-Welsh estimate for the Ising model 
from~\cite[Corollary~1.7]{ChelkakDuminilHonglerCrossingprobaFKIsing}
combined with the results in~\cite{AizenmanBurchardHolderRegularity, KemppainenSmirnovRandomCurves}, 
as argued in~\cite[Remark~3.2]{IzyurovObservableFree}. 
Finally, we established the (non-trivial) continuity of the Loewner chain $\eta$ (Input~\ref{3})  
in Theorem~\ref{thm::loewner_Ztotal_continuity} in the previous section. In summary, we have the following convergence.

\begin{proposition}\label{prop::cvg_Ztotal_global}
In the above setup \textnormal{($\clubsuit$)} with fixed $j \in \{1,2,\ldots,2N\}$, consider
the Ising interface $\eta_j^{\delta}$ starting from $x_j^{\delta}$ 
up to the first time $T^{\delta} = T^{\delta}(j)$ when $\eta_j^{\delta}$ hits one of the other marked points $\{x_1^{\delta}, \ldots, x_{j-1}^{\delta}, x_{j+1}^{\delta}, \ldots, x_{2N}^{\delta}\}$. 
Then, $(\varphi_\delta^{-1} (\eta_j^{\delta}), 0\le t\le T^{\delta})$ converges weakly to $(\eta(t), 0\le t\le T)$ in the metric~\eqref{def::curves_metric}, where $T$ is the first swallowing time~\eqref{eq:total_swallowing_time} by $\eta$ of some spectator point.
\end{proposition}

\begin{remark}
In~\textnormal{\cite[Theorem~1.2]{BeffaraPeltolaWuUniqueness}}, 
it was proved that for each $\alpha \in \LP_N$, there exists a unique ``global'' multiple $\SLE_3$
associated to $\alpha$. This is a probability measure supported on families of curves with the given topological connectivity $\alpha$.
Furthermore, it follows from~\textnormal{\cite[Lemma~4.8 \& Proposition~4.9]{PeltolaWuGlobalMultipleSLEs}} that 
these curves are given  by Loewner chains whose driving functions satisfy the SDEs~\eqref{eqn::loewnerchain_partition}
with $\PartF$ the pure partition function $\PartF_\alpha$ for $\kappa = 3$.
On the other hand, since 
$\PartF_{\Ising} = \sum_{\alpha} \PartF_\alpha$, 
it follows from Proposition~\ref{prop::cvg_Ztotal_global} and the so-called local commutation property 
for multiple $\SLE$s \`a la Dub{\'e}dat~\textnormal{\cite{DubedatCommutationSLE}}
\textnormal{(}cf.~\textnormal{\cite[Sampling~Procedure~A.3]{KytolaPeltolaPurePartitionSLE}}
and~\textnormal{\cite[Corollary~1.2]{PeltolaWuGlobalMultipleSLEs}}\textnormal{)} 
that the probability measure of the collection of limit curves of the Ising interfaces 
is the global multiple $\SLE_3$ which is a convex combination
of the extremal multiple $\SLE_3$ probability measures associated to the various possible connectivity patterns of the 
interfaces --- see also the recent work~\textnormal{\cite{Karrila_multiple_SLE}}. 
From this, we see that the probability amplitudes $\PartF_\alpha$ can be thought of as a manifestation of Doob's h-transform. 
Making this precise, however, would require some further technical work. 
\end{remark}

\subsection{Crossing Probabilities --- Proof of Theorem~\ref{thm::ising_crossingproba}} \label{ProofSec}

Now we are ready to prove the main result of this article:
\IsingTHM*

\begin{proof}
We prove the claim by induction on $N \geq 1$. It is trivial for $N=1$ because both sides of~\eqref{eqn::ising_crossing_proba} equal one.  
Thus, we assume that the claim holds for $N-1$, fix $\alpha\in\LP_N$ and aim to prove the claim for $\PP^{\delta}[\conn^{\delta}=\alpha]$. 
The probabilities $(\PP^{\delta}[\conn^{\delta}=\alpha])_{\delta>0}$ form a sequence of numbers in $[0,1]$, so there is always a subsequential limit. 
To show~\eqref{eqn::ising_crossing_proba}, it hence suffices to prove that 
\begin{align*}
\lim_{n\to\infty}\PP^{\delta_n}[\conn^{\delta_n}=\alpha] \; = \; 
\frac{\PartF_{\alpha}(\Omega;x_{1},\ldots,x_{2N})}{\PartF^{(N)}_{\Ising}(\Omega;x_{1},\ldots,x_{2N})} 
\end{align*}
for any convergent subsequence. Note that the right-hand side is
conformally invariant by the M\"{o}bius covariance~\eqref{eq: multiple SLE Mobius covariance} with $h=1/2$.

For topological reasons, the link pattern $\alpha$ contains at least one link of type $\link{j}{j+1}$. For definiteness, we assume that $j=1$, so $\link{1}{2} \in \alpha$. 
From Proposition~\ref{prop::cvg_Ztotal_global}, we see that 
$(\varphi_{\delta_n}^{-1} (\eta_1^{\delta_n}), 0\le t\le T^{\delta_n})$ converges weakly in the metric~\eqref{def::curves_metric}  to $(\eta(t), 0\le t\le T)$, that is the Loewner chain 
associated to the partition function $\PartF_{\Ising}$ with launching points $(\realpt_1, \ldots, \realpt_{2N})$  
starting from $\realpt_1$. 
For convenience, we couple them by the Skorohod representation theorem in the same probability space so that they converge almost surely.

First, let us analyze the limit curve $\eta$. Using the notation $g_t(\realpt_i) = V_t^i$, for all $i \neq 1$, we define 
\begin{align} \label{eq: Ising martingale}
M_t := \frac{\PartF_{\alpha}(W_t, V_t^2, \ldots, V_t^{2N})}{\PartF_{\Ising}^{(N)}(W_t, V_t^2, \ldots, V_t^{2N})} , \qquad t<T .
\end{align}
The partial differential equations~\eqref{eq: multiple SLE PDEs} show that $M_t$ is a local martingale. 
Let us consider the limit of $M_t$ as $t\to T$. 
From Theorem~\ref{thm::loewner_Ztotal_continuity}, we know that the curve $\eta$ is almost surely continuous up to and including $T$ and
terminates at one of the points $\{\realpt_2, \realpt_4, \ldots, \realpt_{2N}\}$.  
Denote by $D_{\eta}$ the unbounded connected component of $\HH\setminus\eta[0,T]$, and by $\hat{\alpha} = \alpha\removeLink \link{1}{2}\in\LP_{N-1}$. 
On the event $\{\eta(T)=\realpt_2\}$, we see from 
the strong asymptotics properties~\eqref{eqn::partf_alpha_asy_refined} and~\eqref{eqn::partf_total_asy_refined} 
and the conformal covariance property~\eqref{eq:conformal_image} that
\begin{align*}
M_t 
= \; & \bigg( \frac{\PartF_{\alpha}(W_t, V_t^2, \ldots, V_t^{2N})}{\PartF_{\Ising}^{(1)}(W_t, V_t^2)} \bigg)
\bigg( \frac{\PartF_{\Ising}^{(1)}(W_t, V_t^2)}{\PartF_{\Ising}^{(N)}(W_t, V_t^2, \ldots, V_t^{2N})} \bigg) \\
\overset{t\to T}{\longrightarrow} \; & \;
\frac{\PartF_{\hat{\alpha}}(V_T^3, \ldots, V_T^{2N})}{\PartF_{\Ising}^{(N-1)}(V_T^3, \ldots, V_T^{2N})} 
\; = \; 
\frac{\PartF_{\hat{\alpha}}(D_{\eta}; \realpt_3, \ldots, \realpt_{2N})}{\PartF_{\Ising}^{(N-1)}(D_{\eta}; \realpt_3, \ldots, \realpt_{2N})} ,
\qquad \text{almost surely}. 
\end{align*}
On the other hand, for each $\ell \in \{2, 3, \ldots, N\}$, on the event $\{\eta(T)=\realpt_{2\ell}\}$ 
Proposition~\ref{prop::mart_wrong_zero} gives 
\begin{align*}
M_t 
= \frac{\PartF_{\alpha}(W_t, V_t^2, \ldots, V_t^{2N})}{\PartF_{\Ising}^{(N)}(W_t, V_t^2, \ldots, V_t^{2N})}
\quad \overset{t\to T}{\longrightarrow} \quad 
0 ,
\qquad \text{almost surely}.
\end{align*}
In summary, we have 
\begin{align*}
M_t 
\quad \overset{t\to T}{\longrightarrow} \quad
M_T :=\one \{\eta(T)=\realpt_2\} \frac{\PartF_{\hat{\alpha}}(D_{\eta}; \realpt_3, \ldots, \realpt_{2N})}{\PartF_{\Ising}^{(N-1)}(D_{\eta}; \realpt_3, \ldots, \realpt_{2N})} ,
\qquad \text{almost surely}.
\end{align*}
Since $0<\PartF_{\alpha}/\PartF_{\Ising}\le 1$ due to~\eqref{eqn::zsymmetric_def}, we see that $(M_t, t \leq T)$ is a bounded martingale. The optional stopping theorem then gives the identity $M_0=\E[M_T]$, that is,
\begin{align} \label{eqn::ising_crossing_aux}
\frac{\PartF_{\alpha}(\realpt_1, \ldots, \realpt_{2N})}{\PartF_{\Ising}^{(N)}(\realpt_1, \ldots, \realpt_{2N})}
= \E \bigg[\one \{\eta(T)=\realpt_2\}\frac{\PartF_{\hat{\alpha}}(D_{\eta}; \realpt_3, \ldots, \realpt_{2N})}{\PartF_{\Ising}^{(N-1)}(D_{\eta}; \realpt_3, \ldots, \realpt_{2N})}\bigg] .
\end{align}

Next, let us consider the discrete interface $\eta_1^{\delta_n}$. For simplicity of notation, we shall use the superscript ``$n$'' instead of ``$\delta_n$'' in what follows. 
On the event $\{\eta_1^{n}(T^{n}) = x_2^{n}\}$, 
we denote by $D^{n}$ the connected 
component of $\Omega^{n}\setminus\eta_1^{n}$ with $x_3^{n}, \ldots, x_{2N}^{n}$ on its boundary. 
Since $\varphi_{n}^{-1} (\eta_1^{n})$ 
converges to the continuous simple curve $\eta$ in $\HH$   that intersects the boundary $\R$ only at its two endpoints by Theorem~\ref{thm::loewner_Ztotal_continuity}, 
we see that, as $n \to \infty$, 
the polygon $(D^{n}; x_3^{n}, \ldots, x_{2N}^{n})$ converges almost surely to 
the polygon $(\varphi^{-1}(D_{\eta}); x_3, \ldots, x_{2N})$ in the close-Carath\'{e}odory sense. 
Hence, using the domain Markov property of the Ising model,
the induction hypothesis, and the conformal invariance of the right-hand side of~\eqref{eqn::ising_crossing_proba}  
and the conformal invariance of the $\SLE_{3}$ type curve $\eta$, we have
\begin{align} 
\nonumber
\E^{n} [\one{\{\conn^{n}=\alpha\}}\cond \eta_1^{n}]
= \;\; & \one{\{\eta_1^{n}(T^{n}) = x_2^{n}\}} \; 
\E^{n} [\one{\{\smash{\widehat{\conn}^{n}} = \hat{\alpha}\}} \cond \eta_1^{n} ] \\
\nonumber
= \;\; & \one{\{\eta_1^{n}(T^{n}) = x_2^{n}\}} \; 
\hat{\PP}^{n} [\smash{\widehat{\conn}^{n}} = \hat{\alpha} ] \\
\label{eqn::induction_hypo}
\overset{n \to \infty}{\longrightarrow} \; &
\one \{\eta(T)=\realpt_2 \} \;
\frac{\PartF_{\hat{\alpha}}(D_{\eta}; \realpt_3, \ldots, \realpt_{2N})}{\PartF_{\Ising}^{(N-1)}(D_{\eta}; \realpt_3, \ldots, \realpt_{2N})}  ,
\end{align}
where $\hat{\PP}^{n}$ is the law of the Ising interfaces 
on the random polygon $(D^{n}; x_3^{n}, \ldots, x_{2N}^{n})$, measurable with respect to $\eta_1^{n}$, 
which form a random connectivity pattern $\smash{\widehat{\conn}^{n}} \in \LP_{N-1}$, 
and where, by the Skorohod representation theorem, 
we couple all of the random variables on the same probability space so that the convergence takes place almost surely. 
Thus, we conclude (using the bounded convergence theorem) that 
\begin{align*}
\tilde{P}_\alpha 
:= \; & \lim_{n\to\infty}\PP^{n}[\conn^{n}=\alpha] && \\
= \; &  \lim_{n\to\infty} \E^{n} \big[\one{\{\eta_1^{n}(T^{n}) = x_2^{n}\}} \; \E^{n} [\one{\{\conn^{n}=\alpha\}}\cond \eta_1^{n}]\big]
&& \text{[by tower property]}  \\
=\; & \E \bigg[\one \{\eta(T)=\realpt_2 \} 
\frac{\PartF_{\hat{\alpha}}(D_{\eta}; \realpt_3, \ldots, \realpt_{2N})}{\PartF_{\Ising}^{(N-1)}(D_{\eta}; \realpt_3, \ldots, \realpt_{2N})}\bigg]
&& \text{[by~\eqref{eqn::induction_hypo}]}  \\
= \; & \frac{\PartF_{\alpha}(\realpt_1, \ldots, \realpt_{2N})}{\PartF_{\Ising}^{(N)}(\realpt_1, \ldots, \realpt_{2N})} ,
&& \text{[by~\eqref{eqn::ising_crossing_aux}]} 
\end{align*}
This completes the induction step and finishes the proof of Theorem~\ref{thm::ising_crossingproba}.
\end{proof}

\begin{proof}[Proof of Corollary~\ref{cor::CI}]
The asserted properties follow from the corresponding properties of the multiple $\SLE_3$ partition functions $\PartF_\alpha$ and $\PartF_{\Ising}$ by using the right-hand side of~\eqref{eqn::ising_crossing_proba}:
the conformal invariance follows from the M\"{o}bius covariance~\eqref{eq: multiple SLE Mobius covariance} with $h=1/2$;
the asymptotics is a consequence of~(\ref{eq: multiple SLE asymptotics},~\ref{eqn::partf_total_asy_refined})
with $h=1/2$;
and the PDEs~\eqref{eq:PDEIsing} are given by the PDEs~\eqref{eq: multiple SLE PDEs} with $\kappa=3$ and $h=1/2$.
\end{proof}

\newpage
\appendix
\section{Connection Probabilities for Level Lines of GFF}
\label{app::appendix_gff}
In the proof of Lemma~\ref{lem::b_total_ineq} in Section~\ref{sec::Ising_partition_function}, 
we use the following facts concerning the level lines of the Gaussian free field ($\GFF$).
We shall not define the $\GFF$ nor its level lines precisely, because they are not needed to understand the present article.
The reader may find background on this topic, e.g., in~\cite{SheffieldGFFMath, SchrammSheffieldContinuumGFF, WangWuLevellinesGFFI}.
Importantly, the level lines are $\SLE_\kappa$ type curves with $\kappa = 4$.

Fix a constant $\lambda=\pi/2$. Let $\gff$ be the $\GFF$ in $\HH$ with alternating boundary data: 
\begin{align*}
\lambda \text{ on }(x_{2s-1}, x_{2s}), \text{ for } s \in \{ 1, \ldots, N \} ,
\quad \text{and} \quad 
-\lambda \text{ on }(x_{2s}, x_{2s+1}) , \text{ for }  s \in \{ 0, 1, \ldots, N \} ,
\end{align*}
with the convention that $x_0=-\infty$ and $x_{2N+1}=\infty$. 
For $s \in \{1,2,\ldots,N\}$, let $\eta_s$ be the level line of $\gff$ starting from $x_{2s-1}$, 
considered as an oriented curve. 
Similarly as in the case of the Ising model, the endpoints of the level lines $(\eta_1, \ldots, \eta_N)$ 
give rise to a planar pair partition, which we encode in a link pattern $\LA = \LA(\eta_1, \ldots, \eta_N) \in \LP_N$.

\begin{theorem} \label{thm::multiple_sle_4}
\textnormal{\cite[Theorem~1.4]{PeltolaWuGlobalMultipleSLEs}}
Consider the level lines of the $\GFF$ in $(\HH; x_1, \ldots, x_{2N})$ with alternating boundary data. 
Denote by $\LA$ the random connectivity pattern in $\LP_N$ formed by the $N$ level lines. Then, we have
\begin{align}\label{eq::crossing_probabilities_for_kappa4}
\PP[\LA=\alpha] = \frac{\PartF_{\alpha}  (x_1, \ldots, x_{2N}) }
{\PartF^{(N)}_{\GFF}  (x_1, \ldots, x_{2N} ) }, 
\quad \text{ for all } 
\alpha\in \LP_N, \quad\text{where } \PartF^{(N)}_{\GFF} := \sum_{\alpha\in\LP_N} \PartF_{\alpha},
\end{align}
and $\{\PartF_{\alpha} \colon \alpha \in \LP_N\}$ 
is the collection of functions uniquely determined as the solution to the PDE boundary value problem
given in Definition~\ref{defn:PPF}, 
also known as pure partition functions of multiple $\SLE_\kappa$ with $\kappa = 4$. 
\end{theorem}

Moreover, for $a,b\in\{1,2,\ldots,2N\}$, where $a$ is odd and $b$ is even, the probability that the level line 
of the $\GFF$ starting from $x_a$ terminates at $x_b$ is given by~\cite[Proposition~5.6]{PeltolaWuGlobalMultipleSLEs}:
\begin{align}\label{eqn::levellines_proba_lk}
P^{(a,b)} (x_1, \ldots, x_{2N})
= \prod_{\substack{1\le i\le 2N, \\ i\neq a,b}} 
\Big| \frac{x_i - x_a}{x_i - x_b} \Big|^{(-1)^i} .
\end{align}

R.~Kenyon and D.~Wilson have found
explicit formulas for crossing probabilities in the double-dimer model~\cite{Kenyon-Wilson:Boundary_partitions_in_trees_and_dimers}.
These formulas are combinatorial expressions involving the inverse Kasteleyn matrix.
On the other hand, explicit formulas of similar type were obtained 
in~\cite{Kenyon-Wilson:Boundary_partitions_in_trees_and_dimers, PeltolaWuGlobalMultipleSLEs} 
for the connection probabilities of the GFF level lines appearing in~\eqref{eq::crossing_probabilities_for_kappa4}, 
where the inverse Kasteleyn matrix gets replaced by the boundary Poisson 
kernel.
Using Kenyon's results~\cite{Kenyon_domino_tiling}, it should be possible to explicitly check that in suitable approximations,
the double-dimer crossing probabilities converge in the scaling limit to~\eqref{eq::crossing_probabilities_for_kappa4}.
Note, however, that the convergence of double-dimer interfaces to the $\SLE_4$ still remains conjectural.
\newpage
\section{Connection Probabilities for Loop-Erased Random Walks}
\label{app::appendix_lerw}
R.~Kenyon and D.~Wilson also found in~\cite{Kenyon-Wilson:Boundary_partitions_in_trees_and_dimers}
determinantal formulas (analogous to~Fomin~\cite{Fomin-LERW_and_total_positivity}) 
for connectivity probabilities for multichordal loop-erased random walks (LERW). 
It follows from these formulas and the discrete complex analysis developed in~\cite{CFL-uber_die_PDE_der_mathphys,CS-discrete_complex_analysis_on_isoradial} 
that the multichordal LERW connectivity probabilities converge (when suitably renormalized) to the pure partition functions 
of multiple $\SLE_\kappa$ with $\kappa = 2$ --- for a proof, see, e.g.,~\cite[Theorems~3.16 and~4.1]{KarrilaKytolaPeltolaCorrelationsLERWUST} and~\cite[Theorem~2.2]{KarrilaUSTBranches}.

In these results, the multichordal LERWs are realized as boundary touching branches in a uniform spanning tree (UST) with wired boundary conditions\footnote{A~\textit{spanning tree} of a finite connected graph $\graph$ is a subgraph that is connected, has no cycles, and contains every vertex of $\graph$.
A~\textit{uniform} spanning tree in $\graph$ is a spanning tree chosen uniformly at random amongst all spanning trees in $\graph$.
A~uniform spanning tree in a finite planar graph $\graph$ with \textit{wired boundary conditions} is obtained by considering a~uniform spanning tree in the quotient graph obtained from $\graph$ by collapsing the outer boundary of $\graph$ into a single vertex. See, e.g.,~\cite{LawlerSchrammWernerLERWUST, KarrilaKytolaPeltolaCorrelationsLERWUST} for more details.}. 
Such curves converge in the scaling limit to multiple $\SLE_\kappa$ curves with 
$\kappa = 2$~\cite{SchrammScalinglimitsLERWUST, LawlerSchrammWernerLERWUST, Zhan-scaling_limits_of_planar_LERW, Karrila_multiple_SLE}. 
The reader can find the precise definitions of these objects, as well as the detailed setup for the scaling limit results, 
in~\cite[Section~3]{KarrilaKytolaPeltolaCorrelationsLERWUST} --- see also the recent~\cite{KarrilaUSTBranches}.
To give a satisfactory statement, we translate the notations used there into the notations used in the present article.
Recall from Section~\ref{sub::Ising_notation} that, for a finite subgraph $\graph = (V(\graph), E(\graph))$ of $\Z^2$, 
we denote the inner boundary of $\graph$ by 
$\partial \graph = \{ v \in V(\graph) \, \colon \, \exists \, w \not\in V(\graph) \text{ such that } \edge{v}{w} \in E(\Z^2)\}$.
We also define the outer boundary (called boundary in~\cite{KarrilaKytolaPeltolaCorrelationsLERWUST})
of $\graph$ as the vertices $w \not\in V(\graph)$ for which there exists a vertex $v \in V(\graph)$ such that $\edge{v}{w}\in E(\Z^2)$.
Then, as in~\cite[Section~3]{KarrilaKytolaPeltolaCorrelationsLERWUST},
we call the edges $e = \edge{v}{w}$ \textit{boundary edges} of $\graph$, and we denote $v = e^\circ$ and $w = e^\partial$.

Now, for a discrete polygon $(\graph;  v_1,  \ldots,  v_{2N})$, we may consider those branches in the UST that start from the vertices $v_1,  \ldots,  v_{2N}$.
Importantly, these discrete curves may form other topological configurations than those labeled by the link patterns
--- the curves can merge in various ways. 
However, for each link pattern
$\alpha =  \{ \link{a_1}{b_1} , \ldots, \link{a_N}{b_N}\} \in \LP_N$, 
we can consider the probability that the discrete curves form the connectivity pattern $\alpha$ in the following sense.
We group the marked vertices $v_1,  \ldots,  v_{2N}$ into two groups
$\{v_{2s-1} \colon s = 1,2, \ldots, N\}$ and $\{v_{2s} \colon s = 1,2, \ldots, N\}$. 
We form a modified discrete polygon by replacing the latter group by $\{w_{2s} \colon s = 1,2, \ldots, N\}$,
where $v_{2s}$ and $w_{2s}$ form a boundary edge $e_{2s} = \edge{v_{2s}}{w_{2s}}$, so that $v_{2s} = e_{2s}^\circ$ and $w_{2s} = e_{2s}^\partial$.
Then, we may consider the event that there exist $N$ branches in the UST connecting $v_{a_s}$ to $w_{b_s}$, for all $s \in \{1,2,\ldots,N\}$, with the convention that in $\alpha$, the indices $a_1, \ldots, a_N$ are odd and the indices $b_1, \ldots, b_N$ are even. 
Adapting the notations in~\cite{KarrilaKytolaPeltolaCorrelationsLERWUST}, for each $s \in \{1,2,\ldots,N\}$, we denote the corresponding event by
\begin{align*}
\{ \textnormal{ $\exists$ branch in the UST connecting $v_{a_s}$ to $w_{b_s}$ } \} = \{\pathfromto{v_{a_s}}{w_{b_s}}\} 
\end{align*}
so that the desired connectivity event is 
\begin{align*}
\{\textnormal{ connectivity of $v_1,  \ldots,  v_{2N}$ is $\alpha$ } \} = \underset{s=1}{\overset{N}{\bigcap}} \{\pathfromto{v_{a_s}}{w_{b_s}}\} .
\end{align*}

\begin{theorem} \label{thm::LERW_crossingproba}
\textnormal{\cite[Theorems~3.16 and~4.1]{KarrilaKytolaPeltolaCorrelationsLERWUST}}
Let discrete polygons
$(\Omega^{\delta}; x_1^{\delta}, \ldots, x_{2N}^{\delta})$ on $\delta\Z^2$ 
approximate a regular enough\footnote{The boundaries of $\Omega$ and $\Omega^\delta$ are locally horizontal or vertical line segments near the marked boundary points.} 
polygon
$(\Omega; x_1, \ldots, x_{2N})$ as $\delta\to 0$, as detailed in~\textnormal{\cite[Section~3]{KarrilaKytolaPeltolaCorrelationsLERWUST}}.
Consider the uniform spanning tree in $\Omega^{\delta}$ with wired boundary conditions. 
Then, we have
\begin{align}\label{eqn::lerw_crossing_proba}
\lim_{\delta\to 0} \delta^{-2N} 
\PP^{\delta}[\textnormal{ connectivity of $x_1^{\delta}, \ldots, x_{2N}^{\delta}$ is $\alpha$ }] =
\pi^{-N} \PartF_{\alpha}(\Omega;x_{1},\ldots,x_{2N}), \quad \text{for all }\alpha\in\LP_N,
\end{align}
and $\{\PartF_{\alpha} \colon \alpha \in \LP_N\}$ 
is the collection of functions uniquely determined as the solution to the PDE boundary value problem
given in Definition~\ref{defn:PPF}, 
also known as pure partition functions of multiple $\SLE_\kappa$ with $\kappa = 2$. 
\end{theorem}

We remark the normalization factor $\delta^{-2N}$ in the connection probability~\eqref{eqn::lerw_crossing_proba},
absent from Theorem~\ref{thm::ising_crossingproba} for the Ising case, as well as the absence of the symmetric partition function 
$\PartF_{\LERW} := \sum_\alpha \PartF_{\alpha}$. In fact, a formula for 
$\PartF_{\LERW}$ is 
known~\cite[Lemma~4.12]{PeltolaWuGlobalMultipleSLEs}, and one could also consider the UST \textit{conditioned on the event} that
the branches starting from the vertices $x_1^{\delta}, \ldots, x_{2N}^{\delta}$ connect according to some (random) link pattern $\conn^{\delta}$ that belongs to $\LP_N$.
This conditioning accounts to dividing by 
$\PartF_{\LERW}$, and with the conditioned UST probability measure $\tilde{\PP}^{\delta}_{\textrm{LERW}}$, we have
(see~\cite[Theorem~2.2]{KarrilaUSTBranches}) 
\begin{align}\label{eqn::lerw_crossing_proba2}
\lim_{\delta\to 0} \tilde{\PP}^{\delta}_{\textrm{LERW}} [\conn^{\delta}=\alpha] 
= \frac{\PartF_{\alpha}(\Omega;x_{1},\ldots,x_{2N})}{\PartF^{(N)}_{\LERW}(\Omega;x_{1},\ldots,x_{2N})}, \quad \text{for all }\alpha\in\LP_N .
\end{align}
Here, the powers $\delta^{-2N}$ are cancelled by the conditioning, resulting in the normalization factor 
$\PartF_{\LERW}$ on the right-hand side. 
Instead of Theorem~\ref{thm::LERW_crossingproba} which crucially relies on exact solvability of the crossing probabilities in the discrete model in terms of Fomin's formulas~\cite{Fomin-LERW_and_total_positivity},  
similar ideas as in the case of the Ising model could be used here. For such an approach, the main inputs would be the following:
\begin{enumerate}
\item \label{1LERW} local convergence of the branches to multiple $\SLE_2$ curves (proven in~\cite[Theorem~2.1]{KarrilaUSTBranches});
\item \label{2LERW} continuity of the limiting 
curve up to and including the swallowing time of the spectator points  (proven in~\cite[Theorem~6.8]{Karrila_multiple_SLE});~and
\item \label{3LERW} the fact that the Loewner chain associated to 
the partition function $\PartF_{\LERW}$  is continuous up to and including the same stopping time.
Item~\ref{3LERW} could be shown similarly as Theorem~\ref{thm::loewner_Ztotal_continuity}, 
provided that one first proves analogues of Propositions~\ref{prop::Z_total_asy_refined} and~\ref{prop::Z_total_ineq} for $\kappa = 2$. 
\end{enumerate}

\bigskip

{\small
\newcommand{\etalchar}[1]{$^{#1}$}

\end{document}